%% file: main.tex
\documentclass[sigconf]{acmart}

\usepackage{booktabs} 
\usepackage{enumerate}
\usepackage{color, soul}
\usepackage{framed}
\usepackage{paralist}
\usepackage[shortlabels]{enumitem}
\usepackage{subfigure}
\usepackage[linesnumbered,lined,boxed]{algorithm2e}
\usepackage[toc, page]{appendix}
\usepackage{natbib}
\usepackage{bigints}
\usepackage{hyperref}
\usepackage{mathtools}
\usepackage[thinc]{esdiff}
\usepackage{courier}
\usepackage{amsthm}
\usepackage{adjustbox,multirow}
\usepackage{bbm, dsfont}
\usepackage{tikz}
\usepackage{amsmath}

\copyrightyear{2022}
\acmYear{2022}

\setcopyright{rightsretained}
\acmConference[CCS '22]{Proceedings of the 2022 ACM SIGSAC
Conference on Computer and Communications Security}{November
7--11, 2022}{Los Angeles, CA, USA}
\acmBooktitle{Proceedings of the 2022 ACM SIGSAC Conference on
Computer and Communications Security (CCS '22), November 7--11,
2022, Los Angeles, CA, USA}
\acmDOI{10.1145/3548606.3560636}
\acmISBN{978-1-4503-9450-5/22/11}

\usepackage{etoolbox}
\makeatletter
\patchcmd{\maketitle}{\@copyrightpermission}{
   \begin{minipage}{0.3\columnwidth}
     \href{https://creativecommons.org/licenses/by/4.0/}{\includegraphics[width=0.90\textwidth]{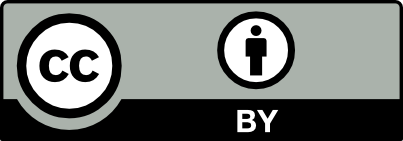}}
   \end{minipage}\hfill
   \begin{minipage}{0.7\columnwidth}
     \href{https://creativecommons.org/licenses/by/4.0/}{This work is licensed under a Creative Commons Attribution International 4.0 License.}
   \end{minipage}

   \vspace{5pt}
}{}{}

\makeatother

\begin{document}
\title{L-SRR: Local Differential Privacy for Location-Based Services with Staircase Randomized Response}

\author{Han Wang}
\email{hwang185@hawk.iit.edu}
\affiliation{%
  \institution{Illinois Institute of Technology}
      \country{USA}
}

\author{Hanbin Hong}
\email{hhong4@hawk.iit.edu}
\affiliation{%
  \institution{Illinois Institute of Technology}
    \country{USA}
}

\author{Li Xiong}
\email{lxiong@emory.edu}
\affiliation{%
  \institution{Emory University}
    \country{USA}
}

\author{Zhan Qin}
\email{qinzhan@zju.edu.cn}
\affiliation{%
  \institution{Zhejiang University}
    \country{China}
}

\author{Yuan Hong}
\email{yuan.hong@iit.edu}
\affiliation{%
  \institution{Illinois Institute of Technology}
  \institution{University of Connecticut}
   \country{USA}
}

\begin{abstract}
Location-based services (LBS) have been significantly developed and widely deployed in mobile devices. It is also well-known that LBS applications may result in severe privacy concerns by collecting sensitive locations. A strong privacy model ``local differential privacy'' (LDP) has been recently deployed in many different applications (e.g., Google RAPPOR, iOS, and Microsoft Telemetry) but not effective for LBS applications due to the low utility of existing LDP mechanisms. To address such deficiency, we propose the first LDP framework for a variety of location-based services (namely ``\texttt{L-SRR}''), which privately collects and analyzes user locations with high utility. Specifically, we design a novel randomization mechanism ``Staircase Randomized Response'' (\texttt{SRR}) and extend the empirical estimation to significantly boost the utility for \texttt{SRR} in different LBS applications (e.g., traffic density estimation, and k-nearest neighbors). We have conducted extensive experiments on four real LBS datasets by benchmarking with other LDP schemes in practical applications. The experimental results demonstrate that \texttt{L-SRR} significantly outperforms them.
\end{abstract}




\maketitle
\section{Introduction}
\label{sec:intro}
\input{intro}

\section{Preliminaries}
\label{sec:model}
\input{model}

\section{L-SRR for Location-Input LBS}
\label{sec:algm}
\input{algm}

\section{L-SRR for Trajectory-Input LBS}
\label{sec:advlbs}
\input{advlbs}

\section{Discussion}
\label{sec:disc}
\input{discuss}

% \vspace{-5mm}
\section{Experiments}
\label{sec:exp}
\input{exp}

\section{Related Work}
\label{sec:related}
\input{related}

\section{Conclusion}
\label{sec:concl}
\input{concl}

\begin{acks}
This work is partially supported by the National Science Foundation (NSF) under the Grants No. CNS-2046335, CNS-2034870, CNS-2125530, CNS-1952192 and CNS-2027783, as well as the Cisco Research Award. In addition, some of the results presented in this paper were obtained using the Chameleon testbed supported by the NSF. Finally, the authors would like to thank the anonymous reviewers for their constructive comments.
\end{acks}

\bibliographystyle{ACM-Reference-Format}

%\clearpage

\appendix
\input{appendix}

\end{document}

%% file: intro.tex
Location-based services (LBS) are widely deployed in mobile devices to provide useful and timely location-based information to users. For instance, WeatherBug provides weather information based on users' regions; Google Map not only navigates the routes with real-time traffic conditions but also responds to queries such as nearby restaurants or gas stations; Waze is similar to Google Map but actively collects extra information (e.g., accidents, road construction, and police) from users and shares them to other users.

All of these LBS applications highly rely on the personal locations collected from millions of users. Such locations should be protected, e.g., per the General Data Protection Regulation (GDPR) since visited places can be sensitive (e.g., hospital) or used to re-identify users from the data (e.g., a sequence of them can be unique). To mitigate such risks, location anonymization models \cite{cloaking1} were first proposed to achieve $k$-anonymity via location generalization. However, $k$-anonymity can only provide a weak privacy guarantee (e.g., vulnerable to the background knowledge attacks \cite{attack}). As a rigorous privacy model against arbitrary prior knowledge known to the adversaries, differential privacy (DP) has been extensively studied to address location privacy risks (e.g., \cite{liulingtra18}). It ensures that adding or removing any user's location or trajectory still generates indistinguishable results. For instance, AdaTrace \cite{liulingtra18}, a differentially private location trace synthesizer was proposed to ensure provable privacy, deterministic attack resilience, and strong utility. However, in the DP scenario setting \cite{SearchLog, R2DP, VideoDP}, it requires an authorized data center to collect user's location. Unfortunately, in the 2011 Microsoft survey, 87$\%$ of participants reported that they care about who accesses their location information; over 78$\%$ workers of Amazon interviewed in 2014 still do not trust these LBS applications on collecting their locations and believed apps accessing to their locations can pose significant privacy threats \cite{LocationSurvey}. 
Thus, it is highly desirable to explore private location collection by an \emph{untrusted} server.

Recently, \emph{local differential privacy} (LDP) techniques \cite{rappor14,ldpusenix17, CormodeLDP18,KUI2021, VideoLDP} have been successfully deployed in industry (e.g., Google \cite{rappor14}, Apple \cite{Apple}, and Microsoft \cite{boling17}) to privately aggregate locally perturbed data. It provides stronger privacy against attackers with arbitrary background knowledge (not only the downstream analysts but also the data aggregator can be \emph{untrusted}). To date, existing LDP schemes such as RAPPOR and generalized randomized response have been extended to privately aggregate different types of data, e.g., set-valued data \cite{boling17}, numerical data \cite{NumLDP}, video \cite{VideoLDP}, and graphs \cite{zhan17}. However, existing LDP schemes are not very effective on private location data collection and analysis due to either limited utility or relaxed privacy protection. To our best knowledge, only \cite{loc_LDP3,loc_LDP2} applied existing LDP schemes to locations but the utility is still poor. Moreover, \texttt{PLDP} \cite{location_LDP} relaxed LDP to personalized LDP (not every user can be protected with $\epsilon$-LDP) in the location collection for spatial density estimation.

Furthermore, some other privacy-enhancing techniques \cite{Geo-indistinguishability,location_LDP} privately collect locations for LBS that provides services to individual users (e.g., GPS navigation \cite{Geolife}, and nearest point-of-interest (POI) search \cite{POIDP}) \emph{without a trusted server}. For instance, geo-indistinguishability \cite{Geo-indistinguishability} adds Laplace noise to the user's location for ensuring privacy in LBS. However, it cannot strictly satisfy LDP (the locations are indistinguishable only within a radius), and the Laplace mechanism has been shown to be worse than randomized response for local perturbation \cite{RRLDP}. 

To address such limitations, we propose the first strict LDP framework (namely, ``\texttt{L-SRR}'') to support a variety of LBS applications. First, we design a novel LDP mechanism ``\emph{staircase randomized response} (\texttt{SRR})'' and revise the empirical estimation to privately aggregate locations with significantly improved utility and strictly satisfied $\epsilon$-LDP. Second, different from all existing works \cite{loc_LDP3,loc_LDP2,Geo-indistinguishability,location_LDP}, we design additional components (e.g., private matching \cite{POIDP}, and private information retrieval \cite{PIR}) into \texttt{L-SRR} to ensure $\epsilon$-LDP for a variety of LBS applications such as $k$ nearest neighbors search \cite{KNN}, origin-destination analysis \cite{OD}, and traffic-aware GPS navigation \cite{Geolife}, which may collect user trajectories or perform individual services with the aggregated locations/trajectories. 

\begin{figure*}[!tbh]
	\centering		\includegraphics[angle=0, width=0.9\linewidth]{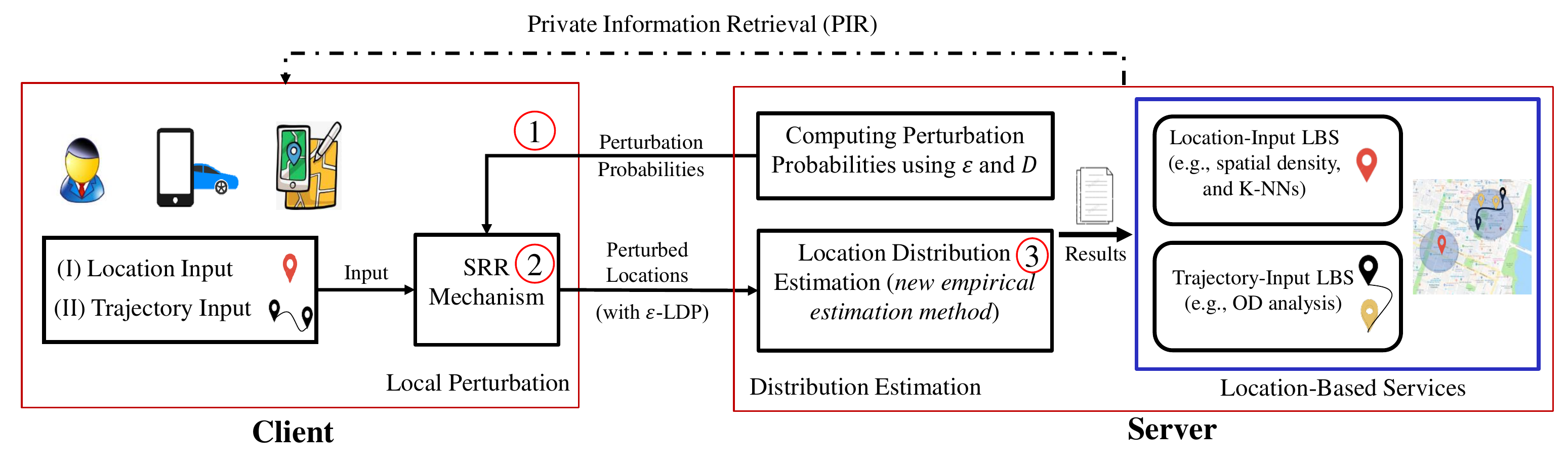}
		\vspace{-0.2in}
	\caption[Optional caption for list of figures]
	{The \texttt{L-SRR} framework}\vspace{-0.1in}
	\label{fig:overview}
\end{figure*}

The utility of \texttt{L-SRR} is significantly enhanced by the proposed new \texttt{SRR} mechanism and estimation method. Specifically, \texttt{SRR} perturbs input locations with \emph{staircase probabilities} for different possible output locations. The probability of perturbing any input $x$ in the domain $\mathcal{D}$ to each possible location $y\in\mathcal{D}$ is optimally pre-computed. Then, users can locally perturb their locations with the optimal probabilities. Different from relaxed privacy notions (e.g., \texttt{PLDP} and geo-indistinguishability), every user is still strictly protected by $\epsilon$-LDP. At the server end (data aggregator), we extend an empirical estimation \cite{emp} to further improve the utility for the \texttt{SRR} mechanism without extra privacy leakage \cite{Dwork14}. Thus, the major contributions of this paper are summarized as below:

\begin{itemize}
\setlength\itemsep{0.5em}

\item To our best knowledge, we design the first LDP mechanism (\texttt{SRR}) to make the strong privacy notion LDP practical (with high utility) for many LBS applications. 

\item In \texttt{SRR}, we propose a novel hierarchical encoding scheme and relevant algorithms to derive the optimal perturbation probabilities independent of the input data. We also extend the empirical estimation method to further improve utility. 

\item We design and integrate components in \texttt{L-SRR} to realize \texttt{SRR} in a series of LBS applications with high accuracy, which may collect locations (e.g., frequency estimation \cite{traden}) or trajectories (e.g., origin-destination analysis \cite{OD}, and traffic-aware GPS navigation \cite{Geolife}). 

\item Besides theoretical studies on the privacy and utility, we conduct extensive experiments on four real LBS datasets, and benchmark with other LDP schemes, e.g., Generalized Randomized Response (\texttt{GRR})\cite{ldpusenix17}, Local Hash (\texttt{OLH-H})\cite{ldpusenix17}, \texttt{PLDP} (based on Unary Encoding)\cite{location_LDP}, and Hadamard Response (\texttt{HR}) \cite{emp}. \texttt{L-SRR} greatly outperforms them in almost all the scenarios.
\end{itemize}

The remainder of this paper is organized as follows. Section \ref{sec:model} introduces some preliminaries. Section \ref{sec:algm} illustrates the \texttt{SRR} mechanism, and Section \ref{sec:advlbs} extends \texttt{SRR} to collect trajectories. Section \ref{sec:disc} gives related discussions. Section \ref{sec:exp} shows the experimental results. Section \ref{sec:related} and \ref{sec:concl} discuss the literature and conclude the paper.

%% file: model.tex
\subsection{LBS Applications}
\label{app}

We first categorize two different types of LBS applications \footnote{The discrete location domain is considered in these applications.}:

\vspace{0.05in}
    \noindent \textbf{Location-Input LBS}: The locations from users are collected by the LBS Apps, and the untrusted server privately analyzes the aggregated data, e.g., identifying the top crowded areas \cite{croden}, and spatial density estimation \cite{location_LDP}. In some LBS applications, the clients may query the analysis results from the server (e.g., location-based advertising \cite{LBSAdv}, and $k$ nearest point of interests (POIs) for each user \cite{KNN}).
 
\vspace{0.05in}
    
    \noindent \textbf{Trajectory-Input LBS}: LBS App collects multiple sequential locations (trajectory) from each user \cite{trajectoryTDSC}, and the untrusted server privately analyzes the aggregated data, e.g., aggregating users' origin-destination (OD) pairs to learn the traffic flow \cite{OD,route}. Similarly, users may query the analysis results computed by the server, e.g., users query the real-time traffic for the GPS navigation \cite{route}.
    
\subsection{Privacy Model}

Users in \texttt{L-SRR} will locally randomize their location(s) \cite{histogram1} with algorithm $\mathcal{A}$ and send the noisy results to the untrusted server. After local perturbation, 
all the input locations can be indistinguishable \cite{rappor14}.
The privacy notion is formally defined as below:
\begin{definition}[$\epsilon$-LDP] A randomization algorithm $\mathcal{A}$ satisfies $\epsilon$-Local Differential Privacy, if and only if for any pair of input locations $x, x'\in \mathcal{D}$, and for any perturbed output $y \in range(\mathcal{A})$ sent to the untrusted server, we have: $Pr[\mathcal{A}(x)=y]\leq e^{\epsilon}\cdot Pr[\mathcal{A}(x')=y]$.
\label{def:SLDP}
\end{definition}

After each user locally perturbs its data, LDP can be ensured for all the input locations \cite{CormodeLDP18,rappor14,ldpusenix17}, where the privacy bound $\epsilon$ reflects the degree of indistinguishability. The untrusted server will aggregate and analyze the noisy data with estimation methods.

\subsection{L-SRR Framework}
As shown in Figure \ref{fig:overview}, we design three major components in \texttt{L-SRR}: perturbation (by client), analysis (by server), and private retrieval (by both client and server only when the user needs to privately query the analysis results, e.g., traffic-aware GPS navigation):

\begin{enumerate}

\setlength\itemsep{0.5em}

\item \textbf{Perturbation (client)}: Each user's location data (location or trajectory) is locally perturbed by the client with $\epsilon$-LDP. \texttt{SRR} optimizes the utility after hierarchically encoding the location domain $\mathcal{D}$. Encoding and optimal perturbation probabilities are pre-computed by the server (only based on $\epsilon$ and $\mathcal{D}$) to ensure $\epsilon$-LDP. See details in Section \ref{sec:Encoding}. 

\item \textbf{Analysis (server)}: Before the perturbation, the server shares the pre-computed perturbation probabilities with all the clients. After receiving the perturbed user locations, the  server estimates the \emph{location distribution} with a revised empirical estimation method. Then, the server loads such results into specific LBS (along with the required components) to privately derive the analysis result. See details in Section \ref{sec:algm}.

\item \textbf{Private Retrieval (only for LBS with client queries)}: It is an optional component of \texttt{L-SRR}. If requested in specific LBS (with client queries), \texttt{L-SRR} first provides the server-side LBS analysis (e.g., estimating the overall traffic density) with LDP guarantees. At the client end, each user privately queries his/her result (e.g., nearby traffic) from the analysis results at the server side. This can be achieved with a private information retrieval (PIR) protocol \cite{PIR2}. With the PIR for client queries, server does not know which result is delivered to which user, and each user does not know other users' results either. \footnote{If we directly design a cryptographic protocol for each LBS, it involves location data encryption by the client, and the server should extend each LBS algorithm over encrypted data to a cryptographic protocol, which would result in extremely high computation and communication overheads. Compared to that, PIR establishes a secure channel for privately retrieving the results, which can be independent of the LBS algorithms and extensible to all the analyses on noisy data by the server.} 
\end{enumerate}

\noindent\textbf{User Requirements}. \texttt{L-SRR} can be deployed as a privacy preserving API in each LBS App. Users only need to periodically update the privacy bound $\epsilon$ with the server. In each LBS, users only need to locally perturb their location(s) with the pre-computed perturbation probabilities, and send the perturbed result to the server. The integrated PIR \cite{PIR} also requires very minor computation and communication overheads without affecting the LDP guarantee (see the discussion in Section \ref{sec:disc}). 

\vspace{0.05in}

\noindent\textbf{LDP Protection}. Similar to existing LDP models \cite{rappor14,ldpusenix17}, \texttt{L-SRR} ensures strong privacy against inferences on users' local data based on arbitrary background knowledge, which is orthogonal to mitigating other types of risks  (e.g., encryption \cite{EncLBS} and defenses against side-channel attacks \cite{side_channel}). Thus, \texttt{L-SRR} can be integrated with them to further improve security and privacy if necessary.

%% file: algm.tex
In this section, we design the \texttt{SRR} mechanism to privately collect a location from each user for analysis (standard LDP setting \cite{ldpusenix17,location_LDP}).

\subsection{Staircase Randomized Response}

We first review a family of LDP mechanisms. Randomized Response (\texttt{RR}) based schemes, such as generalized randomized response (\texttt{GRR}) \cite{GRR} and unary encoding (\texttt{UE}) \cite{ldpusenix17}, satisfy $\epsilon$-LDP. For instance, in \texttt{GRR}, given the domain size $d = |\mathcal{D}|$, privacy bound $\epsilon$, and input $x \in \mathcal{D}$, the true value has a higher probability to be sampled (output $y$). The following perturbation probabilities $q(y|x)$ ensure $\epsilon$-LDP.

\vspace{-0.05in}

%\small
\begin{equation}
\begin{gathered}
\text{\texttt{GRR}}:~q(y|x)={\centering
\begin{cases}
\frac{e^\epsilon}{d+e^\epsilon-1},~~~\text{if $y=x$}\\
\frac{1}{d+e^\epsilon-1}, ~~~\text{otherwise}\\
\end{cases}}
\end{gathered}
\label{eq:rr}
\end{equation}

%\normalsize

%\vspace{0.05in}
Also, Hadamard Response (\texttt{HR}) \cite{emp} has a subset domain for each value $x$ and a higher probability for values in the subset to be sampled. Then, the remaining values in the domain are sampled with a smaller probability. However, \emph{only two different perturbation probabilities} are defined in the existing LDP mechanisms (e.g., \texttt{GRR} \cite{GRR}, \texttt{UE} \cite{ldpusenix17}, and \texttt{HR} \cite{emp}), not sufficiently fine-grained to optimize the utility (since the perturbation probabilities simply treat all the other output locations in the domain equally). 

Thus, we propose a novel Staircase Randomized Response (\texttt{SRR}) mechanism for locations and LBS. Intuitively, if the probabilities for locations that are closer to the input location $x$ can be higher, it is more possible for users that the query results of the LBS are the same. To this end, \texttt{SRR} will first consider the location distances to the input location $x$. Then, a set of fine-grained probabilities should be pre-computed for all the possible output locations $y \in \mathcal{D}$. 

When pre-computing these probabilities, there are several issues in practice. For instance, for each input location $x$, if we compute the probability $q(y|x)$ for each possible output $y \in \mathcal{D}$, the number of probabilities is the domain size $d$. Then, $\forall x \in \mathcal{D}$, there are $d$ probabilities for each location $x$ and $d\times d$ different probabilities for all the locations in the domain. Thus, there are $d^2$ unknown probabilities to be determined, which makes it time-consuming to derive the optimal probabilities \cite{InputDis2020} and not extensible if the domain is updated. Second, general objective function (e.g., the variance) to optimize the perturbation probabilities is dependent on the unknown true frequencies. To address this, output locations can be partitioned into different groups in terms of their distances to $x$ (\emph{the probabilities of all the output locations in the same group could be identical}), and we can derive the perturbation probability for each group to boost the utility while globally satisfying $\epsilon$-LDP. 

\vspace{-0.1in}

\begin{figure}[!h]
	\centering
	\subfigure[\texttt{GRR} mechanism \cite{GRR}]{
		\includegraphics[angle=0, width=0.49\linewidth]{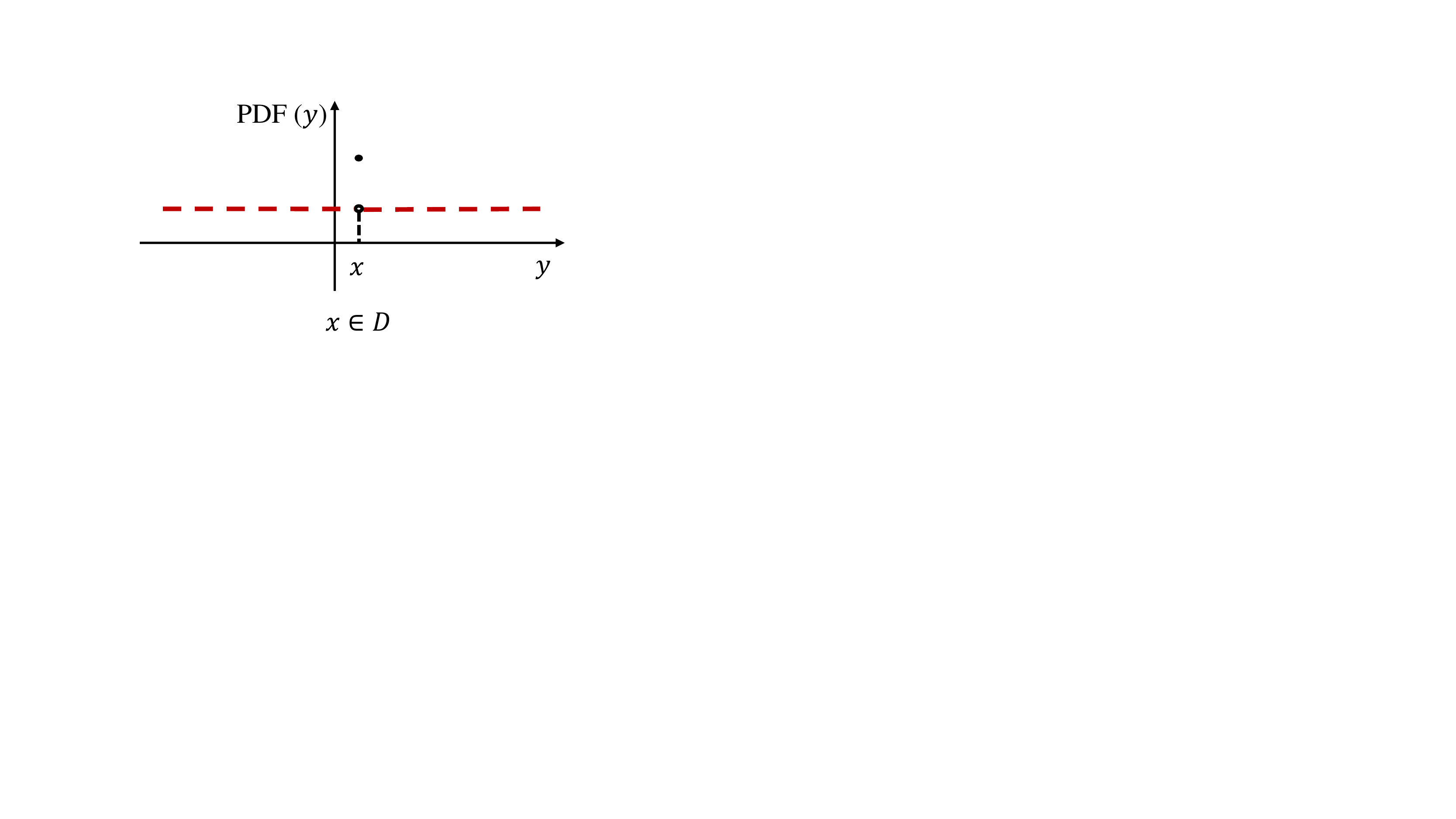}
		\label{fig:RR} }
	\hspace{-0.12in}	
	\subfigure[\texttt{SRR} mechanism (for \texttt{L-SRR})]{
		\includegraphics[angle=0, width=0.49\linewidth]{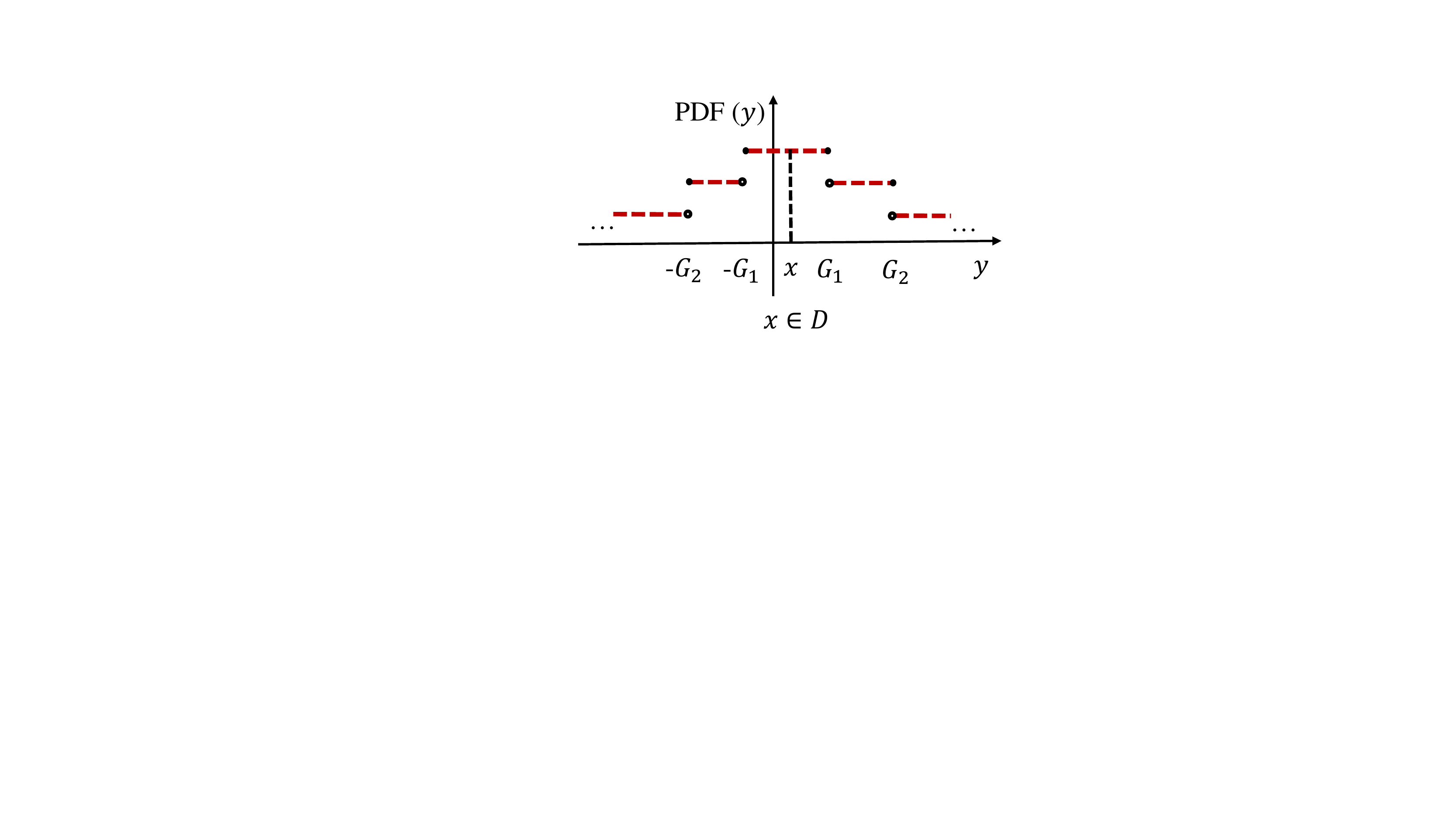}
		\label{fig:L-SRR}}\vspace{-0.15in}
	\caption{Probability density function (PDF) for \texttt{GRR} and \texttt{SRR}}\vspace{-0.1in}
	\label{fig:Mechanism}
\end{figure}

The probability density functions (PDFs) of \texttt{GRR} (w.l.o.g.) and \texttt{SRR} are illustrated in Figure \ref{fig:Mechanism}. It is worth noting that the Figure \ref{fig:Mechanism} is the 1-D representation of the 2-D discrete locations in the domain. In \texttt{GRR}, the probability that outputs the true value (the point in Figure \ref{fig:RR}) is higher than other values. On the contrary, since \texttt{SRR} discretizes the perturbation probabilities for all the grouped possible output locations, the PDF of \texttt{SRR} has a similar shape to the staircase mechanism in differential privacy \cite{staircase}, which also has a staircase PDF for different groups to satisfy $\epsilon$-DP. Motivated by that, we name our new randomization mechanism as the ``Staircase Randomized Response'' (\texttt{SRR}) in local differential privacy. We formally define the perturbation probabilities from input $x$ to all the output locations as follows. 

Given the domain $\mathcal{D}$, for any input $x\in \mathcal{D}$, all the possible output locations can be partitioned into $m$ groups $G_1(x),..., G_m(x)$ based on their distances to $x$.\footnote{W.l.o.g., the distances from $x$ to locations in $G_j(x)$ are farther if $j$ is larger. The closest group is $G_1(x)$ whereas the farthest group is $G_m(x)$.} Notice that, the partitioning $G_j(x)$ is dependent on the input location $x$. For each input location $x$, all its $m$ location groups and the perturbation probabilities (for perturbing $x$ to any output location $y$) will be efficiently computed as: 

\vspace{-0.05in}

\begin{equation}
\begin{gathered}
\text{\texttt{SRR}}: \forall x\in\mathcal{D}, q(y|x)={\centering
\begin{cases}
\alpha_1(x),~~~\text{if $y \in G_1(x)$}\\
\quad\vdots\quad\quad\vdots\quad\quad\vdots\\
\alpha_m(x), ~~~\text{if $y \in G_m(x)$}\\
\end{cases}}
\end{gathered}
\label{eq:SRR}
\end{equation}

%\normalsize

where $\alpha_1(x),..., \alpha_m(x)$ are the distance-based perturbation probabilities for locations in $m$ different groups perturbed from $x\in\mathcal{D}$, and the gap between the perturbation probabilities in every adjacent groups is the same (``Staircase PDF'') in $\alpha_1(x),..., \alpha_m(x)$. 

Also, the sum of all the perturbation probabilities for each input location $x$ should satisfy: $\sum_{j\in [1,m]}\sum_{y\in G_j(x)} q(y|x)=1$. The details for computing the probabilities will be given in Section \ref{sec:prob}. \texttt{SRR} generates more accurate locally perturbed locations than the state-of-the-art LDP mechanisms with only two perturbation probabilities (e.g., \texttt{GRR} \cite{GRR} and \texttt{HR} \cite{emp}), as validated in Section \ref{sec:exp}.

\subsection{Data Encoding and Domain Partitioning}
\label{sec:Encoding}

\noindent\textbf{Hierarchical Location Encoding}. 
To encode the location data, we use a hierarchical encoding scheme based on the Bing Map Tiles System \cite{Map}, which recursively partitions geo-coordinates into 4 blocks, and indexes all the locations to reach the desired resolution \cite{liulingtra18}. Then, the locations are encoded into bit strings by hierarchically concatenating the indices of all the levels for every specific location. Figure \ref{fig:loc_encoding} illustrates an example for the encoding. Specifically, starting from the root node, at each level $h$, the $4$ children of each node (four sub-blocks) can be encoded by $00, 01, 10, 11$ (2-bit), and thus form $4^h$ blocks for indexing locations. Then, we can derive the encoded bit string by concatenating the bits from the first level to the leaf node level. For all the locations on the earth, $h$ can be as large as 23 (46 bits for a location) to index each 4.7m$\times $4.7m region. As a result, all the locations can be encoded with the same length of bits if the same precision ($h$) is applied to all the locations.  

\vspace{-0.1in}

\begin{figure}[!tbh]
\centering		\includegraphics[angle=0,width=0.85\linewidth]{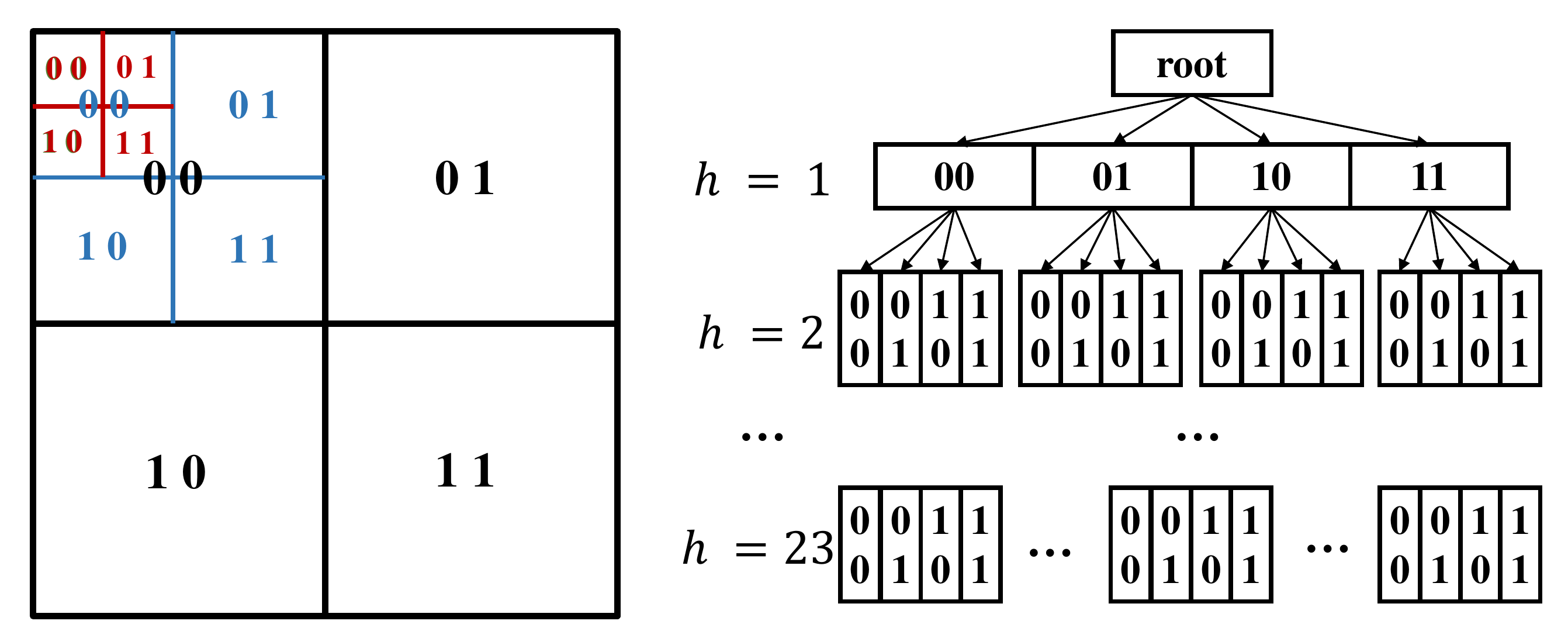}
	\vspace{-0.05in}
	\caption[Optional caption for list of figures]
	{Hierarchical encoding for locations}
	\label{fig:loc_encoding}	\vspace{-0.15in}
\end{figure}

\begin{example}[Encoding for ``New York'']
The coordinates of the center of ``New York'' are $(40.730610, -73.935242)$. Given $h=23$, the location is encoded as ``e1147b6afff'' (hex of the bit string). \end{example}

\noindent\textbf{Location Groups}. 
With hierarchical encoding for the location domain $\mathcal{D}$, the distance between any two locations $x, x'\in \mathcal{D}$ can be directly measured by the longest common prefixes (LCP) of their encoded bit strings.
Then, given a location $x$ and any of its output groups $G_j(x), j\in [1,m]$, we define the LCP of the group.

\begin{definition}[Group LCP]
Given an input location $x$ and any of its groups $G_j(x), j\in [1,m]$, the group LCP (aka. GLCP) is the shortest LCP between the input location $x$ and $\forall y\in G_j(x)$. The length of GLCP for group $G_j(x)$ is denoted as $\beta_j(x)$.
\end{definition}

Thus, the distance between the input location $x$ and each location group $G_j(x), j\in[1,m]$ can be measured by the length of its GLCP $\beta_j(x)$: the larger, the closer. Then, we can partition all the output locations into groups using the GLCP lengths. In each group $G_j(x)$, all the locations share a prefix with at least $\beta_j(x)$ bits with location $x$ (applying such rule for partitioning could reduce the complexity of partitioning to $O(d)$ though not optimal). 
For the group with a longer GLCP shared with the input location $x$, higher probabilities will be assigned to them (for perturbing $x$).\footnote{In \texttt{SRR}, every input location $x$ will be only perturbed to another location $y$ in the domain $\mathcal{D}$ (rather than an arbitrary location on the map).}

\vspace{0.05in}

\noindent\textbf{Location Partitioning}. We next partition the locations into $m$ groups for each input $x\in\mathcal{D}$, and assign the same perturbation probability to all the locations in the same group. 
Specifically, for $m$ groups, we define a GLCP length vector $\{\beta_1(x), \dots,\beta_m(x)\}$. All the encoded locations in group $G_j(x), 1\leq j\leq m$ share \emph{at least} $\beta_j(x)$-bit prefix with $x$. Then, $\beta_1(x)>\beta_2(x)>\dots>\beta_m(x)$ since $G_1(x)$ is the closest group to the input location $x$. 

\begin{figure}[!tbh]
\centering		\includegraphics[angle=0,width=0.85\linewidth]{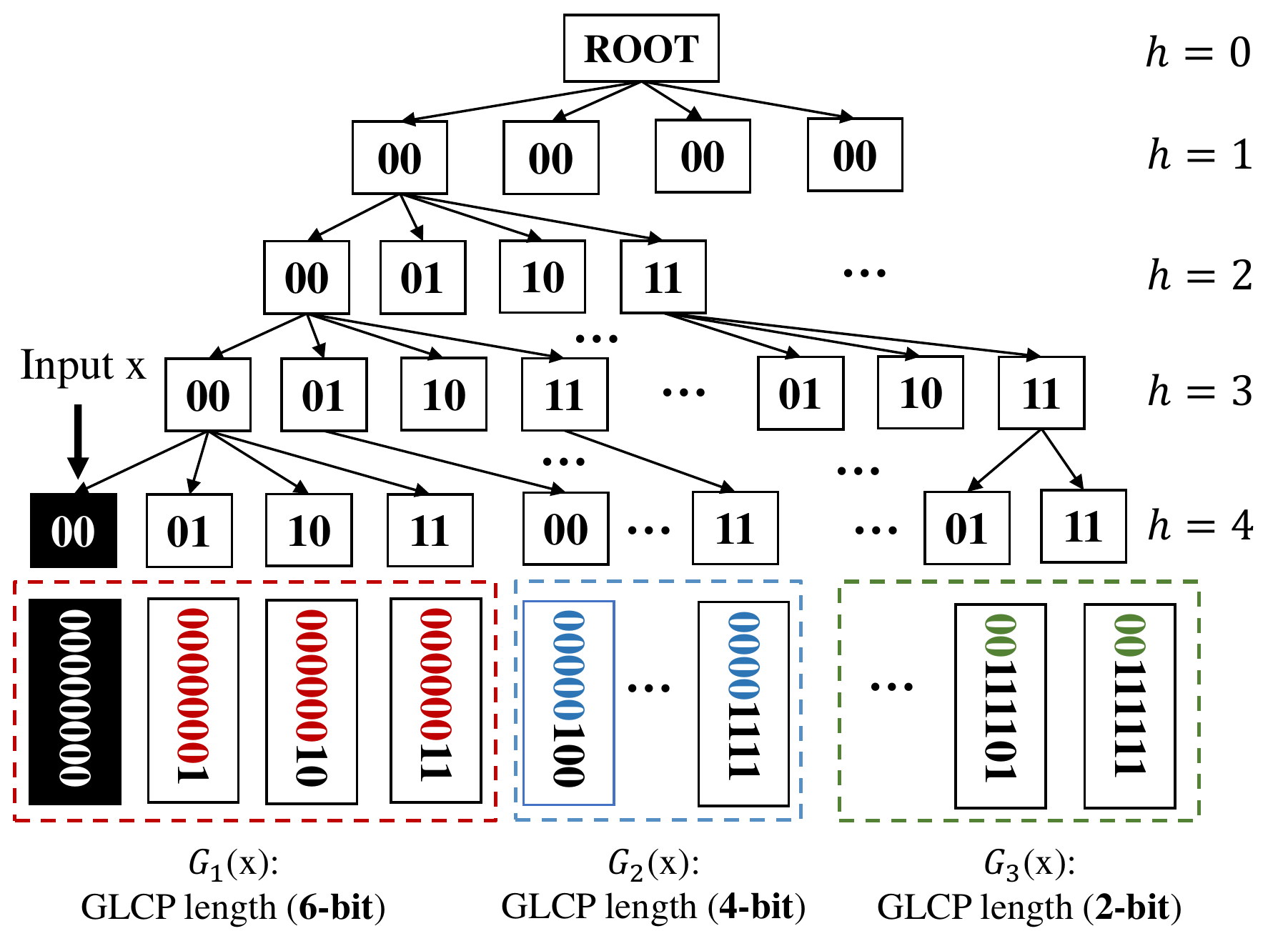}
	%	\hspace{-0.24in}	
	\vspace{-0.1in}
\caption[Optional caption for list of figures]
	{Example of location domain partitioning}
	\vspace{-0.1in}
	\label{fig:example}
\end{figure}

Figure \ref{fig:example} shows an example for partitioning the location domain. 
Given the input location $x$, all the locations are partitioned into three groups with the GLCP lengths $\{\beta_1(x)=6, \beta_2(x)=4,\beta_3(x)=2\}$ where $m=3$. In $G_1(x)$, $G_2(x)$ and $G_3(x)$, all the locations share at least 6-bit, 4-bit and 2-bit prefix with $x$, respectively. Thus, given any GLCP length vector $\beta_1(x),\dots, \beta_m(x)$, the $m$ groups $G_1(x),\dots, G_m(x)$ for the input location $x$ can be efficiently generated with complexity $O(d)$. 
Then, denoting the LCP between input $x$ and output $y$ as $LCP(x,y)$, the optimal $\{\beta_1(x)$, $\dots,\beta_m(x)\}$ and the $m$ groups that maximizes
$\sum_{\forall y\in\mathcal{D}} LCP(x,y)$ can be derived. 

More specifically, if $m$ is not large, we can traverse all the GLCP lengths $\{\beta_1(x)$, $\dots,\beta_m(x)\}$ where $\beta_1(x)>\beta_2(x)>\dots>\beta_m(x)$ to find the optimal result. Otherwise, the server can apply a meta-heuristic algorithm (e.g., simulated annealing \cite{SA}) to derive a near-optimal $\{\beta_1(x)$, $\dots,\beta_m(x)\}$ for partitioning. 
Next, the location domain $\mathcal{D}$ can be efficiently partitioned by the optimal $\{\beta_1(x)$, $\dots,\beta_m(x)\}$. First, locations sharing a $\beta_1(x)$-bit or longer prefix with $x$ will be assigned to $G_1(x)$; second, the locations sharing a prefix (length between $\beta_2(x)$-bit and $(\beta_1(x)-1)$-bit) with $x$ will be assigned to $G_2(x)$; repeat the above until $G_m(x)$ is formed. 

\vspace{0.05in}

\noindent\textbf{Offline Computation}. Since the optimization and partitioning are solely based on the domain $\mathcal{D}$, they can be executed offline and periodically updated with $\mathcal{D}$ by the server in \texttt{L-SRR}. In general, the location domain is stored in the server of companies and released as public knowledge for users (e.g., Google Maps) and these companies will take about several days to update the domains since these companies have to verify locations before making the changes available to the public. Then, for each $x\in\mathcal{D}$, the perturbation probabilities for all the $m$ output location groups $\alpha_1(x),\dots,\alpha_m(x)$ can also be derived offline (see Section \ref{sec:prob}). This is consistent with other LDP schemes \cite{rappor14,ldpusenix17,boling17}. 

\subsection{Optimal Perturbation Probabilities}
\label{sec:prob}

Recall that the possible output locations can be partitioned into $m$ groups based on their distances to input and the PDF similar to the staircase mechanism \cite{staircase} in differential privacy. We define the perturbation probabilities from input $x$ to all the output locations as follows. 
Given any two output locations $y$ and $y'$ in any two neighboring groups $y\in G_j(x)$ and $y'\in G_{j+1}(x)$, we have probability $q(y|x)=q(y'|x)+\Delta(x)$ where the step $\Delta(x)\in[0,1)$ is the constant probability difference for any two neighboring groups of input $x$. Compared to the staircase mechanism in differential privacy which aims to the unbounded domain (entire real line or the set of all integers) and these probabilities are geometric sequence to maintain $\epsilon$-DP, for the bounded location domain, the probabilities in the \texttt{L-SRR} follow a linear sequence. Note that the perturbation probability from the given input location $x$ to output location $y$ decreases as $y$ moves to further groups (larger $j$).

Denoting $\alpha_{max}(x)$ and $\alpha_{min}(x)$ as the max and min probabilities in $\alpha_1(x),...,\alpha_m(x)$, we have $\alpha_{max}(x)=\alpha_1(x)$ and $\alpha_{min}(x)=\alpha_m(x)$. 
In \texttt{SRR}, for all the input locations $x\in\mathcal{D}$, we specify a constant $c\geq 1$ as the ratio $\frac{\alpha_{max}(x)}{\alpha_{min}(x)}$. Thus, we have: 

\vspace{-0.05in}

%\small

\begin{equation}
\Delta(x)=\frac{\alpha_{max}(x)-\alpha_{min}(x)}{m-1}=\frac{{\alpha_{min}(x)} \cdot (c-1)}{m-1} 
\label{eq:delta}
\end{equation}

\normalsize

For each $x\in\mathcal{D}$, the sum of the perturbation probabilities of all the output locations is $1$. Given the differences of perturbation probabilities for output locations in different groups in Equation \ref{eq:delta} and the number of output locations in each group, all the perturbation probabilities can be derived, including $\alpha_{max}(x)$ and $\alpha_{min}(x)$:

\vspace{-0.18in}

%\small

\begin{align}
%\scriptsize
\alpha_{min}(x)=&\frac{m-1}{(m-1)d\cdot c-(c-1)\sum_{j=2}^{m}[(j-1)\cdot|G_j(x)|]}\nonumber\\
\alpha_{max}(x)=&\alpha_{min}(x)\cdot c 
\label{eq:alpha_min}
\end{align}

\normalsize

where $d$ is the location domain size and $|G_j(x)|$ is the size of group $G_j(x)$. Notice that, different $\alpha_1(x),\dots, \alpha_m(x)$ will be derived for different input location $x$ since the group sizes $\forall j\in[1,m], |G_j(x)|$ might be different for different $x$. Thus, the privacy upper bound $\epsilon$ can be computed (for any two input locations $x,x'\in\mathcal{D}$). 

\begin{theorem}
Staircase randomized response (\texttt{SRR}) satisfies $\epsilon$-local differential privacy, where 

\vspace{-0.1in}

\small
\begin{equation*}
\epsilon=\max_{x,x'\in \mathcal{D}}\log(c \cdot \frac{(m-1)d\cdot c-(c-1)\sum_{j=2}^{m-1}[(j-1)\cdot|G_j(x)|]}{(m-1)d\cdot c-(c-1)\sum_{j=2}^{m-1}[(j-1)\cdot|G_j(x')|]})
\end{equation*}
\normalsize
\label{theorem:epsilondp}
\end{theorem}
\begin{proof}
Please see the proof detail in Appendix \ref{sec:ldpproof}.
\end{proof}

For each input location $x\in\mathcal{D}$, the groups $G_1(x)$, $\dots$, $G_m(x)$ are constants if $m$ and $\mathcal{D}$ are specified (as discussed in Section \ref{sec:Encoding}). Thus, given the value of $c$, we can derive a constant $\epsilon$ as a strict privacy upper bound for the LDP guarantee.

\vspace{0.05in}

\noindent\textbf{Selecting $c$ for $\epsilon$-LDP}. Since $\epsilon$ is positively correlated to $c$, for any desired $\epsilon$-LDP, the required $c$ can be uniquely calculated using $\epsilon$, $\mathcal{D}$ and $m$ (see the relationship between $\epsilon$ and $c$ in Figure \ref{fig:mc}). Then, all the perturbation probabilities $\alpha_1(x),\dots, \alpha_m(x)$ for all the input locations $x\in \mathcal{D}$ can be derived and made available to the users. 

\vspace{0.05in}

\noindent\textbf{Optimal $m$ with Mutual Information}. In practice, both the server and clients do not know the data distribution before collecting them. Hence, it is critical to learn that the optimal $m$ is also \emph{independent of input data} and ensure good utility for all possible location data distributions in the \texttt{SRR} mechanism. To this end, we will optimize $m$ for location domain partitioning with the mutual information \cite{MI2,MI} between the input $x$ and output $y$, which can measure the mutual dependence between them. As mutual information varies for different distributions, the maximum mutual information can cover all the cases (since the mutual dependence of any case would not violate such dependence \cite{NumLDP}). Thus, the optimal $m$ can be derived by the upper bound of mutual information for all the distributions \cite{MI3,NumLDP}. Specifically, the mutual information between $x$ and $y$ is expressed by the difference between the differential entropy and conditional differential entropy of $x$ and $y$ \cite{NumLDP}:

\vspace{-0.05in}

\small

\begin{equation}
I(X,Y)=H(X)-H(X|Y)=H(Y)-H(Y|X)
\end{equation}

\normalsize

where $H(\cdot)$ is the entropy function. $X$ and $Y$ are the input and output random variables representing the input and output, respectively. Since no prior knowledge on the input data, it considers the distribution of $y$ as uniform distribution $U$ to maximize the mutual information (the output $y$ is the random sampling result) \cite{ULDP19}. $H(U)$ is an upper bound for any possible input distribution \cite{MI2}. Thus, we have:

\vspace{-0.1in}

\small

\begin{equation}
I(X,Y) \leq H(U)-H(Y|X)
\end{equation}

\normalsize

where $H(U)=\log d$. The conditional differential entropy $H(Y|X)$ can be computed as below:

\vspace{-0.15in}

\small

\begin{align*}
H(Y|X)
&=-[\sum_{j=1}^m |G_j(x)| \cdot \alpha_j(x) \cdot \log\alpha_j(x)]\nonumber\\
&\geq -d\cdot\alpha_{min}(x) \log \alpha_{max}(x)
\end{align*}

\normalsize

Thus, $H(Y|X)$ is lower bounded by $-d \cdot \alpha_{min}(x)\log \alpha_{max}(x)$ for $\alpha_1(x),\dots, \alpha_m(x)$. Finally, the upper bound of mutual information can be expressed with the number of groups $m$:

\vspace{-0.15in}

\small

\begin{align*}
I(X,Y) &\leq \log d-H(Y|X) \leq  \log d + d \cdot\alpha_{min}(x) \log \alpha_{max}(x)
\end{align*}

\normalsize

We then explore the optimal $m$ based on the mutual information metric. Since the smaller mutual information between two variables indicates more independence between them, and the mutual information on $m$ for LDP is convex (as proven in Appendix \ref{sec:convex}), the optimal $m$ can be computed by making the derivation of the upper bound to 0 which is equal to minimize the mutual information bound.

\begin{lemma}
\label{lemma:opt}
The optimal $m$ to minimize the mutual information bound is $m=\frac{2 \cdot (c \cdot d - e^{1+\log c})}{(c-1) \cdot d}$. 
\end{lemma}

\begin{proof}
The mutual information bound is
$\log d + d \cdot  \frac{m-1}{(m-1)\cdot c \cdot d- R} \cdot \log\frac{c(m-1)}{(m-1)\cdot c \cdot d-R}$
where $R=  (\sum_{j=2}^{m}\{(j-1)\cdot|G_j|\}) \cdot (c-1)$ is a part of $\alpha_{min}(x)$ (see Equation \ref{eq:alpha_min}). We can see that $R$ is also determined by $m$. If $|G_1| \neq |G_2| \neq \cdots \neq |G_m|$, $R$ non-differentiable (discrete). To solve this, we consider the worst case: assuming group size $d$ and $R$ is replaced with $R_{max}=(\sum_{j=2}^{m}\{(j-1)\cdot d\}) \cdot (c-1)$ (relaxed). The mutual information bound can be derived as below:

\vspace{-0.1in}

\small
\begin{align*}
& ~~[\frac{m-1}{(m-1)\cdot c \cdot d- R_{max}}
\cdot \log\frac{c(m-1)}{(m-1)\cdot c \cdot d- R_{max}}]'\\
=&~~(\log\frac{m-1}{(m-1)\cdot c \cdot d- R_{max}}+\log c+1)\cdot \frac{(m-1)\cdot R'_{max}-R_{max}}{[(m-1)\cdot c \cdot d-R_{max}]^2}
\end{align*}

\normalsize

Due to $R_{max}=(\sum_{j=2}^{m}(j-1)\cdot d) \cdot (c-1)$, we have:

\vspace{-0.1in}

\small

\begin{equation}
R_{max}=(c-1)\cdot d\cdot \frac{m^2-m}{2},~~~R'_{max}=(c-1)\cdot d\cdot (m-\frac{1}{2})
\end{equation}

\normalsize

Then, we replace the derivative of mutual information with $R_{max}$ and $R'_{max}$. Since $(m-1)\cdot (R'_{max})<R_{max}$, the second part of the derivative cannot be 0. Thus, $m$ is optimal when $\log\frac{m-1}{(m-1)\cdot c \cdot d- R_{max}}+\log c+1=0$, and we have $m=\frac{2 \cdot (c \cdot d - e^{\log c+1})}{(c-1) \cdot d}$.
\end{proof}

\noindent\textbf{Specifying $\epsilon$ for LDP}. In our setting, there are three parameters $\epsilon$, $c$ and $m$. With the given privacy requirement $\epsilon$, the server can calculate the $m$ and the corresponding $c$ with Lemma \ref{lemma:opt} and Theorem \ref{theorem:epsilondp} to make the privacy meet the requirement $\epsilon$. Specifically, we can set a value $c$ and get the corresponding $m$ with Lemma \ref{lemma:opt}. Since the location domain can be partitioned into $m$ groups and $\forall j\in[2, m-1], |G_j(x)|$ are fixed for all $x$, we can then calculate the privacy bound by Theorem \ref{theorem:epsilondp} to see if it meets the privacy requirement $\epsilon$. Per Theorem \ref{theorem:epsilondp}, the $\epsilon$ is positively correlated to $c$ with the fixed $m$ and partition groups. Thus, there should be many values $c$ that make the privacy requirement satisfy $\epsilon$. For example, if the $c$ value equals to 5 to meet the privacy requirement $\epsilon=6$, the value less than 5 would make $\epsilon$ smaller which also meets the privacy requirement. However, to fully utilize the privacy that can make the utility maximize, it should  only take the maximum of $c$ with the fixed domain.  
Figure \ref{fig:rp} shows the numeric results for $c=\frac{\alpha_{max}(x)}{\alpha_{min}(x)}, \forall x\in\mathcal{D}$ and the optimal $m$ versus a varying $\epsilon\in[0.01,20]$ (given four different domains in our experimental datasets). The plots confirm that $\epsilon$ is  positively correlated to $c$ (given any domain $\mathcal{D}$), and $c$ is extremely close to $e^\epsilon$ (slightly smaller). In the experiment, with the given $\epsilon$ value and the domain $\mathcal{D}$, we search the maximum value $c$ to satisfy the $\epsilon$-LDP by the binary search method. In Figure \ref{fig:rec}, the optimal $m$ is 
mainly determined by $\epsilon$. The optimal $m$ (rounded to its floor or ceiling) is a small integer, e.g., 2-6 for all the four different domains. 

\vspace{-0.1in}

\begin{figure}[!h]
	\centering
	\subfigure[$\log c$ vs $\epsilon$ (baseline curve $\epsilon=\log c$)]{
		\includegraphics[angle=0, width=0.49\linewidth]{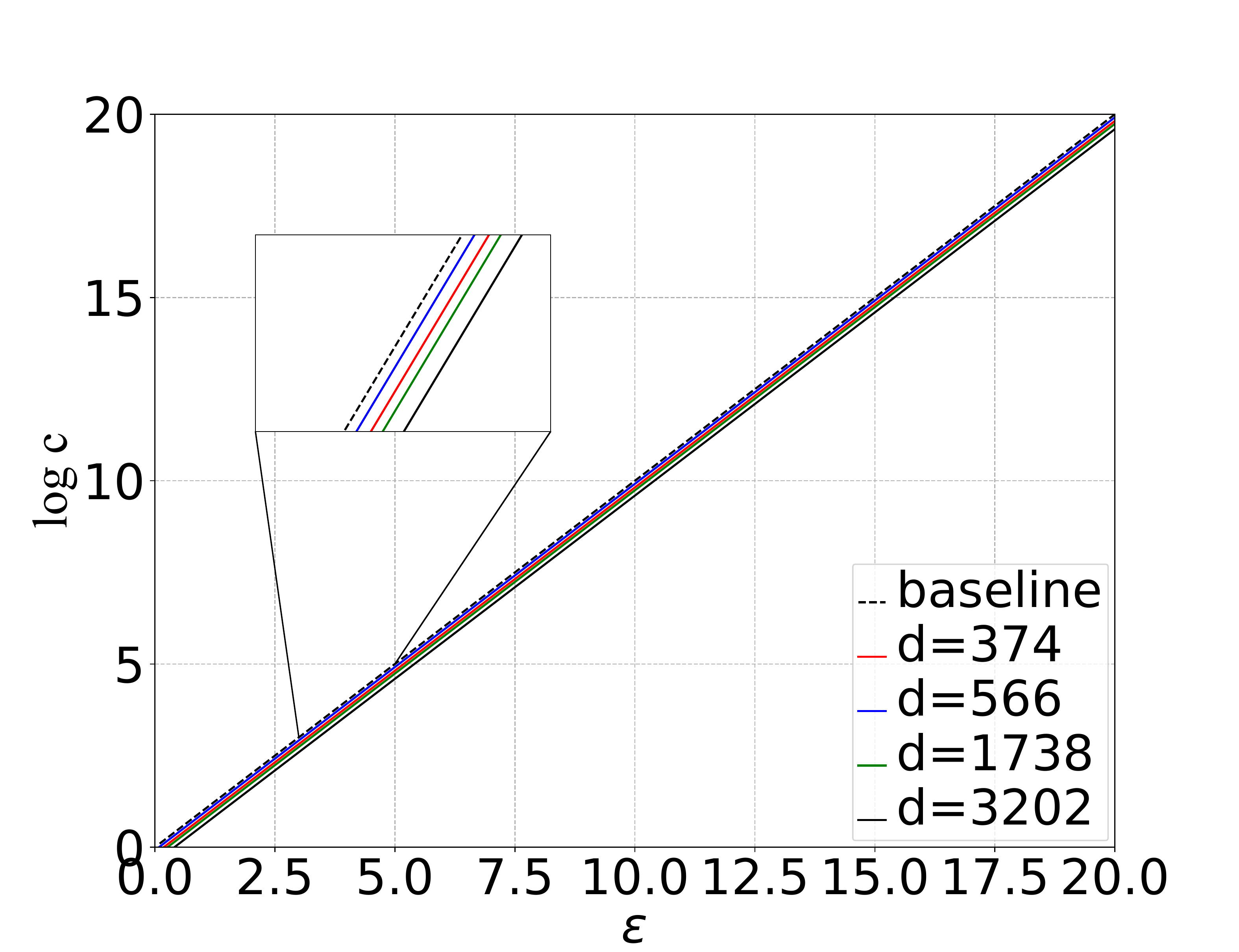}
		\label{fig:mc}}
	\hspace{-0.1in}
	\subfigure[Optimal $m$ vs $\epsilon$]{
		\includegraphics[angle=0, width=0.48\linewidth]{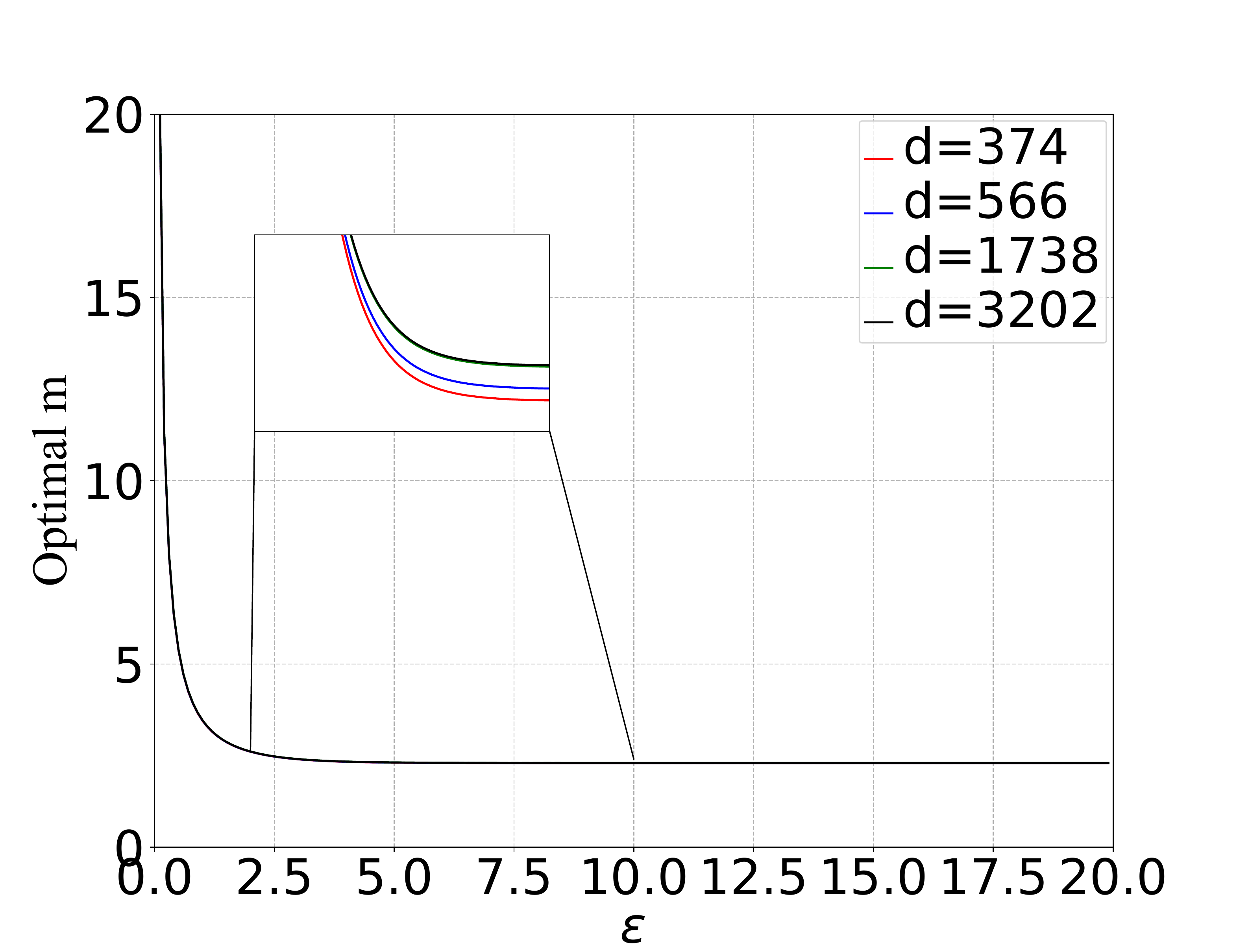}
		\label{fig:rec} }\vspace{-0.15in}
	\caption[Optional caption for list of figures]
	{$\log c$ and optimal $m$ vs $\epsilon$ with various domain size $d$; domain size $d$ is $374, 566, 1738$, and $3202$ in datasets Portcabs \cite{cabs}, Geolife \cite{Geolife}, Gowalla \cite{Gowalla}, and Foursquare \cite{Four}, respectively}\vspace{-0.15in}
	\label{fig:rp}
\end{figure}

\subsection{Perturbation Algorithm}

For each location $x\in \mathcal{D}$, the server partitions $m$ groups $G_1(x),\dots,\\G_m(x)$ and derives the perturbation probabilities for all the output locations in $m$ groups $\alpha_1(x),\dots, \alpha_m(x)$. After receiving such information from the server, each client perturbs its location $x$ by sampling the output location $y$. See details in Algorithm \ref{algm:sampling}. 

\IncMargin{1em}

\begin{algorithm}[!h]
\small
%\footnotesize
\SetKwInOut{Input}{Input}\SetKwInOut{Output}{Output}

\Input{user location $x$, privacy budget $\epsilon$, and domain $\mathcal{D}$}

\Output{perturbed location $y$}

server pre-computes the optimal $m$ and $\beta_j(x), j\in[1,m]$

\ForEach{location $x\in\mathcal{D}$}{
\ForEach{group $j \in [1,m]$}{
\ForEach{location $z\in \mathcal{D}$}{
\If{$length(LCP(z,x))\geq \beta_j(x)$}{

$G_j(x)\leftarrow z$; $\mathcal{D}\leftarrow \mathcal{D} \setminus z$

}

}

}

\ForEach{$j \in [1,m]$}{
compute the perturbation probability $\alpha_j(x)$ for locations in $G_j(x)$ 
}

} 

client samples an output location $y$ from all the locations in $G_1(x), \cdots, G_m(x)$ (per Equation \ref{eq:SRR}) and submit it to the server

\caption{Staircase Randomized Response}
	\label{algm:sampling}
\end{algorithm} 

\DecMargin{1em}

\subsection{Distribution Estimation}
Similar to other LDP mechanisms, the expectation of the aggregated random location counts would be biased \cite{ldpusenix17}. Given samples from unknown data distribution $p$, estimating the distribution $\tilde{p}$ of $p$ has been extensively studied \cite{emp,NumLDP}. In \texttt{L-SRR}, we extend the empirical estimation method with two perturbation probabilities \cite{emp} to estimate the location distribution from the perturbed locations using staircase perturbation probabilities. In our experiment, we also compare the performance of \cite{emp} (named \texttt{HR}) with \texttt{L-SRR}.

In the GRR, the estimation counts of location $x$ is only related to the sampled counts of location $x$. Then, users try to send more information by the perturbation mechanism to have more accurate estimation results. Specifically, in the empirical estimation, for each $x\in\mathcal{D}$, the server creates a candidate location set $C_x$ for input $x$ to estimate the item distribution $\tilde{p}$ from the observed noisy distribution $p$. Each set $C_x$ which contains $\frac{d}{2}$ locations is a subset of the domain\footnote{We follow the generation of $C_x$ in \cite{emp}.}. The server will estimate the $p(x)$ by the $C_x$. In \texttt{L-SRR}, the server generates a candidate location set $C_x$ for each $x$ with a Hadamard matrix (a square matrix with either $+1$ or $-1$ entries and mutually orthogonal rows). Given $\mathcal {H}_1=1$, for any $\mathcal {H}_K$, we have:

\begin{equation}      
\mathcal{H}_K=\left(              
  \begin{array}{cc}   
    \mathcal{H}_{K/2} & \mathcal{H}_{K/2}\\  
    \mathcal{H}_{K/2} & -\mathcal{H}_{K/2}\\  
  \end{array}
\right)               
\end{equation}

The server then applies a recursion algorithm \cite{emp} to generate such Hadamard matrix with size $K \times K$ (denoting it as $\mathcal{H}_K \in \{-1,+1\}^{K \times K}$) where $K=2^{\lceil\log_2(d+1)\rceil}$ and $d$ is the domain size \cite{emp}. Then, each row of $\mathcal{H}_K$ except the 1st row (the 1st row includes only ``$1$'' and $\mathcal{H}_K$ includes $d+1$ rows) can be mapped into a unique location in domain $\mathcal{D}$. Specifically, given location $x\in \mathcal{D}$, its candidate set will be derived using the $(i+1)$th row in $\mathcal{H}_K$ where $i$ is the index of $x$ in $\mathcal{D}$. Then, $\forall x \in \mathcal{D}$, we can generate the candidate set $C_{x}$ for each user's input $x$ as the locations related to the column indices with a ``+1'' in the mapping row of matrix $\mathcal{H}_K$ \cite{emp}. We denote the candidate set of all the locations in $\mathcal{D}$ as $\mathcal{H}_K \circ \mathcal{D}$.

Let $p(C_{x})$ be the probability for sampling $y \in C_{x}$. 
Then, we can derive $p(C_{x})$ with the output $y$ in the corresponding candidate set in case of inputs $x$ and $x'$ ($x$ differs from $x'$ and $C_x$ also differs from $C_{x'}$). Thus, we have $\forall x\in\mathcal{D}$, $p(C_{x})=p(x) \sum_{y\in C_{x}}q(y|x)+\sum_{x'_i \neq x} p(x')\cdot [\sum_{y\in C_{x} \setminus C_{x'}}q(y|x')+\sum_{y\in C_{x} \cap C_{x'}}q(y|x')]$,
where $p(x)$ is the distribution of $x$ (to be estimated). 

All the perturbation probabilities $q(y|x)$ are known in Equation \ref{eq:SRR}. Thus, for each $x\in\mathcal{D}$, there exists one equation as above. Given $d$ independent linear equations (due to random coefficients), the $d$ variables $\forall x\in\mathcal{D}, p(x)$ can always be solvable.  
Specifically, $\forall x\in\mathcal{D}, p(C_{x})$ are the observed distribution of all the locations from the aggregated noisy data. Each user sends its perturbed location to the server, which derives the total frequency of all the locations in the pre-computed candidate set of location $x$. Then, the above $d$ equations can be constructed for estimating the distribution of all the locations $\forall x\in \mathcal{D}, p(x)$. We apply the lower-upper $(LU)$ decomposition algorithm \cite{LU1,LU2} to solve these independent linear equations. Moreover,
if the domain $D$ is too large, we can make the heuristic decision using the sampled counts of $x'$ in place of the true count of $x'$ \cite{CLDP}. Algorithm \ref{algm:Est} presents the details. 

\vspace{-0.1in}

\IncMargin{1em}
\vspace{+2.5mm}
\begin{algorithm}[!h]
\small
%\footnotesize
\SetKwInOut{Input}{Input}\SetKwInOut{Output}{Output}

\Input{perturbed locations $y_1,...,y_n$}

\Output{estimated location distribution $\forall x\in\mathcal{D}, \tilde{p}(x)$}

server generates the candidate location set $\mathcal{H}_K \circ \mathcal{D}$ for all the locations in $\mathcal{D}$

\tcp{$\mathbb I$ returns 1 if $y\in C_{x}$; otherwise, 0}

\ForEach{$x \in \mathcal{D}$}{calculate the $p(C_{x})$ with $y_1,...,y_n$:
$p(C_{x}):= \sum_{j=1}^n\frac{\mathbb I \{y_j\in C_{x}\}}{n}$ 

construct a linear equation for $x$ with $p(C_{x})$ and perturbation probabilities
}

solve linear equations with the $LU$ decomposition to derive $\forall x\in \mathcal{D}, p(x)$

return the estimated location distribution $\forall x\in\mathcal{D}, \widetilde{p}(x)=p(x)$

\caption{Location Distribution Estimation}
\label{algm:Est}
\end{algorithm} 
\vspace{-2.5mm}
\DecMargin{1em}
\subsection{Private Retrieval for Client Queries}
\label{sec:pir}

Recall that the client may need to query the estimated location distribution with its true location, e.g.,  $k$ nearest users \cite{KNN} (see Section \ref{sec:case}), and traffic-aware GPS navigation \cite{route} (see Section \ref{sec:gps}). 
In \texttt{L-SRR}, users can retrieve the results from the server using the Private Information Retrieval (PIR) protocol \cite{PIR2, PIR3, PIR} (when needed), which enables any user to privately retrieve information from a database server without letting the server know which record has been retrieved. In the PIR, the database server has an $n$-bit string $V =\{v_1, ...., v_n\}$, and the client would like to know $v_i$. The client first sends an encrypted request $E(i)$ for the $i$-th value to the server, where $E(\cdot)$ denotes encryption function. The server also responds with an encrypted value $r(v_i, E(i))$ (e.g., by quadratic residuosity). Finally, the client can retrieve the record $v_i$ privately based on the server's encrypted response.

Most of the off-the-shelf PIR algorithms can work as a post-processing component  (e.g., \cite{PIR} takes only a few seconds in our experiments).  
Moreover, the local perturbation and distribution estimation require only $\sim 0.014$ second for the client and a few seconds for the server (see Section \ref{sec:ablation}). Thus, the system performance of \texttt{L-SRR} would be very efficient for real-time LBS deployment.

\subsection{Privacy and Utility Analysis}
\label{sec:analysis}

%\vspace{0.05in}
\noindent\textbf{Privacy Analysis.} $\epsilon$-LDP has been proven for the \texttt{SRR} mechanism in Theorem \ref{theorem:epsilondp}.
The server cannot distinguish users' true locations from the noisy data. 
Moreover, as post-processing procedures applied on the results of LDP scheme, the empirical estimation and PIR (if needed) do not leak any extra information \cite{Dwork14}.

\begin{figure*}[!tbh]
\centering		\includegraphics[angle=0,width=0.8\linewidth]{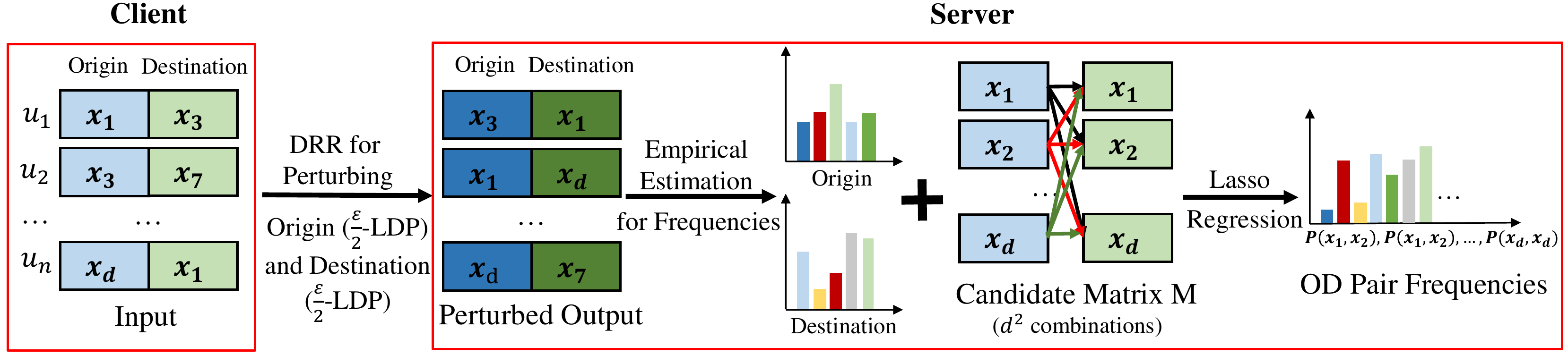}\vspace{-0.1in}	
	\caption{Extending \texttt{SRR} to collect and aggregate origin-destination pairs with $\epsilon$-LDP}\vspace{-0.15in}
	\label{fig:pair2}
\end{figure*}

\vspace{0.05in}
\noindent\textbf{Error Bounds.} 
Error bounds for the estimation methods in LDP schemes can be derived to understand the expectation of the randomized noise.
Then, we derive the error bounds (based on the expectation of the $L_1$ and $L_2$-distance) for the estimated distribution of all the locations $\tilde{p}$ deviated from the true distribution $p$. 
\begin{theorem}
In \texttt{SRR}, $\mathbb{E}[L_1(\tilde{p},p)]\leq\frac{2d}{\sqrt{n} \cdot(2\gamma-d \cdot \mu)}$, 
where $\gamma= \min\{ \sum_{y \in C_{x}} q(y|x), x \in \mathcal{D}\}$ and $\mu=\min\{\alpha_{min}(x), x \in \mathcal{D}\}$.
\label{theorem:l1}
\end{theorem}

\begin{theorem}
in \texttt{SRR}, $\mathbb{E}[L_2(\tilde{p},p)]\leq\frac{2\sqrt{d}}{\sqrt{n}(2\gamma-d \cdot \mu)}$, 
where $\gamma= \min\{ \sum_{y \in C_{x}} q(y|x), x \in \mathcal{D}\}$ and $\mu=\min\{\alpha_{min}(x), x \in \mathcal{D}\}$.
\label{theorem:l2}
\end{theorem}

Both theorems are proven in Appendix \ref{sec:pri_ut}. Both error bounds decline if increasing the privacy bound $\epsilon$ or the number of users $n$ (thus the error bound would be minor in real-world LBS due to a large number of users). Notice that, the expected $L_1$-distance for the \texttt{GRR} mechanism is upper bounded by $\frac{d}{\epsilon}\sqrt{\frac{2(d-1)}{n\pi}}$ \cite{ULDP19}, which can be $\sqrt{d}$ times of the \texttt{SRR} error bound in the worst case.

%% file: advlbs.tex
In this section, we extend \texttt{SRR} to support trajectory-input LBS using two example applications: (1) collecting the origin and destination (OD) of users for OD analysis \cite{OD}, and (2) collecting a sequence of user locations for traffic-aware GPS navigation \cite{Geolife}. 

\subsection{Origin-Destination Analysis}
\label{sec:od}

OD analysis aggregates a pair of origin-destination from each user to estimate the traffic flow \cite{OD}. In this case, the LDP notion (Definition \ref{def:SLDP}) should be extended to protect each user's OD pair.  

\begin{definition}[$\epsilon$-Local Differential Privacy] A randomization algorithm $\mathcal{A}$ satisfies $\epsilon$-LDP, if for any two different location pairs $(x_o, x_d), (x'_o, x'_d)\in \mathcal{D}\times \mathcal{D}$, and for any output location pair $(y_o, y_d) \in range(\mathcal{A})$ sent to the untrusted server, we have
$Pr[\mathcal{A}(x_o, x_d)=(y_o, y_d)]\leq e^{\epsilon}\cdot Pr[\mathcal{A}(x'_o,x'_d)=(y_o, y_d)]$.
\label{def:pair} 
\end{definition}

The LDP scheme for OD analysis should preserve the sequential correlation from the origin to the destination (OD pair). Thus, the domain has been greatly expanded to $d^2$ OD pairs in $\mathcal{D}\times \mathcal{D}$. To avoid the bad utility resulted from a large domain, we extend the Lasso regression \cite{LoPub} to a novel \emph{private matching method} to preserve the OD sequence.\footnote{Lasso regression was used to generate the synthetic high-dimensional dataset with LDP and preserve the correlation across dimensions \cite{LoPub}.} Then, we integrate the private matching into \texttt{L-SRR} to ensure accurate OD distribution with $\epsilon$-LDP.

Specifically, users perturb their two locations with privacy budget $\frac{\epsilon}{2}$ for each. The server receives a large number of noisy samples of all users from specific distributions for origins and destinations, respectively. The server may estimate the distribution from the noisy sample space using the linear regression $\vec y$ = $\textsc{M}*\vec w$, where matrix $\textsc{M}$ includes the predictor variables, vector $\vec y$ includes the response variables, and vector $\vec w$ includes the regression coefficients. The predictor variables in $\textsc{M}$ consist of all the combinations of trajectories from each origin to each destination ($d^2$ pairs), which could be known to the server and client beforehand. Moreover, the response variables $\vec y$ can be estimated from the \texttt{SRR} perturbed values. Notice that, the frequencies of most combinations $(x_o,x_d) \in \mathcal{D}\times\mathcal{D}$ are very small or even equal to zero in LBS. Thus, Lasso regression \cite{LoPub} can effectively solve such sparse linear regression by encoding the predictor variables $\textsc{M}$ for all the OD pairs.

As shown in Figure \ref{fig:pair2}, we have two steps in \texttt{L-SRR}: (1) perturbing the origin and destination separately by each client, and (2) estimating the joint distribution of OD pairs using Lasso regression by the server. Each client first applies \texttt{SRR} to perturb the origin and destination with privacy budget $\frac{\epsilon}{2}$ each. Then, the server estimates the distribution of origin and destination to generate the vector $\vec y$. Meanwhile, the server encodes the overall candidate set of OD pairs $\textsc{M}$ based on the location domain $\mathcal{D}$. Finally, the server fits a Lasso regression model to the vector $\vec y$ and the candidate matrix $\textsc{M}$ to learn $\vec w$. Therefore, the non-zero coefficients in $w$ will be considered as the frequencies for the candidate OD pairs.

\vspace{0.05in}

\noindent\textbf{Privacy Bound}. Although the origin and destination are correlated, each user sends these two perturbed locations sequentially. The sequential composition of releasing two locations would only result in the total leakage ($\epsilon$-LDP) even if they are highly correlated \cite{Dwork14}. The Lasso regression is performed on the two sets of perturbed data (one set of origins and another set of destinations) as post-processing to retain the correlation, which would not consume privacy budget \cite{Dwork14}. Thus, the OD analysis still satisfies $\epsilon$-LDP.

\subsection{Traffic-Aware GPS Navigation}
\label{sec:gps}
In this App, users may seek the route with shortest time by avoiding congested roads. At that moment, users may update and send multiple locations to the server in sequence. Meanwhile, each user will privately retrieve the real-time nearby traffic from the server to help update the route in case of traffic congestion. 

Specifically, the route recommendation algorithm can be deployed in the client to compute the best route with the shortest traveling time on an offline map (integrated with the real-time traffic information from the server) \cite{Geolife}. For any route, the total traveling time $t$ can be predicted with the historical dataset.\footnote{These historical datasets could be obtained from public traces and check-in datasets, or datasets generated from LBS applications.} Also, each user can send the current location $x_i$ to the server again and learn the current traffic density. Then, the client may recompute the best route and update the estimated traveling time. Intuitively, if the suggested route does not have any traffic, it is unnecessary to update the user's location to learn the real-time traffic density (this would avoid consuming more privacy budget). Thus, we follow this idea to extend our \texttt{SRR}. In \texttt{L-SRR}, the client will identify these ``location updates'' (similar to \cite{EvolvingData}).  
Let $\mathbb{T}$ denote a trajectory and $Agg(x_o, x_i), x_i\in \mathbb{T}$ represent the actual traveling time from the origin $x_o$ to current location $x_i$. In the meanwhile, the GPS can predict the piece-wise traveling times between the origin $x_o$ and any location $x_i \in \mathbb{T}$ before the arrival. It is worth noting that the time is treated as the condition for the update (as above). It can be extended to update the location with other criterion in specific applications (e.g., distance, and checkpoints). 

Denoting such predicted time as $Agg^p(x_o, x_i), x_i\in \mathbb{T}$, the client will examine the difference between their actual traveling time $Agg^t(x_o, x_i)$ and the predicted time $Agg^p(x_o, x_i)$ at different locations $x_i\in \mathbb{T}$.  
If the client finds that the actual traveling time $Agg^t(x_o,x_i)$ is significantly more than predicted one $Agg^p(x_o, x_i)$, e.g., delayed time exceeds a threshold: $Agg^t(x_o,x_i)-Agg^p(x_o,x_i)>\theta$, there is likely a traffic congestion. Then, the client requests a ``location update'' to privately upload the perturbed location to the server, and privately retrieve the current traffic density. Moreover, the server will periodically estimate the traffic density using all the perturbed locations collected from the clients in the past time window (e.g., 5 minutes for each time window). 
Once a location update is requested by any client, the server privately delivers the traffic density to the client via the PIR protocol.

\begin{comment}

\begin{algorithm}[!h]

\SetKwInOut{Input}{Input}\SetKwInOut{Output}{Output}

\Input{number of users $n$, privacy parameter $\epsilon$, each user's origin $x_o$ and destination $x_d$, time threshold $\theta$, user threshold $\theta_2$}

Initialize the least time path from $x_o$ to $x_d$ for each user offline with current location distribution and predict consuming time $t$ between any two locations during the path

\tcp{For each user, this step would be executed by client offline}

\ForEach{$i \in [1,n]$}{
\If{$Agg^t_i(x_o,x_i)-Agg^p_i(x_o,x_i)\geq \theta$}
{User $i$ vote ``Yes'' to update: $u^{v}_i=1$}
}

$GlobalUpdate \leftarrow \sum_{i=1}^n u^{v}_i$

\If{$GlobalUpdate \geq \theta_2$}{\ForEach{$i \in [1,n]$}{User $i$ perturbs the location with Algorithm \ref{algm:sampling} with privacy budget $\epsilon$}}
\tcp{Once users update the location, the new origin will be updated to current location} 
\caption{\texttt{SRR} with Voting in \texttt{L-SRR}}
	\label{algm:vote}
\end{algorithm} 
\end{comment}

\vspace{0.05in}

\noindent\textbf{Privacy Bound}. 
Since every perturbed location is individually aggregated (based on individual locations) rather than as a combination, such data collection can be done for all the locations separately and simply follows sequential composition \cite{InputDis2020}. Thus, \texttt{SRR} for such trajectory-input LBS satisfies $\lambda \epsilon$-LDP where $\lambda$ is the number of requested location updates from the origin to the destination. We have empirically evaluated that $\lambda$ is small in practice (e.g., 2 or 3). Finally, PIR may result in side-channel leakage (e.g., who requested the location update may be in the congested areas). If necessary, this can be simply mitigated by an anonymizer (e.g., shuffler \cite{shuffler}), which also further amplifies the LDP protection \cite{shuffler}.

%% file: discuss.tex
\noindent \textbf{Relaxed LDP}. Some recent works \cite{InputDis2020,CLDP,Geo-indistinguishability} relaxed the LDP by considering the input variants. For instance, \texttt{ID-LDP} \cite{InputDis2020} relaxes the LDP with different $\epsilon$ for different inputs; 
geo-indistinguishability (\texttt{GI}) makes every pair of locations indistinguishable, but the ``level" of indistinguishability depends on their distance (locations that are far apart are more distinguishable than locations that are close together); \texttt{CLDP} \cite{CLDP} provides distance discriminative privacy, and relaxes the protection for different pairs of inputs. Different from \texttt{L-SRR}, all of them cannot strictly satisfy $\epsilon$-LDP. 
To validate their limitations on rigorous LDP guarantee, we present some numeric analysis with the same setting (by converting them to $\epsilon$-LDP). \emph{\texttt{PLDP} \cite{location_LDP} is experimentally compared in Section \ref{sec:exp} since it focuses on LBS}.

First, we generate a synthetic dataset including items with uniformly distributed frequencies (the distance between inputs can also be directly measured). For \texttt{ID-LDP}, we randomly assign the privacy bound from $\{0.5\epsilon, 0.8\epsilon, \epsilon\}$ to each distinct item. Since $\{0.5\epsilon, 0.8\epsilon, \epsilon\}$-ID-LDP satisfies $\min\{\{\epsilon\}, 2\times\{0.5\epsilon\}\}$-LDP, it can guarantee $\epsilon$-LDP for all the items. For \texttt{GI}, we sample the output $y$ with the Laplace-based PDF centered at input $x$. For \texttt{CLDP}, we adopt the conversion between $\epsilon$ and $\alpha$ \cite{CLDP}. Table \ref{tab:sim}  
shows the $L_1$-distance of the outputs on different $\epsilon$. The utility of \texttt{L-SRR} significantly outperforms all the relaxed LDP with the same LDP guarantees.

\begin{table}[!h]
 \centering
  \caption{Average $L_1$-distance}
 \vspace{-0.15in}
 \begin{tabular}{|c|ccccc|}
  \hline
Privacy Bound $\epsilon$ &$0.5$& $1$ & $2$  & $3$ & $4$\\  
  \hline
  \multirow{1}{*}{\texttt{ID-LDP}}&2.14 & 1.97 & 1.64 & 1.45 & 1.18\\
    \multirow{1}{*}{\texttt{GI}}&2.21&1.84& 1.75 & 1.43 &1.21\\
  \multirow{1}{*}{\texttt{CLDP}}&0.93&0.90& 0.84 & 0.72 &0.70\\
   \hline
  \multirow{1}{*}{\texttt{L-SRR}}&0.65 & 0.62 & 0.51  & 0.44 & 0.36\\
   \hline
  \end{tabular}%\vspace{-0.15in}
  \label{tab:sim}
\end{table}

\vspace{0.05in}

\noindent \textbf{Generalization.} \texttt{L-SRR} can be potentially extended to other data types if the distances between values/items can be measured (e.g., numerical data). In such contexts, the data items can also be partitioned and staircase perturbation probabilities can be derived and allocated to values/items in different groups. 
We will evaluate its performance in other domains and benchmark with the corresponding LDP schemes (e.g., \texttt{Piecewise} \cite{PM}) in the future.

\vspace{0.05in}

\noindent\textbf{Encoding and Precision}. The precision of the encoded locations can be tuned by the level of the bit string hierarchy. Although larger $h$ more accurately encodes locations, the domain size will grow and thus the perturbation probability (for the true location) may decline for the same privacy. Thus, larger $h$ does not necessarily make the staircase perturbation scheme more accurate (thus we use the standard $h=23$ as Bing Map). In the experiment, every location can only be possibly flipped to other locations in the domain not every pixel on the map. There are two benefits for such encoding and design: (1) locations will not be perturbed to an unrealistic location (e.g., in the ocean), and (2) it is more efficient to compute the perturbation probabilities offline (due to reduced domain size).

\vspace{0.05in}

\noindent \textbf{Larger and Worldwide Domain} In this paper, we evaluated our scheme within each city (four datasets) by following the same settings as other LBS since each experimental dataset is collected within a city. If all the locations on the planet are considered, the domain size would be much larger and the utility might be degraded since the error bound is related to the domain size $d$. To see the utility for larger domain, we are working on a set of experiments by comparing the LDP schemes on 1 city, 2 cities, and more merged cities (merging the domain/data). See the details in Appendix \ref{subsec:largerdomain}.

 \begin{figure*}[!h]
	\centering
	\subfigure[Gowalla]{
		\includegraphics[angle=0, width=0.25\linewidth]{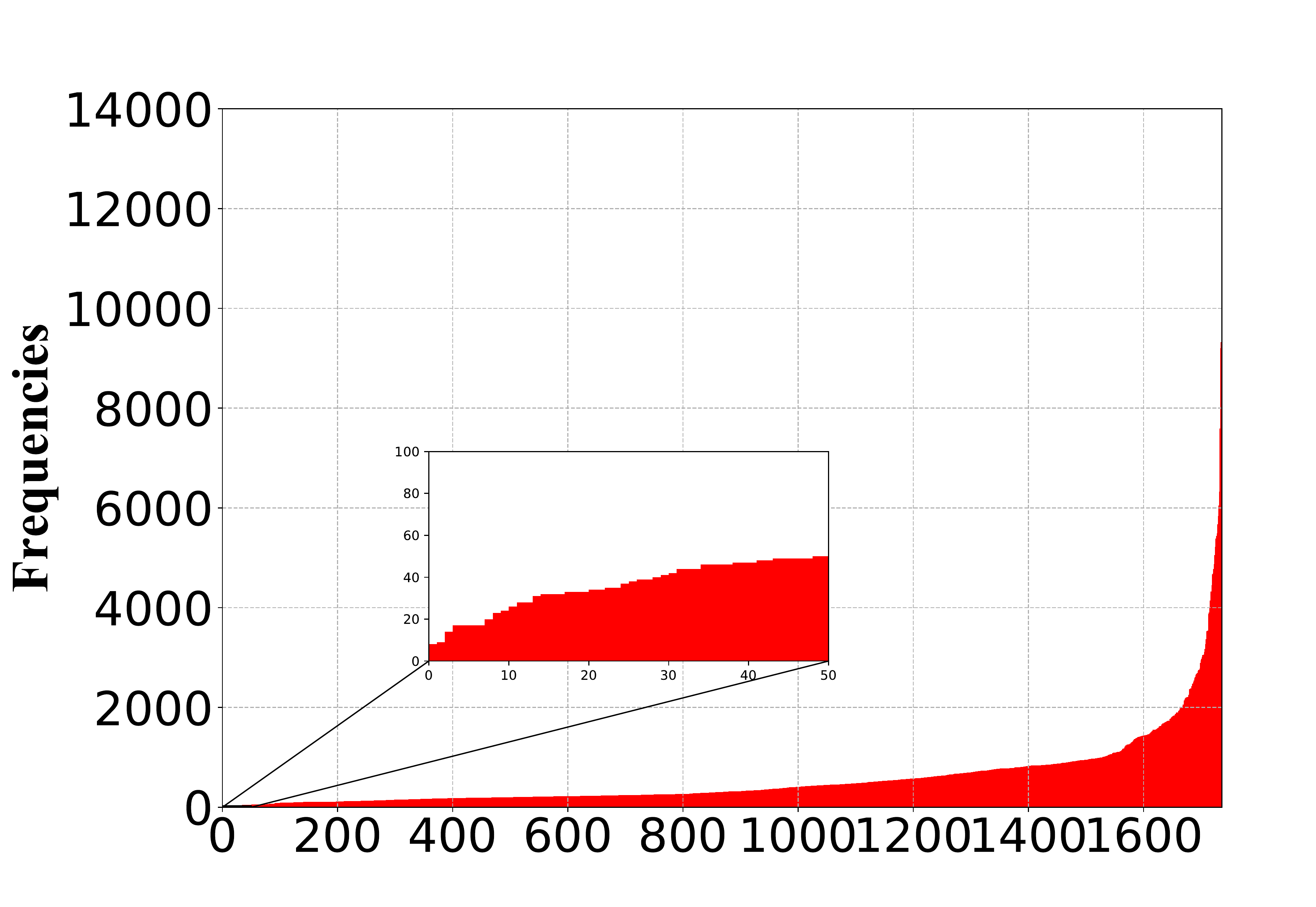}
		\label{fig:O_C} }
		\hspace{-0.15in}
	\subfigure[Geolife]{
		\includegraphics[angle=0, width=0.25\linewidth]{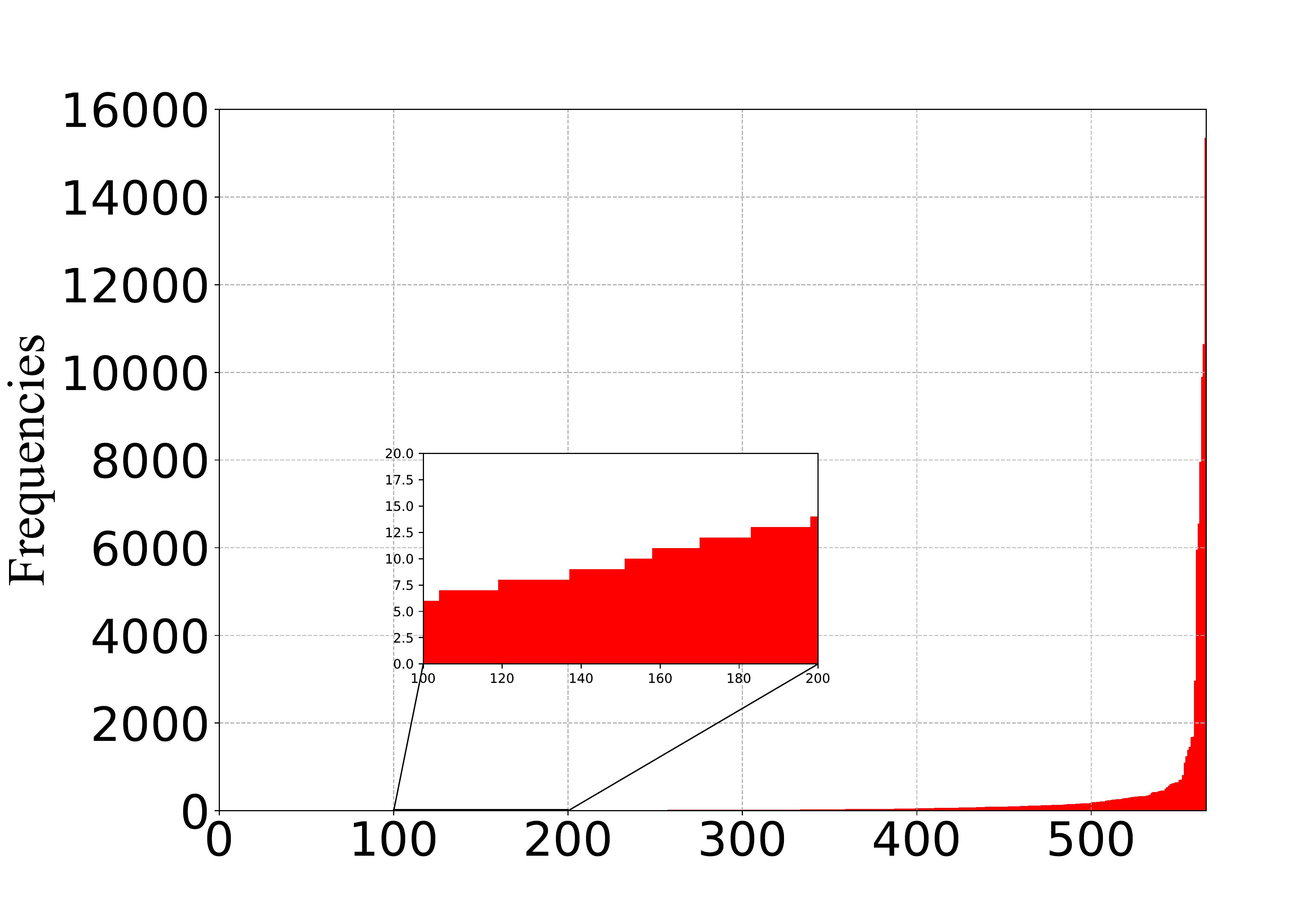}
		\label{fig:O_G} }
		\hspace{-0.15in}
	\subfigure[Portocabs]{
		\includegraphics[angle=0, width=0.25\linewidth]{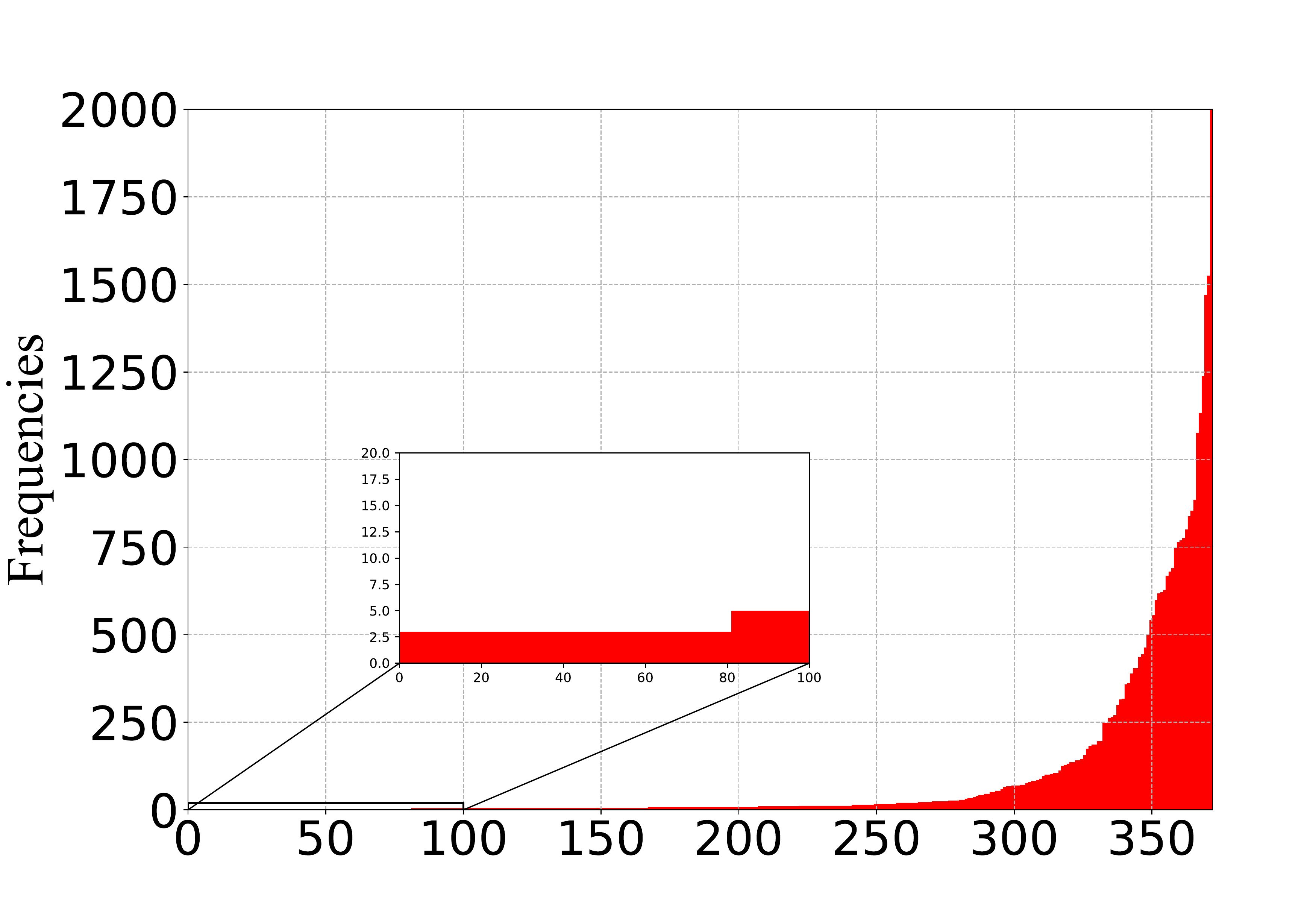}
		\label{fig:O_P}}
		\hspace{-0.15in}
	\subfigure[Foursquare]{
		\includegraphics[angle=0, width=0.25\linewidth]{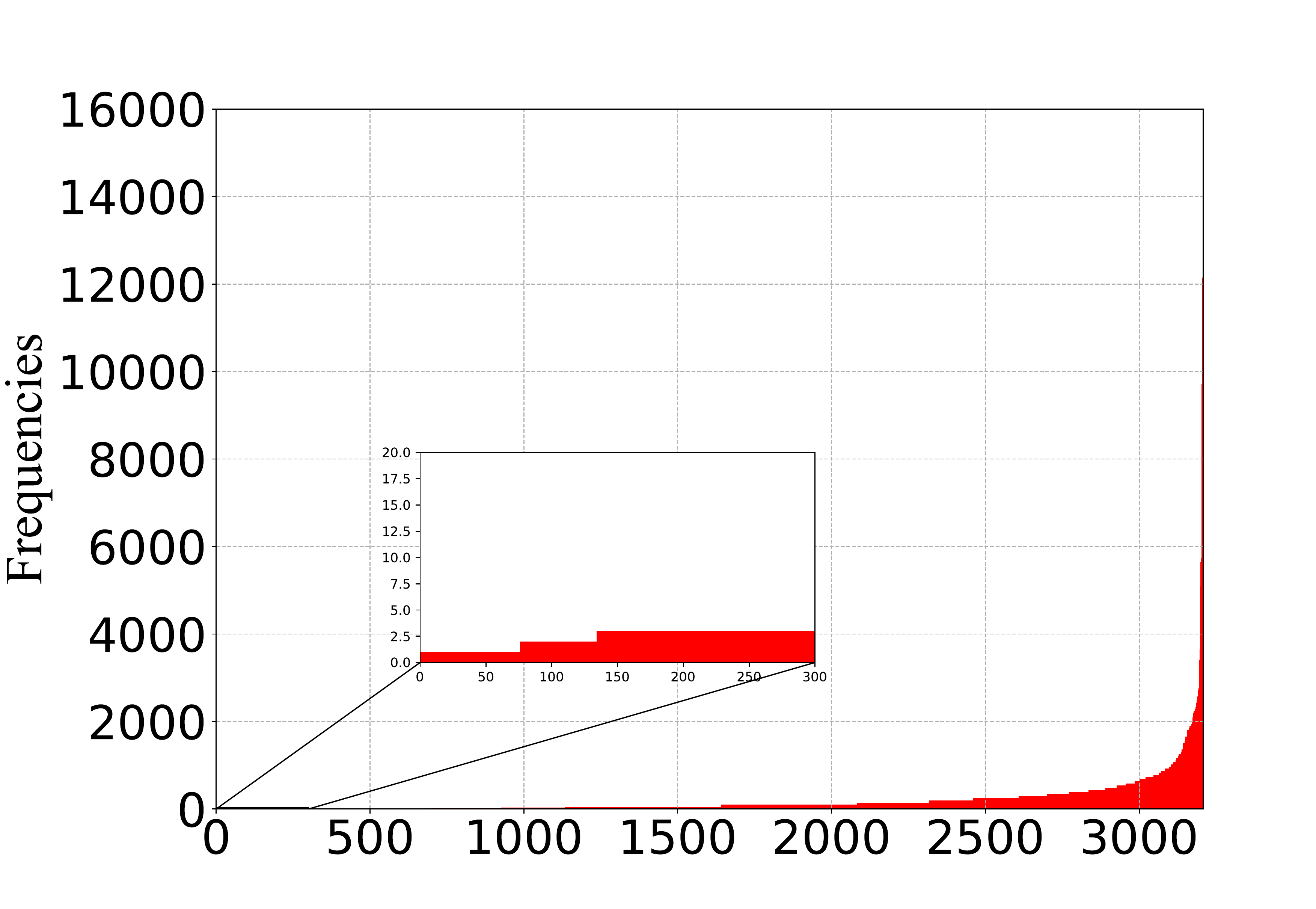}
		\label{fig:O_S}}
		\vspace{-0.15in}
	\caption{Location frequencies in experimental datasets}\vspace{-0.15in}
	\label{fig:Org}
\end{figure*}

\vspace{0.05in}

\noindent \textbf{System Deployment}. \texttt{L-SRR} can be deployed as an application or integrated with the existing LBS applications in the server and clients (e.g., mobile devices). Given the privacy bound $\epsilon$ and a location domain $\mathcal{D}$, the server will pre-compute the required $c$, the optimal $m$, the GLCP for group partitioning $\forall x\in\mathcal{D}, \beta_1(x),\dots,\beta_m(x)$, and the perturbation probabilities $\forall x\in\mathcal{D}, \alpha_1(x),\dots,\alpha_m(x)$ for \texttt{SRR}, and then share them to all the clients. In \texttt{L-SRR}, the location domain is updated periodically by the server rather than per users' requests. It would not cause any privacy leakage, and it is very efficient to update the domain. If a user is at a location not in the domain before the update, the client will approximate it to the nearest location in the domain. 
Each client only needs to perturb their locations based on the stored $md$ perturbation probabilities $q(y|x)$, and then directly send the output to the server. Even if the client may privately retrieve the analysis result related to his/her location from the server, the PIR protocol can be efficiently executed without many overheads. Thus, the clients do not need to be equipped with strong computing capabilities (mobile devices suffice). Each client should download an offline map if required in certain LBS applications, e.g., traffic-aware GPS navigation. 

\vspace{0.05in}

\noindent\textbf{Provable Privacy for PIR and LDP.} The PIR protocol is applied as the post-processing to the query results that guarantees $\epsilon$-LDP. The index (w.r.t. the domain) can be public knowledge and shared to users. The PIR protocol does not cause any additional information leakage since the query results are retrieved based on the encrypted location by employing the provably secure cryptographic technique. From the viewpoint of the server, the PIR request might be originated from any user. Therefore, the probability to identify every user as the querying user is exactly $\frac{1}{n}$ (for all the users). Thus, it does not cause additional leakage from such random guess either (after the private data collection with $\epsilon$-LDP). 

\vspace{0.05in}

\noindent\textbf{Complex Applications}. The staircase randomized response can generate more accurate location distribution than existing LDP mechanisms. As a key building block of LBS applications, such high accurate location frequency/distribution estimation by the proposed \texttt{SRR} mechanism could universally support different LBS applications, including complex LBS such as traffic-aware GPS navigation. In our experiments, we simulate the route recommendation by the GPS, which shows better performance of \texttt{SRR} (see the details in Appendix \ref{sec:trajectory}). In practice, as the LBS application becomes more complicated (e.g., more data collection), \texttt{SRR} would outperform the state-of-the-art LDP schemes more.

\begin{figure*}[!tbh]
	\centering
	\subfigure[Gowalla]{
		\includegraphics[angle=0, width=0.26\linewidth]{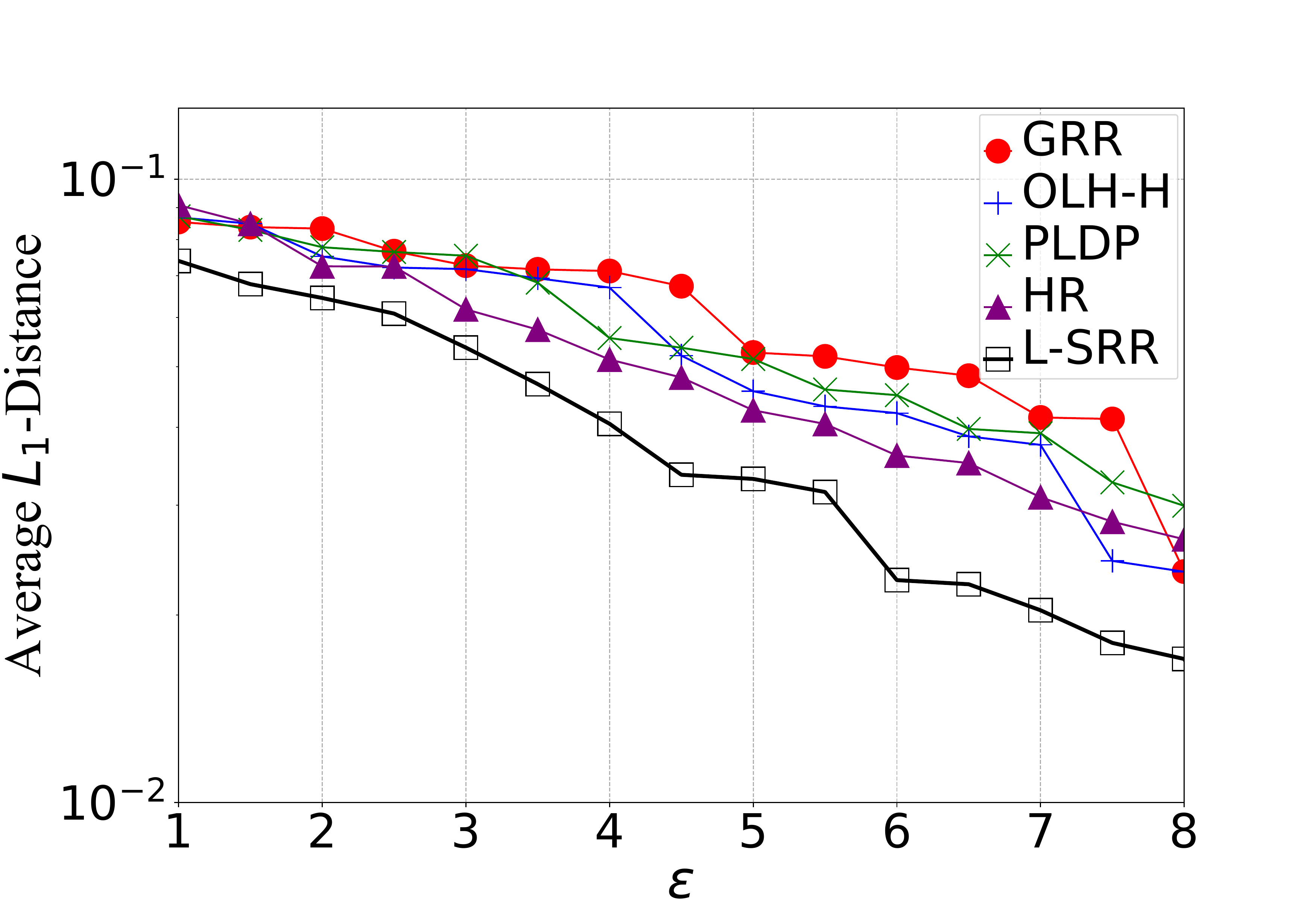}
		\label{fig:L1_C} }
		\hspace{-0.25in}
	\subfigure[Geolife ]{
		\includegraphics[angle=0, width=0.26\linewidth]{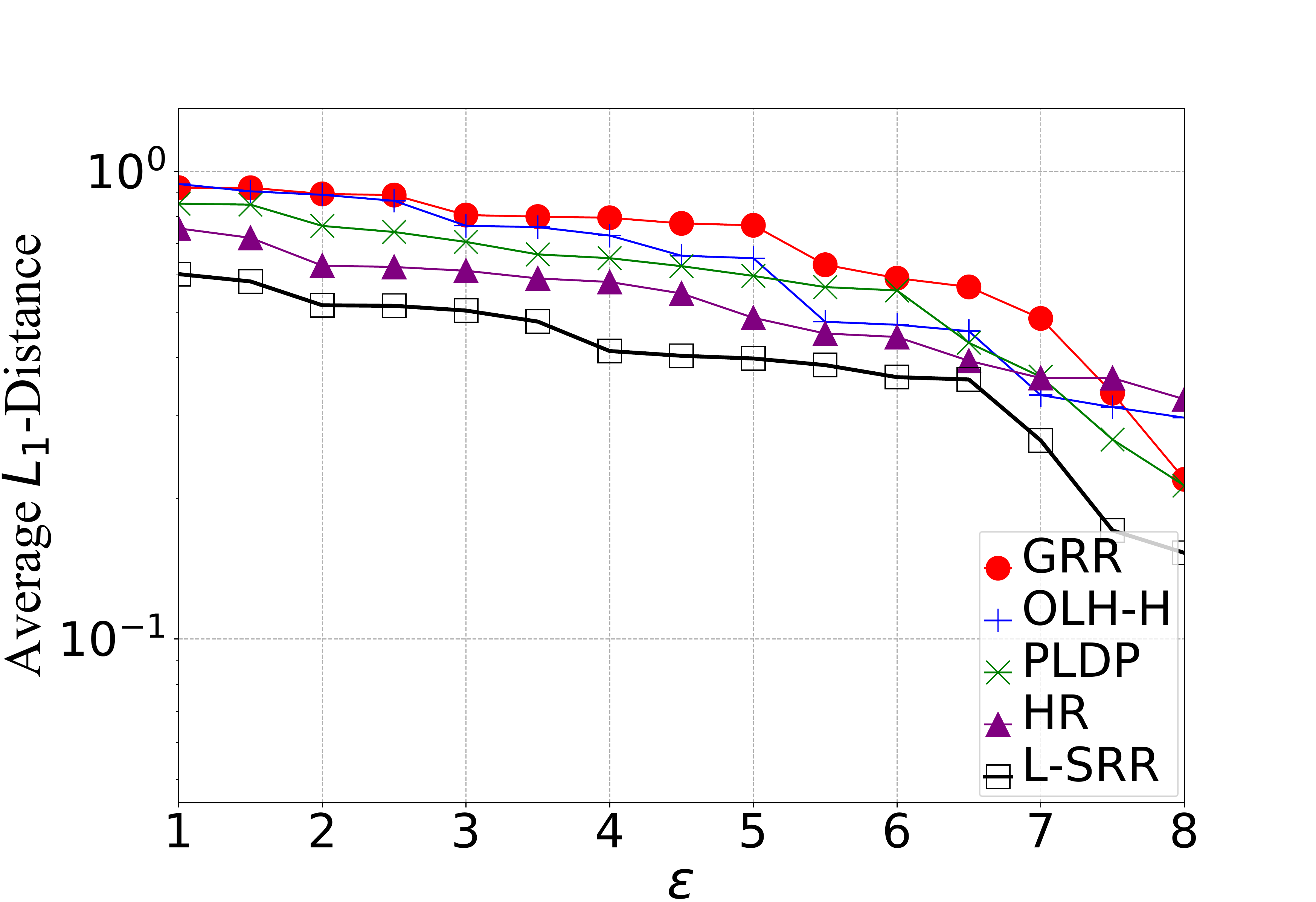}
		\label{fig:L1_G} }
		\hspace{-0.25in}
	\subfigure[Portocabs]{
		\includegraphics[angle=0, width=0.26\linewidth]{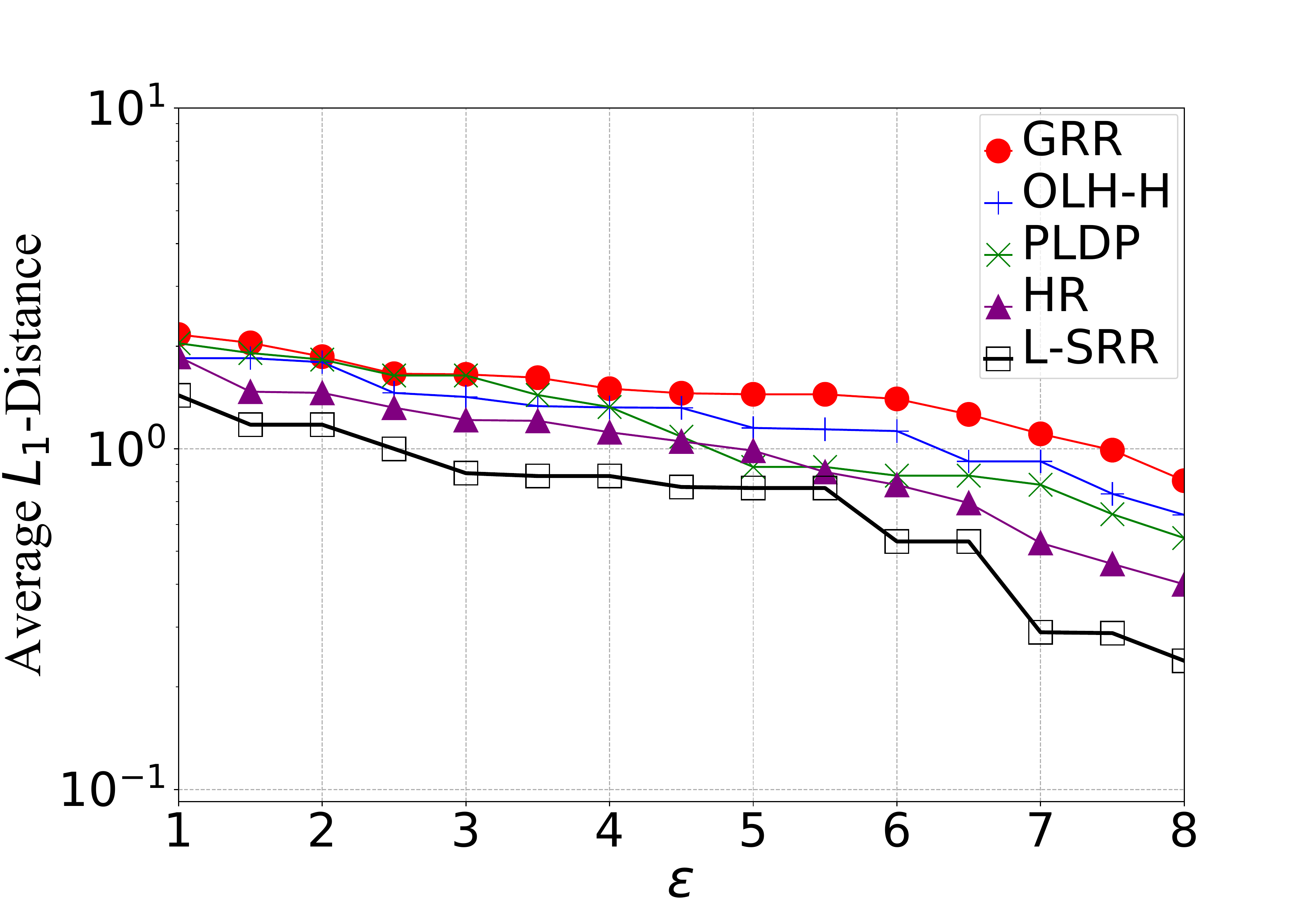}
		\label{fig:L1_P}}
		\hspace{-0.25in}
	\subfigure[Foursquare]{
		\includegraphics[angle=0, width=0.26\linewidth]{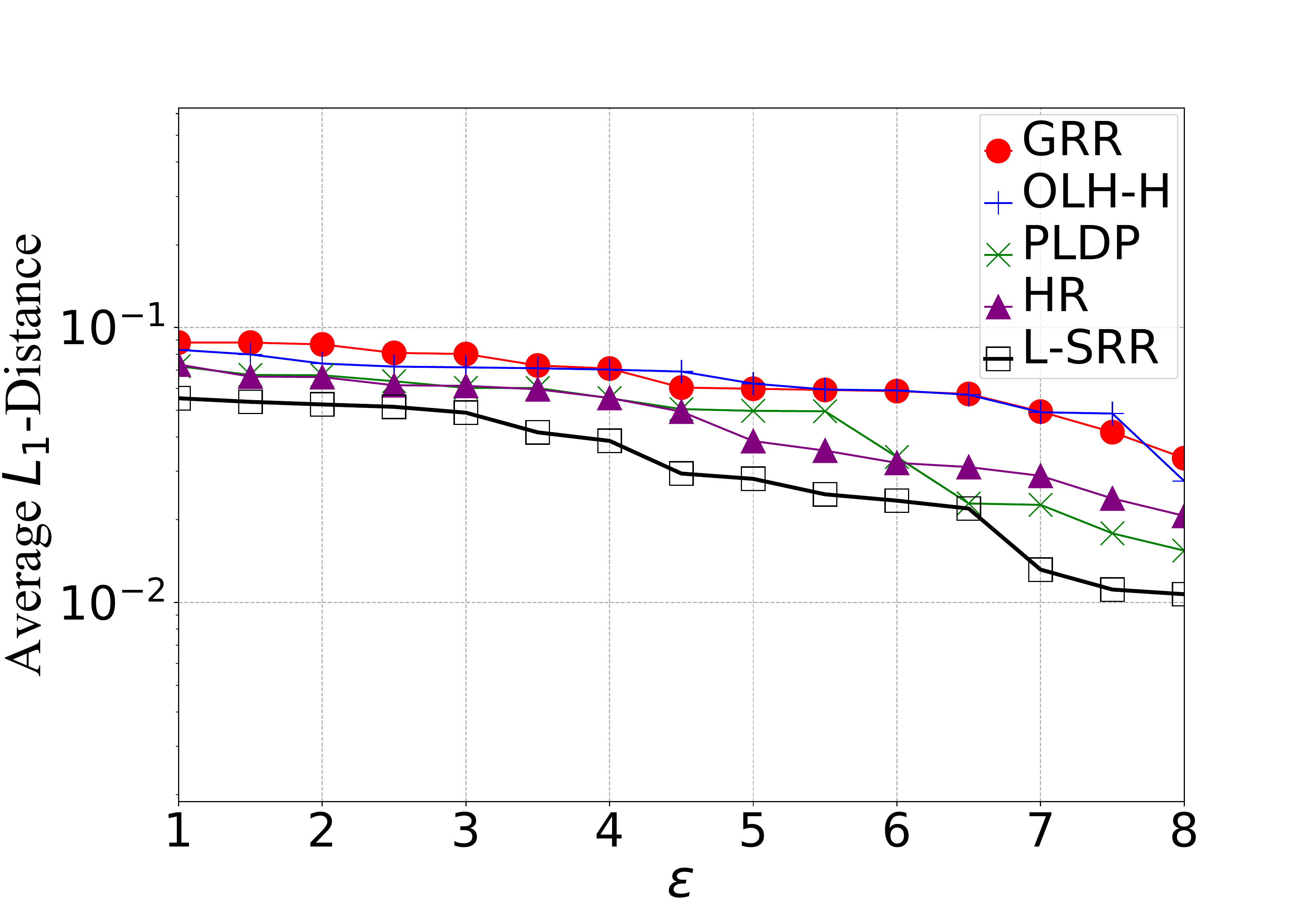}
		\label{fig:L1_S}}
	\subfigure[Gowalla]{
		\includegraphics[angle=0, width=0.26\linewidth]{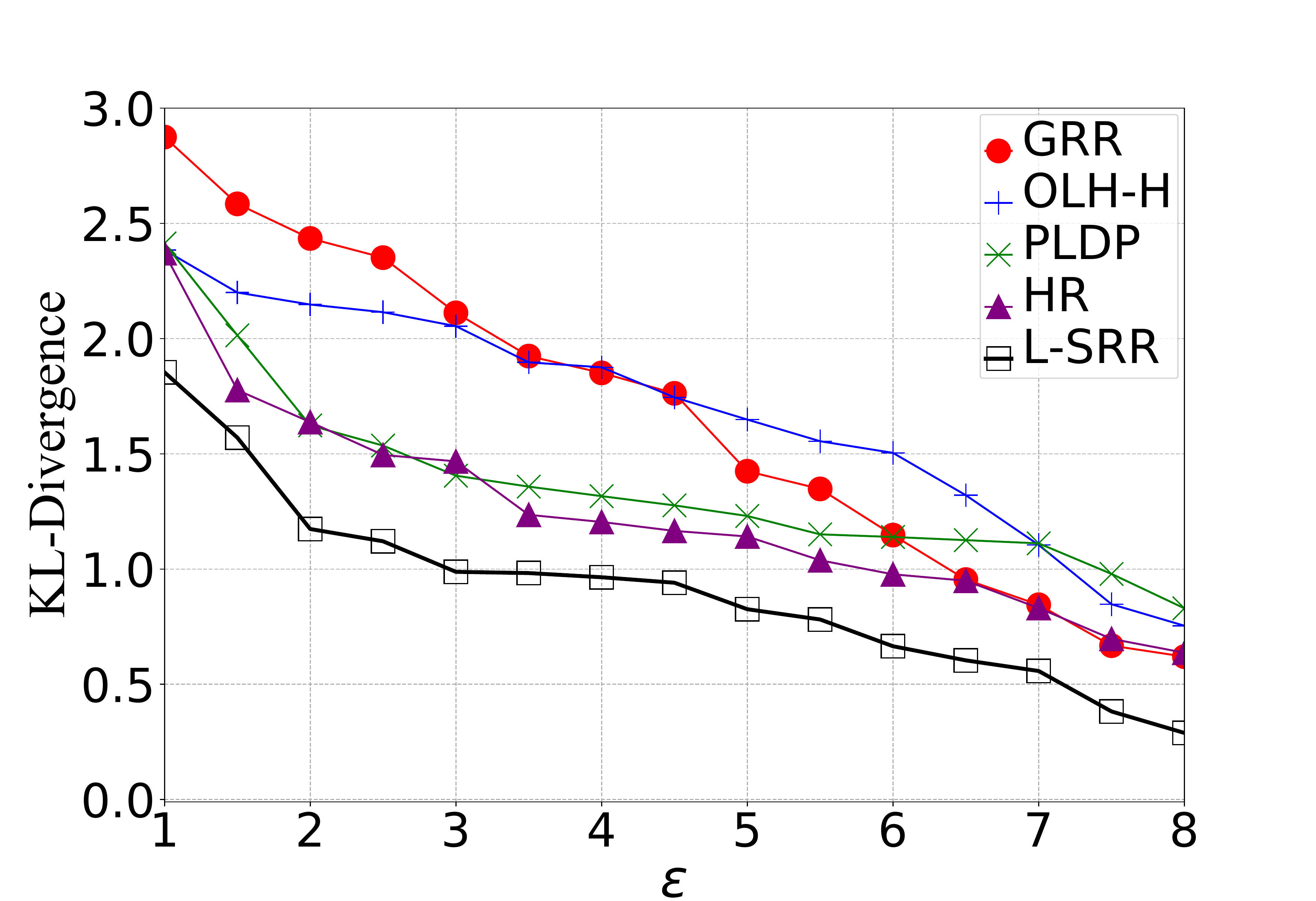}
		\label{fig:KL_C} }
		\hspace{-0.25in}
	\subfigure[Geolife]{
		\includegraphics[angle=0, width=0.26\linewidth]{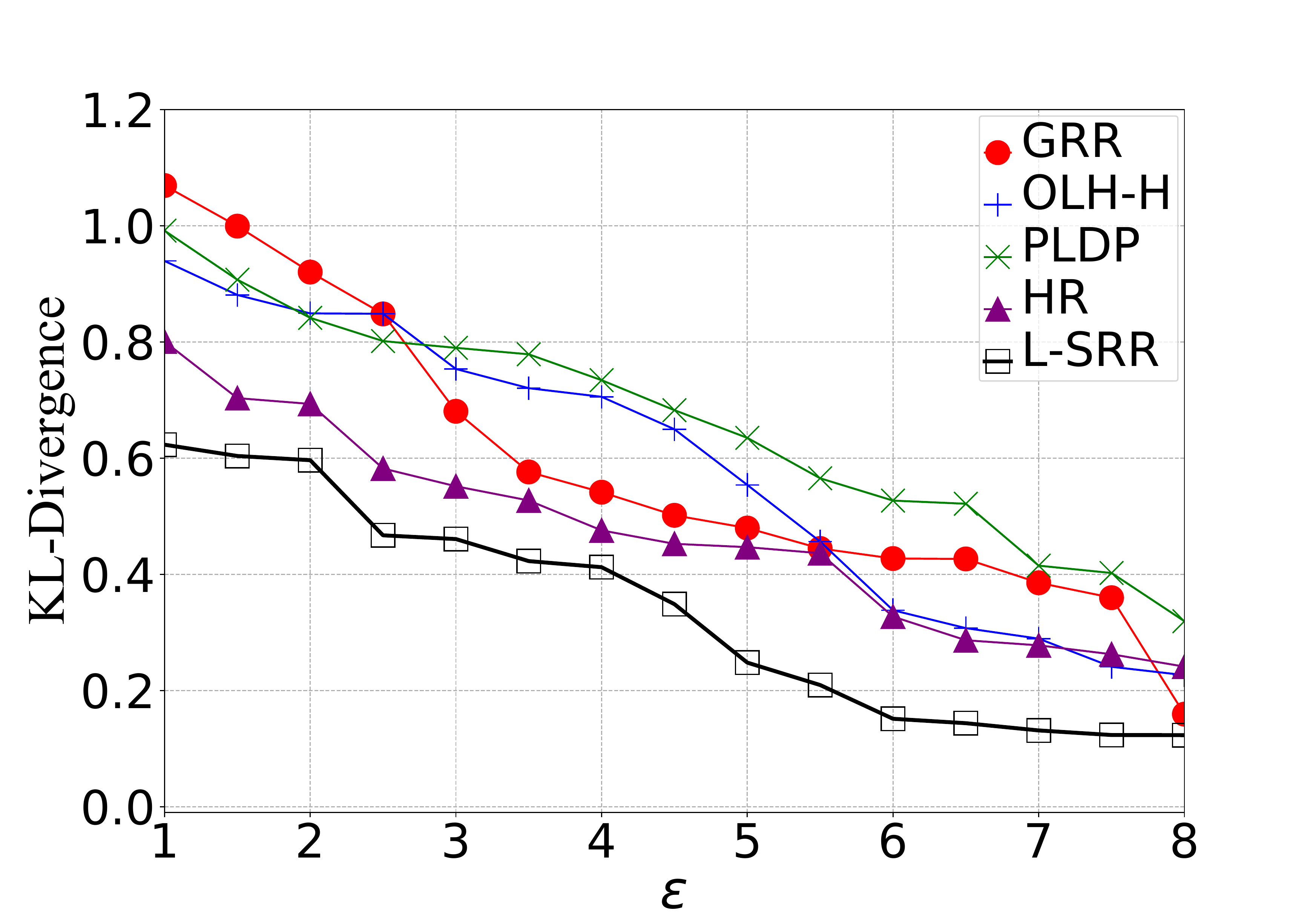}
		\label{fig:KL_G} }
		\hspace{-0.25in}
	\subfigure[Portocabs]{
		\includegraphics[angle=0, width=0.26\linewidth]{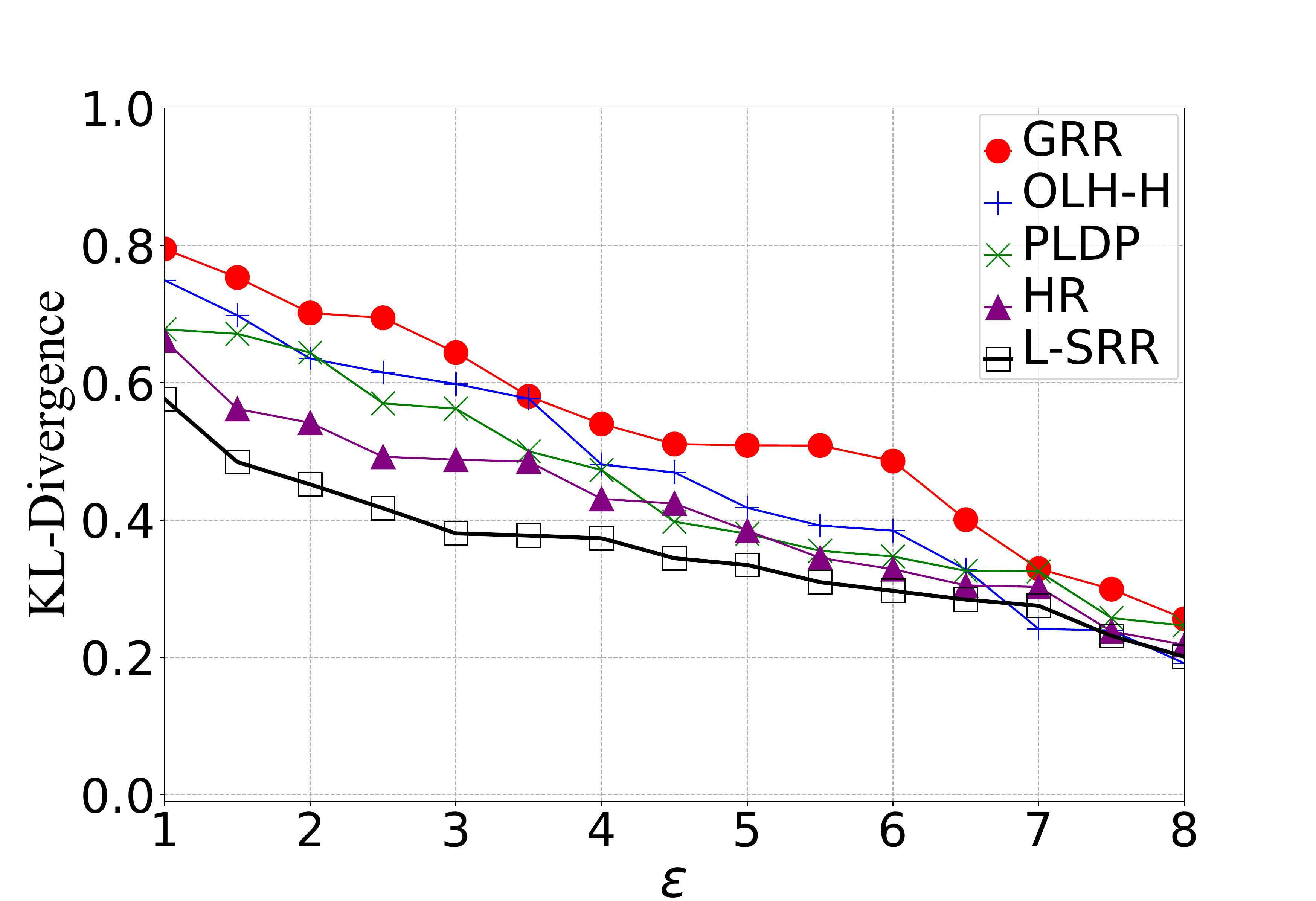}
		\label{fig:KL_P}}
		\hspace{-0.25in}
	\subfigure[Foursquare]{
		\includegraphics[angle=0, width=0.26\linewidth]{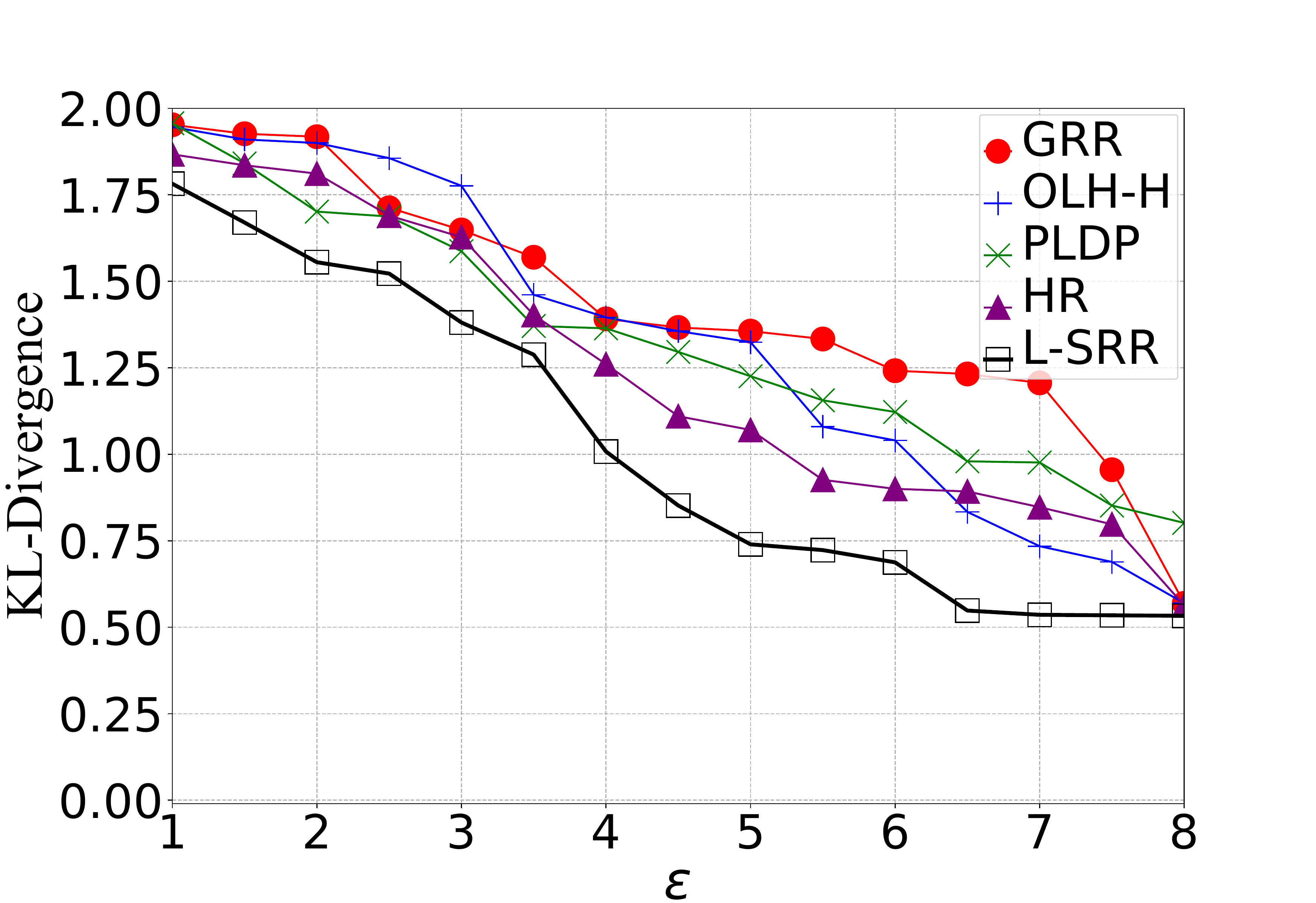}
		\label{fig:KL_S}}\vspace{-0.1in}
	\caption{Average $L_1$-distance and KL-divergence for the distribution estimation on four datasets using different LDP schemes}\vspace{-0.15in}
	\label{fig:distance}
\end{figure*}

 \begin{figure*}[!tbh]
	\centering
	\subfigure[Gowalla]{
		\includegraphics[angle=0, width=0.26\linewidth]{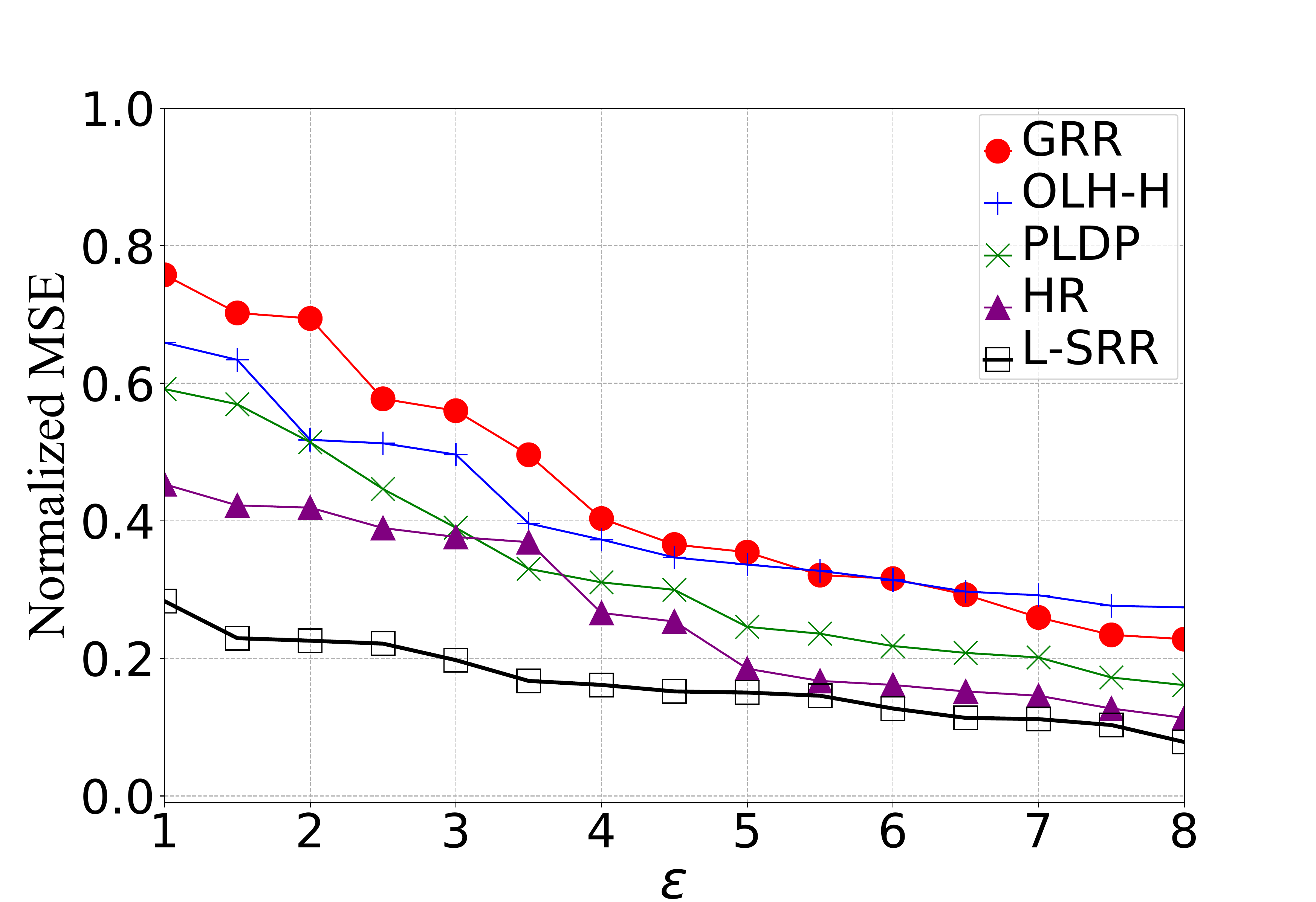}
		\label{fig:N_C} }
		\hspace{-0.25in}
	\subfigure[Geolife]{
		\includegraphics[angle=0, width=0.26\linewidth]{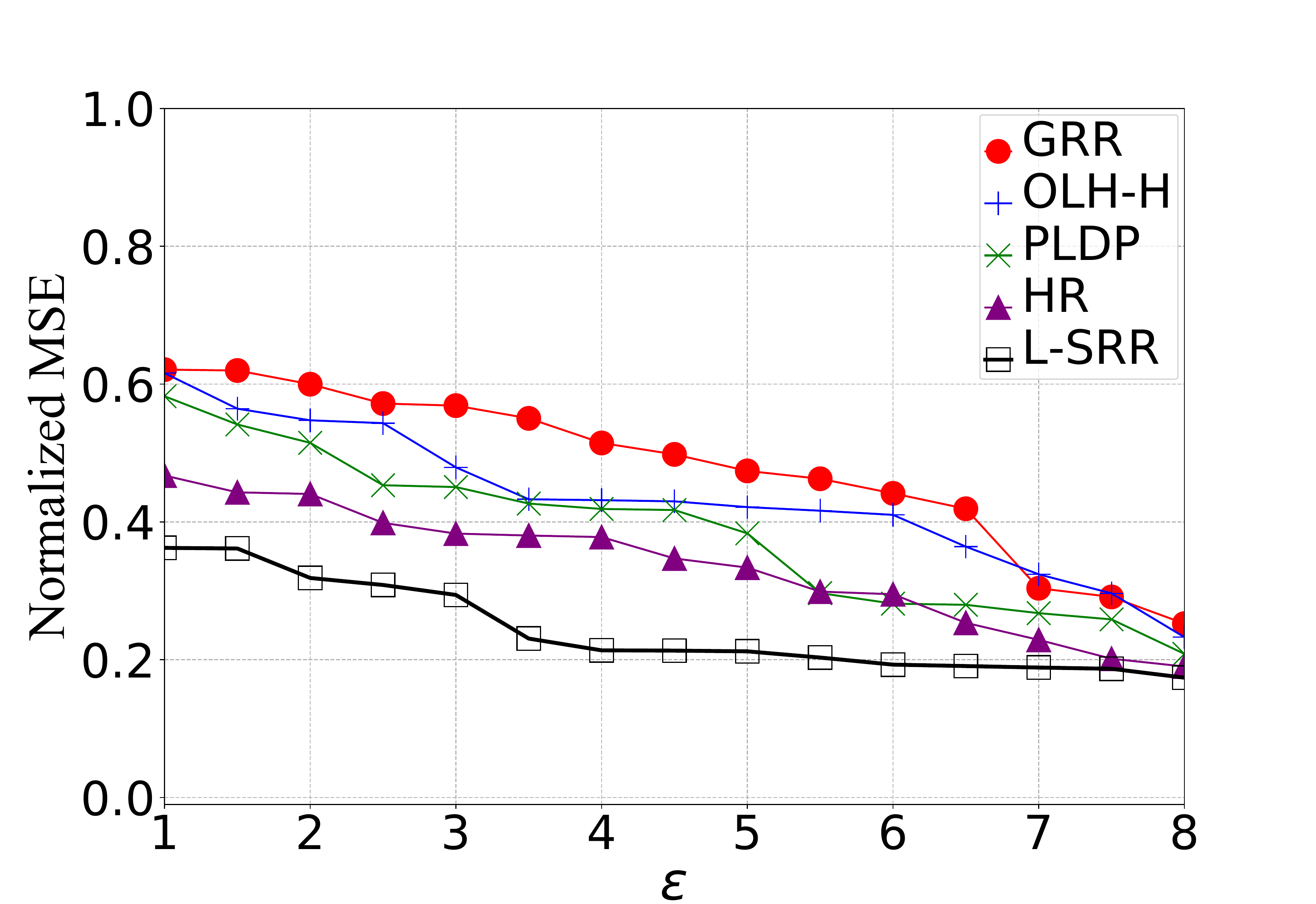}
		\label{fig:N_G} }
		\hspace{-0.25in}
	\subfigure[Portocabs]{
		\includegraphics[angle=0, width=0.26\linewidth]{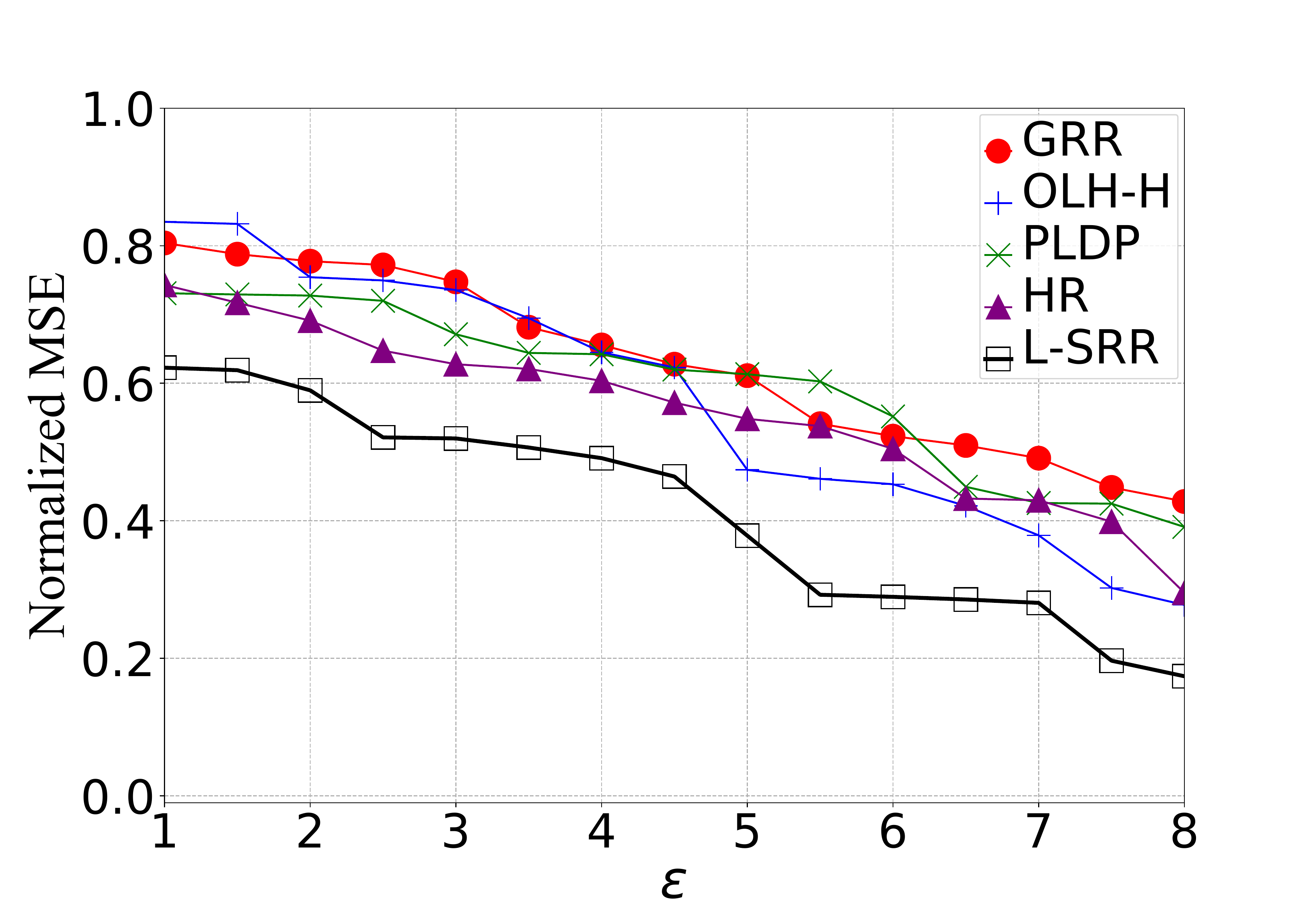}
		\label{fig:N_P}}
		\hspace{-0.25in}
	\subfigure[Foursquare]{
		\includegraphics[angle=0, width=0.26\linewidth]{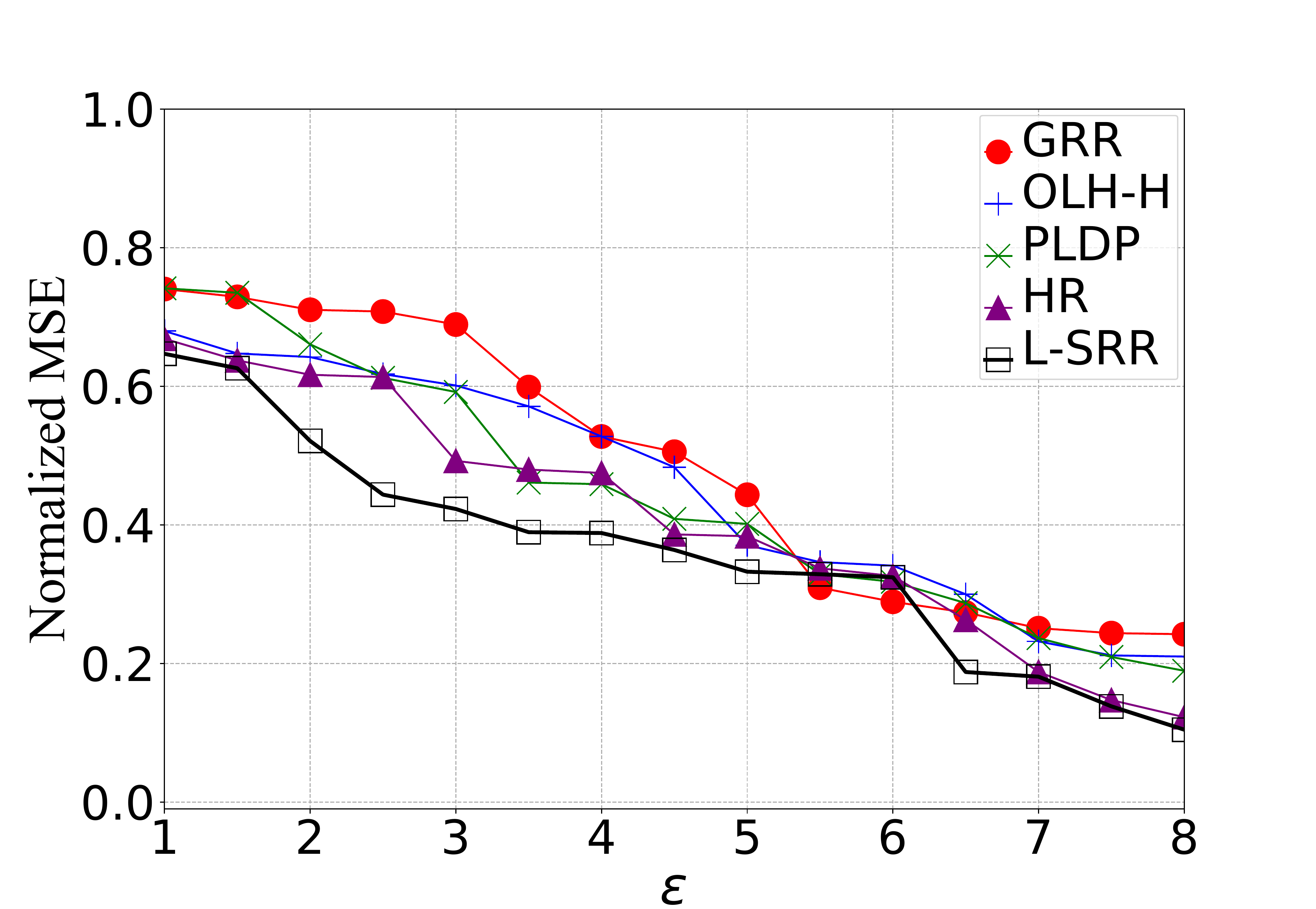}
		\label{fig:N_S}}\vspace{-0.1in}
	\caption{MSE of all the locations' $k$-NN lists on four datasets using different LDP schemes ($k=25$)}
	\label{fig:nearest}
\end{figure*}

%\vspace{0.05in}

%% file: exp.tex
\subsection{Experimental Setting}
\noindent\textbf{Experimental Datasets}. We conduct our experiments on four real-world location datasets.

\begin{itemize}

\setlength\itemsep{0em}
    \item 
\emph{Gowalla Dataset} \cite{Gowalla} collects $6,442,890$ check-ins records of $196,591$ users in Austin, USA via the social network app Gowalla between 02/2009 and 10/2010.  

\item\emph{Geolife Dataset} \cite{Geolife} 
collects 17,621 GPS trajectories of 182 users in Beijing between 04/2007 and 08/2012. 

\item\emph{Portocabs Dataset} \cite{cabs} collects the GPS trajectories of 441 taxis in Porto between 07/2013 and 06/2014. 

\item\emph{Foursquare Dataset} \cite{Four} collects $90,048,627$ check-in locations of $2,733,324$ users in New York City, USA.
\end{itemize}

Since each of the four datasets is collected from locations within a city, we focus on a large geographical region covering a $40\times30$km$^2$ area for each dataset. Only the reported locations in this area are considered as the domain. Since the encoded bit strings for all the locations in each dataset share a 20-bit common prefix, the last 26 bits (out of 46 bits for $h=23$) could sufficiently index all the locations with high accuracy for all the 4.7m$\times$4.7m regions (removing the common prefixes does not affect the accuracy due to fixed domain size and groups). All the experiments were performed on the NSF Chameleon Cluster with Intel(R) Xeon(R)Gold 6126 2.60GHz CPUs and 192G RAM \cite{keahey2020lessons}. Docker is used to start containers to emulate the server/clients with system and network setup.

\vspace{0.05in}

\noindent\textbf{Dataset Characteristics.}
 \label{sec:data}Table \ref{table:ChaLo} presents the number of locations and users in four datasets. The total user number can vary from $30,000$ to $1$M. As we know, infrequent locations in the LDP can cause more utility loss than frequent locations \cite{ldpusenix17}. So, we use four dataset that  have different densities of users. Figure \ref{fig:Org} presents the original frequencies of all the locations in four datasets.
 \vspace{-0.1in}
 
 \begin{table}[!h]%\small
\caption{Characteristics of datasets (after pre-processing)}
\vspace{-0.1in}
\begin{center}
\begin{tabular}{|c|c|c|}
\hline
		\textbf{Dataset} &  \textbf{Location $\#$} & \textbf{User $\#$} \\
		\hline
		 	Gowalla  &1,738 & 1,120,147 \\\hline
			Geolife  &566 & 104,488 \\\hline
			Portcabs  &374 & 34,438 \\\hline
	 	Foursquare  &3,202 & 701,528  \\\hline
\end{tabular}
\label{table:ChaLo}
\end{center}\vspace{-0.15in}
\end{table}
 
\begin{table*}[!h]
 \centering
 \caption{Precision and recall for the derived $k$-NNs of all the users ($k$=25)}
 \vspace{-0.1in}
  \begin{tabular}{|c|c|cc|cc|cc|cc|cc|}
  \hline
 \multirow{2}{*}{Dataset} & 
  \multirow{2}{*}{$\epsilon$} &
 \multicolumn{2}{c|}{\texttt{GRR}} & 
 \multicolumn{2}{c|}{\texttt{OLH-H}}  & 
 \multicolumn{2}{c|}{\texttt{PLDP}}  &  \multicolumn{2}{c|}{\texttt{HR}} &
 \multicolumn{2}{c|}{\texttt{L-SRR}}\\ \cline{3-4} \cline{5-6}\cline{7-8}\cline{9-10}\cline{11-12}
  & & Precision &  Recall & Precision &  Recall  &   Precision &  Recall & Precision &  Recall  & Precision &  Recall\\
  \hline
& 1 & 31.9\% & 47.6\% & 27.1\%  & 35.6\% & 38.1\% & 46.4\% & 53.2\% &63.1\% & \textbf{60.7\%} & \textbf{69.4\%}\\

 & 3 & 55.5\% & 54.1\%& 30.6\% & 38.9\% & 50.1\% & 56.9\% & 66.8\% &74.7\% & \textbf{68.7\%} & \textbf{77.4\%}\\

Gowalla & 5 & 63.5\% &  66.2\% &51.0\% & 57.2\% &  67.6\% & 74.3\% & 68.3\% & 75.8\% & \textbf{73.6\%} & \textbf{78.1\%}\\

 & 7 & 78.4\% & 81.2\%  & 68.8\% & 73.3\% & 73.9\% & 79.5\% & 75.4\% &80.9\% & \textbf{80.1\%} & \textbf{81.9\%}\\

 & 9 & 86.1\% &  87.2\% & 69.2\% & 74.4\% & 80.3\% & 85.3\% & 82.1\% & 84.1\% &\textbf{87.7\%} & \textbf{89.3\%}\\

 \hline
& 1 &  17.8\% &   26.4\% & 30.1\% & 34.2\% & 33.4\% & 34.2\% & 30.8\% & 35.3\% & \textbf{35.2\%} &	\textbf{39.9\%}\\

 & 3 &35.0\% &43.6\%  &42.4\% &49.1\% & 48.7\% & 54.5\% & 50.4\% & 53.1\% &	\textbf{51.6\%} & \textbf{58.7\%} \\

Geolife & 5 & 53.4\% & 60.3\% &	60.5\% & 65.9\% &	69.9\% &74.9\% &	67.1\% & 68.8\% &	\textbf{78.3\%} & \textbf{83.3\%}\\

 & 7 & 78.6\% & 82.9\% & 73.1\% &76.5\% & 77.0\% & 80.7\% &	76.4\% &78.1\% & \textbf{85.9\%} & \textbf{88.9\%}\\

 & 9 & 91.4\% &93.0\% & 89.8\% &90.8\% &	90.2\% & 93.8\% &	90.8\% &92.2\% &	\textbf{92.7\%}
 & \textbf{94.2\%}\\

      \hline
& 1 & 41.9\% &50.7\% & 30.4\% &40.4\% &	48.8\% &58.4\% & 51.2\% &58.4\% &	\textbf{56.2\%} & \textbf{64.1\%}\\

 & 3 & 63.8\% &72.6\% & 43.6\% & 50.6\% & 55.6\% & 63.3\% & 57.7\% & 63.1\% &\textbf{68.3\%} & \textbf{75.8\%}\\
 
Portocabs & 5 & 70.5\% & 78.2\% &	61.9\% & 66.9\% & 70.6\% & 76.1\% &	59.7\% &65.0\% & \textbf{77.4\%} & \textbf{83.8\%}\\

 & 7 & 87.8\% & 93.3\% & 66.0\% & 69.4\% & 76.2\% & 81.3\% & 84.9\% & 88.2\% &	\textbf{92.7\%} & \textbf{98.1\%}\\

 & 9 & 93.4\% & 98.7\% &	86.5\% &89.5\% & 86.7\% & 89.2\% &	91.6\% & 93.3\% &	\textbf{95.9\%} & \textbf{98.9\%}\\

\hline
& 1 & 32.2\% &40.9\% &	42.2\% &52.1\% &	46.6\% &56.1\% & 52.2\% & 60.7\%
&	\textbf{55.7\%} & \textbf{65.3\%}\\

 & 3 & 58.8\% &65.2\%&	50.1\% &57.4\%&	50.1\% &58.0\%&	59.1\% &67.3\%&	\textbf{67.1\%} &\textbf{75.3\%}\\
 
Foursquare & 5 & 80.7\% & 84.6\% & 64.6\% & 68.6\% & 68.1\% & 75.7\%&	80.6\% & 86.9\% & \textbf{83.9\%} & \textbf{87.2\%}\\

 & 7 & 87.1\% & 89.7\% & 65.4\% &69.1\% &	68.7\% &76.3\%&	85.4\% &88.9\% &	\textbf{87.2\%} &\textbf{91.3\%}\\

 & 9 & 88.1\% & 92.3\% & 76.1\% &77.4\% & 82.6\% &86.9\% & 86.1\% &89.1\% &	\textbf{91.3\%} & \textbf{94.6\%}\\
  \hline
  %\bottomrule
  \end{tabular}\vspace{-0.15in}
  \label{tab:NN}
\end{table*}

\subsection{Distribution Estimation (Location-Input)}
\label{sec:dis}
We first evaluate the utility of \texttt{L-SRR} for the distribution estimation while benchmarking with the state-of-the-art LDP schemes, including Generalized Randomized Response \cite{GRR} (\texttt{GRR}), Optimal Local Hash with hierarchy structure \cite{TianhaoMultiLDP} (\texttt{OLH-H}), the Location Data Aggregation \cite{location_LDP} (\texttt{PLDP}), and the Hadamard Response (\texttt{HR}) \cite{emp}. We follow the original perturbation and estimation method in each benchmark. Here, we choose the OLH mechanism since it has better utility than unary encoding (\texttt{UE}), and choose the existing location LDP framework \texttt{PLDP} instead of existing location framework in \cite{loc_LDP3, loc_LDP2} since \texttt{PLDP} is an optimized framework that boosts the utility. For fair comparisons, in \texttt{OLH-H}, we randomly sample a hierarchical level for each location. Then, we adapt the constrained inference \cite{consistency} to adjust the frequencies of parent and leaf nodes for consistency. In \texttt{PLDP}, we assign the same protection region level for all the users as other LDP schemes to satisfy the strict $\epsilon$-LDP.  

The server derives the spatial density for many LBS applications, e.g., urban traffic density \cite{traden}, and crowd density for events \cite{croden}. In most existing LDP settings, the $\epsilon$ is in the range between $0.5$ to $10$ for privacy protection. Too large $\epsilon$ value can't protect user's location well. Similar to that, we set $\epsilon$ between 1 and 8 with a step of 0.5 (covering both strong and weak privacy regime).

Figure \ref{fig:distance} shows the average $L_1$-distance and KL-divergence between the true and estimated distributions of all the locations. Both $L_1$-distance and KL-divergence decrease as $\epsilon$ increases. Especially for the \texttt{GRR}, the error dramatically decreases (e.g., Figure \ref{fig:KL_C}) since the probability grows exponentially. However, \texttt{L-SRR} still greatly outperforms other LDP schemes on all the four datasets. 

\subsection{Case Study I: $k$-NN Query (Location-Input)} 
\label{sec:case}
We first evaluate the performance of \texttt{SRR} in specific applications on recommendations based on the location distribution. $k$-nearest neighbors ($k$-NNs) is a typical application in which queries can be made for the nearest point-of-interests or users. We next show the results for querying the $k$-NN users \cite{KNN}, which can be extended from the distribution estimation. 
The $k$-NN lists for any user (with a location) are the other $k$ closest users, measured by the MSE of their coordinates. Then, given the estimated location distribution, the server can directly derive each location's list of $k$-NNs.

\begin{figure*}[!tbh]
	\centering
	\subfigure[Gowalla]{
		\includegraphics[angle=0, width=0.26\linewidth]{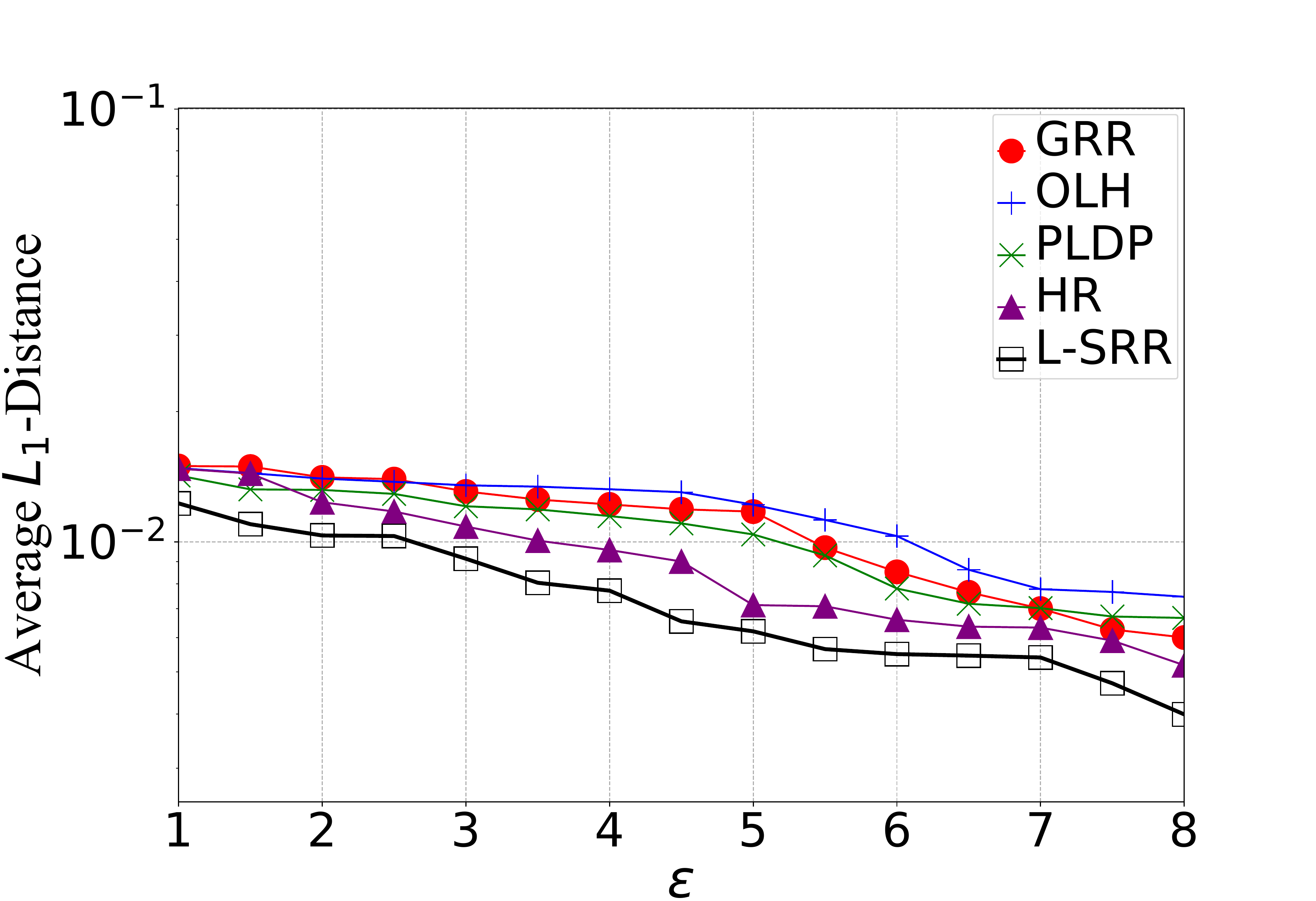}
		\label{fig:od_C} }
		\hspace{-0.25in}
	\subfigure[Geolife]{
		\includegraphics[angle=0, width=0.26\linewidth]{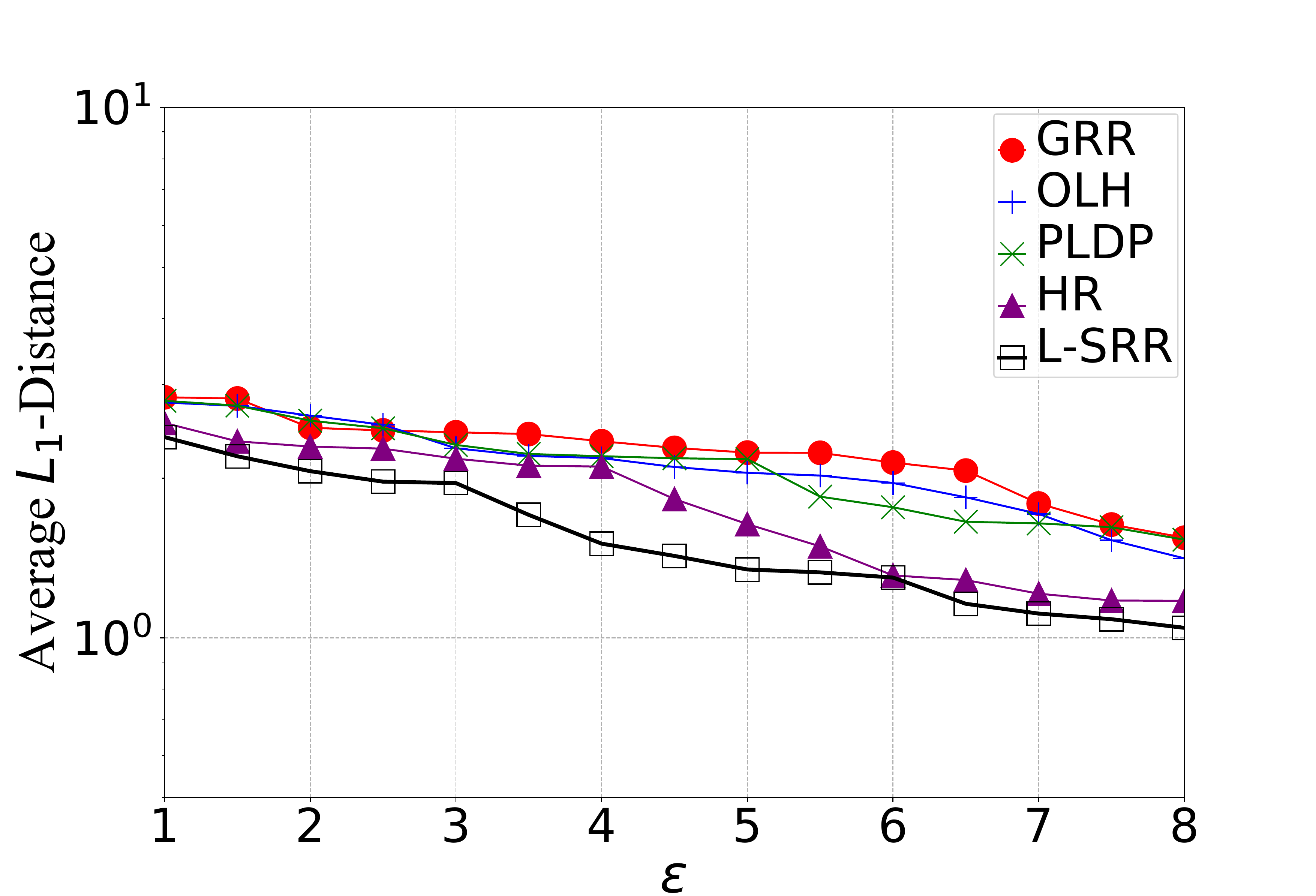}
 		\label{fig:od_G} }
		\hspace{-0.25in}
	\subfigure[Portocabs]{
		\includegraphics[angle=0, width=0.26\linewidth]{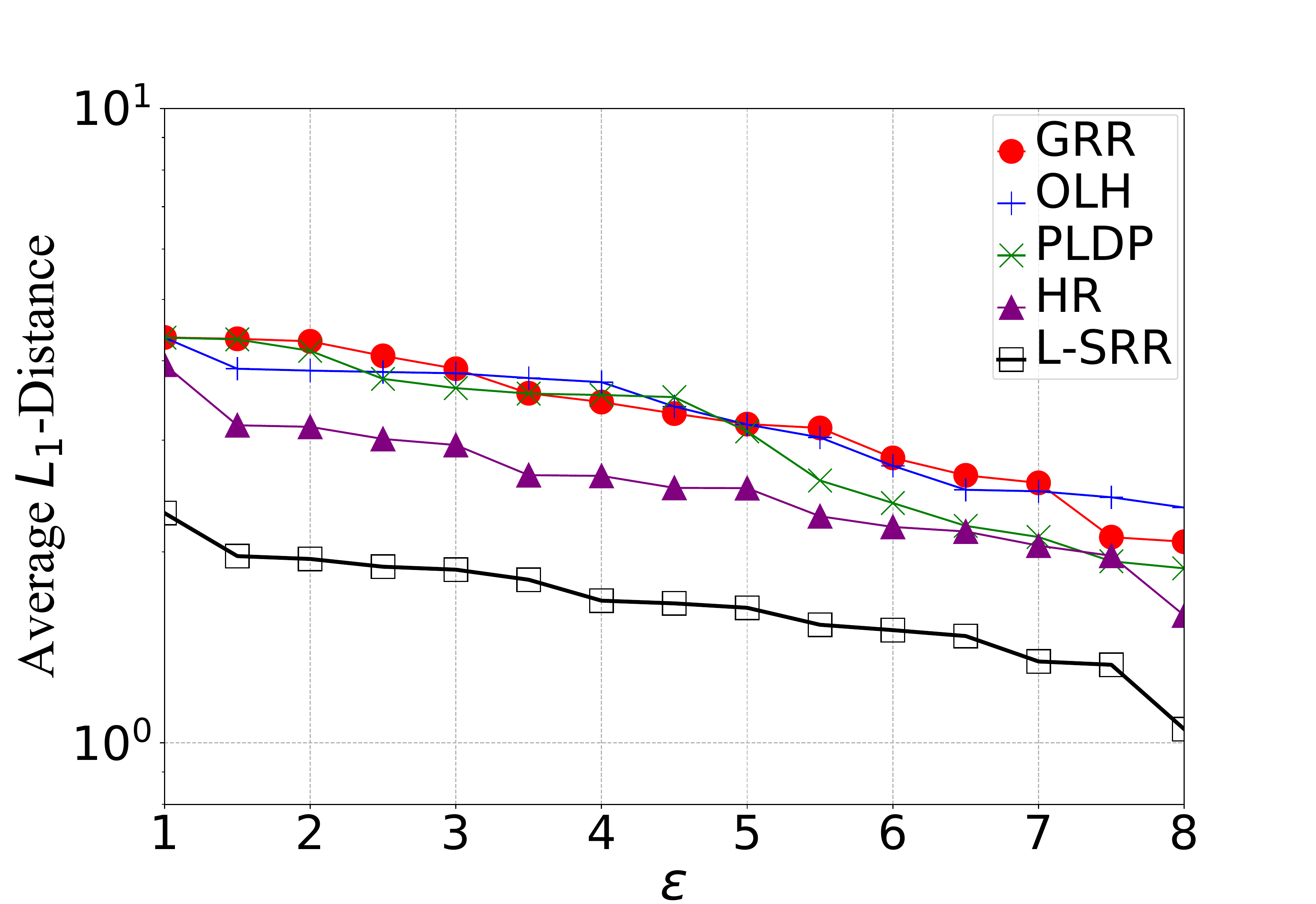}
		\label{fig:od_P}}
		\hspace{-0.25in}
	\subfigure[Foursquare]{
		\includegraphics[angle=0,width=0.26\linewidth]{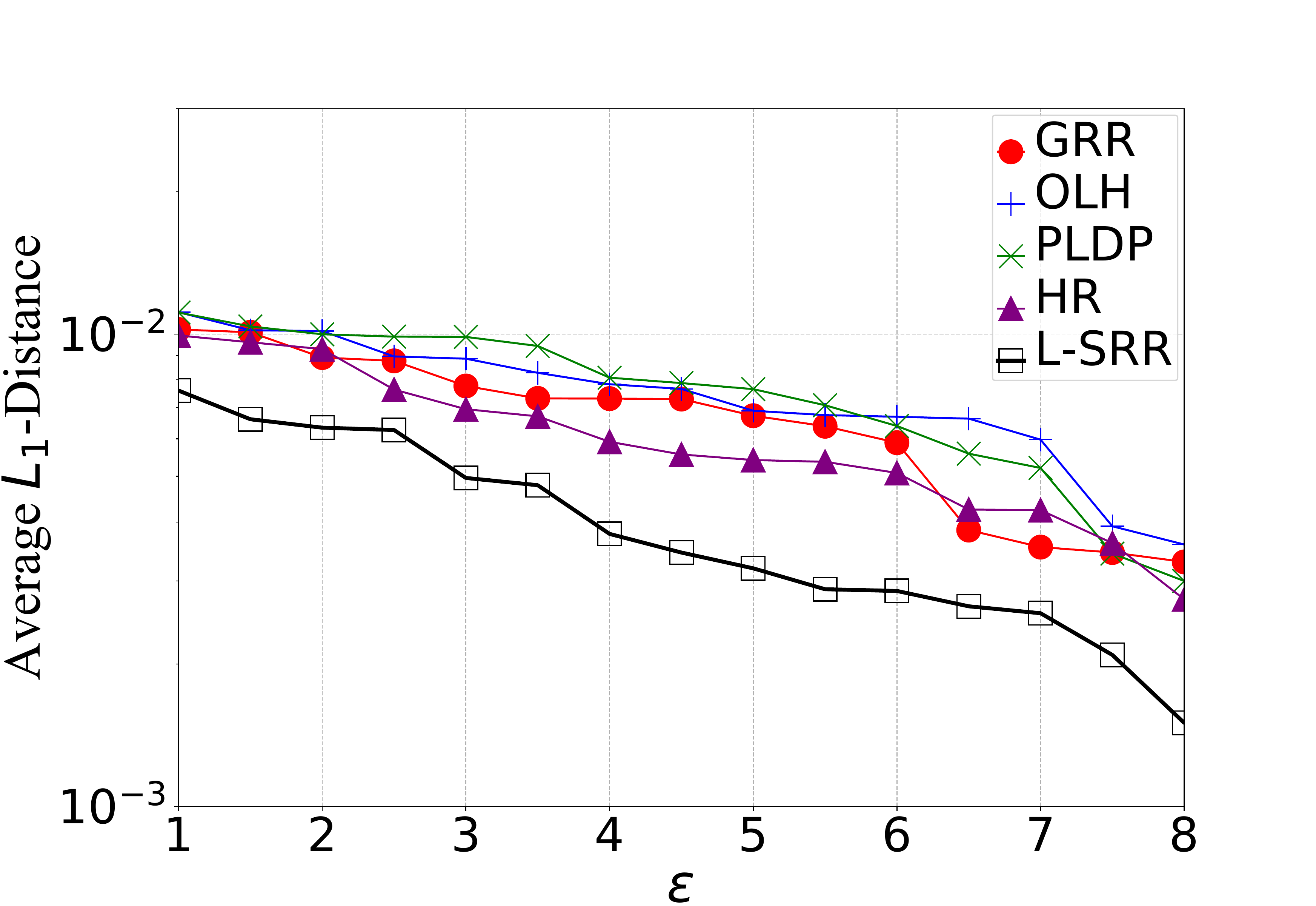}
		\label{fig:od_S}}\vspace{-0.1in}
	\caption{Average $L_1$-distance for the OD pair frequency on four datasets using different LDP schemes}\vspace{-0.1in}
	\label{fig:OD}
\end{figure*}

\begin{figure*}[!tbh]
	\centering
	\subfigure[Gowalla]{
		\includegraphics[angle=0, width=0.26\linewidth]{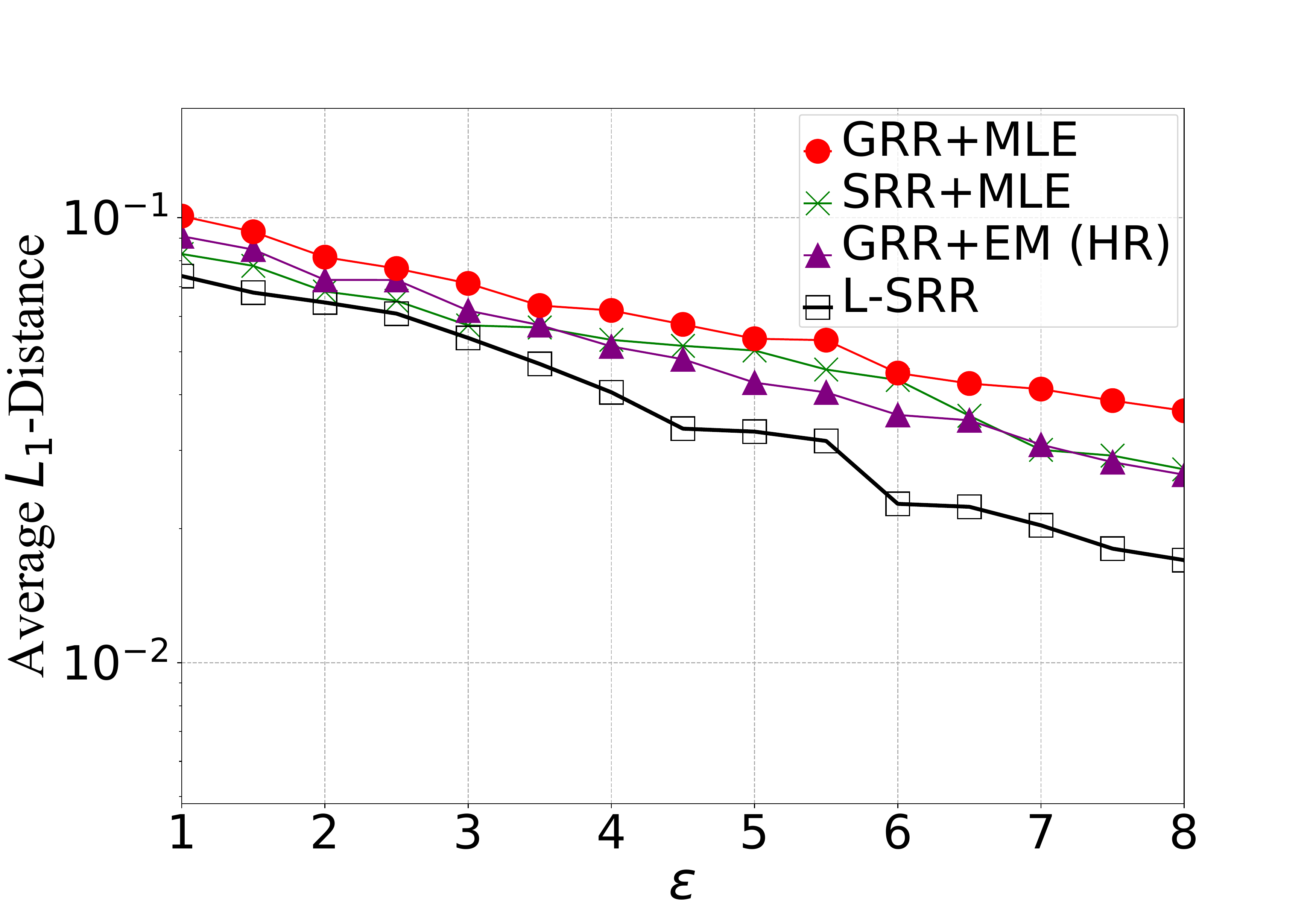}
		\label{fig:ab_C} }
		\hspace{-0.25in}
	\subfigure[Geolife]{
		\includegraphics[angle=0, width=0.26\linewidth]{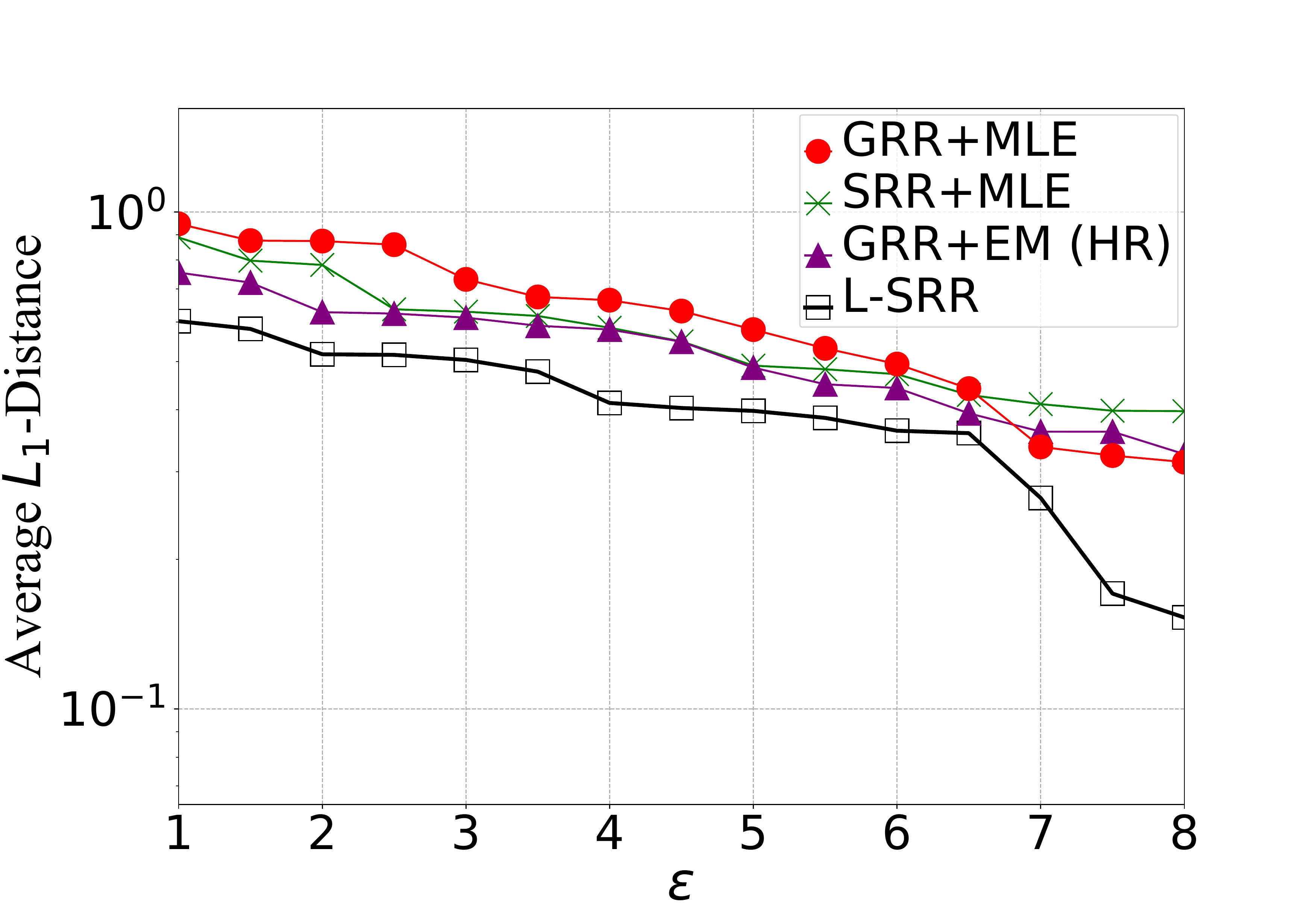}
 		\label{fig:ab_G} }
		\hspace{-0.25in}
	\subfigure[Portocabs]{
		\includegraphics[angle=0, width=0.26\linewidth]{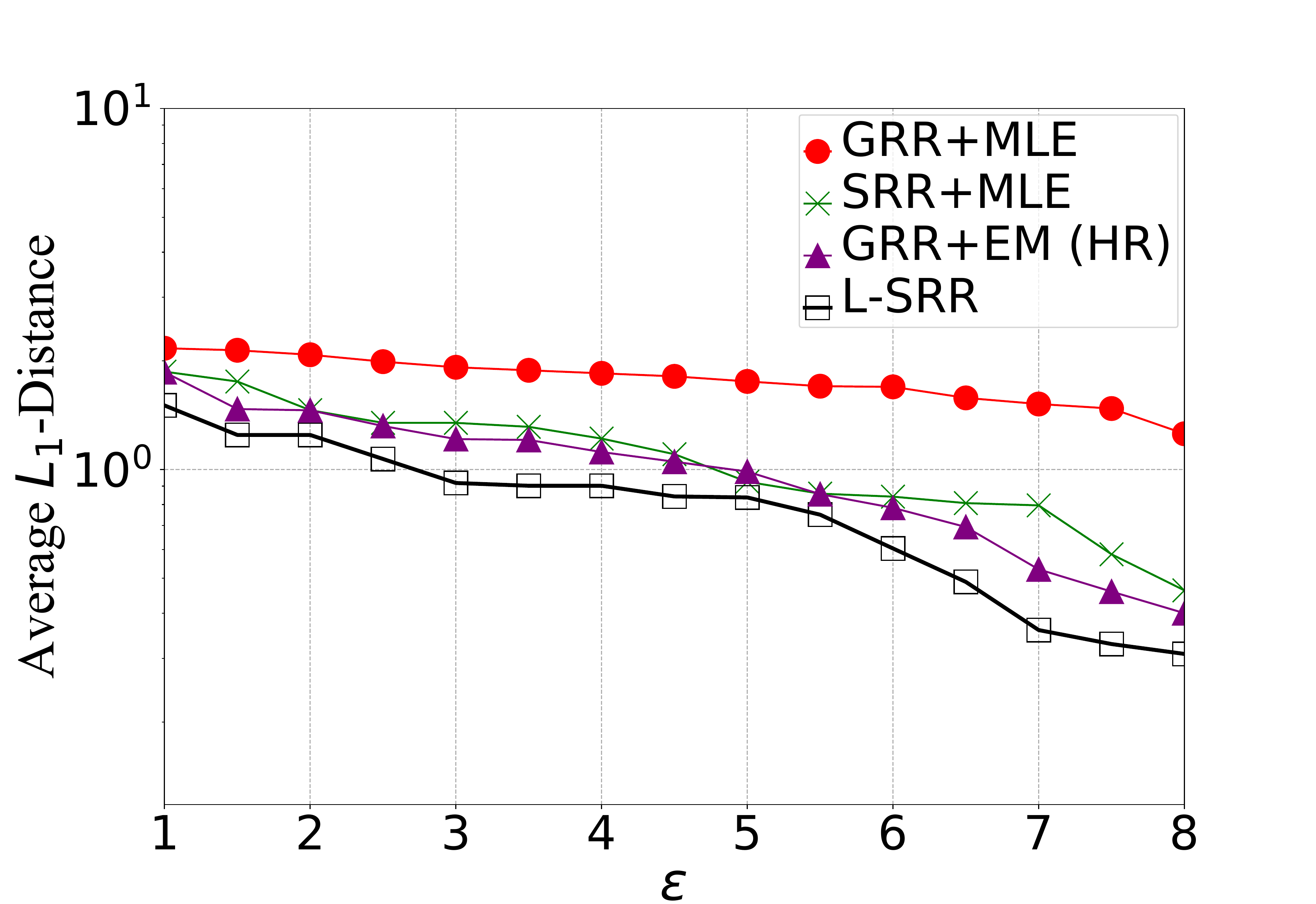}
		\label{fig:ab_P}}
		\hspace{-0.25in}
	\subfigure[Foursquare]{
		\includegraphics[angle=0,width=0.26\linewidth]{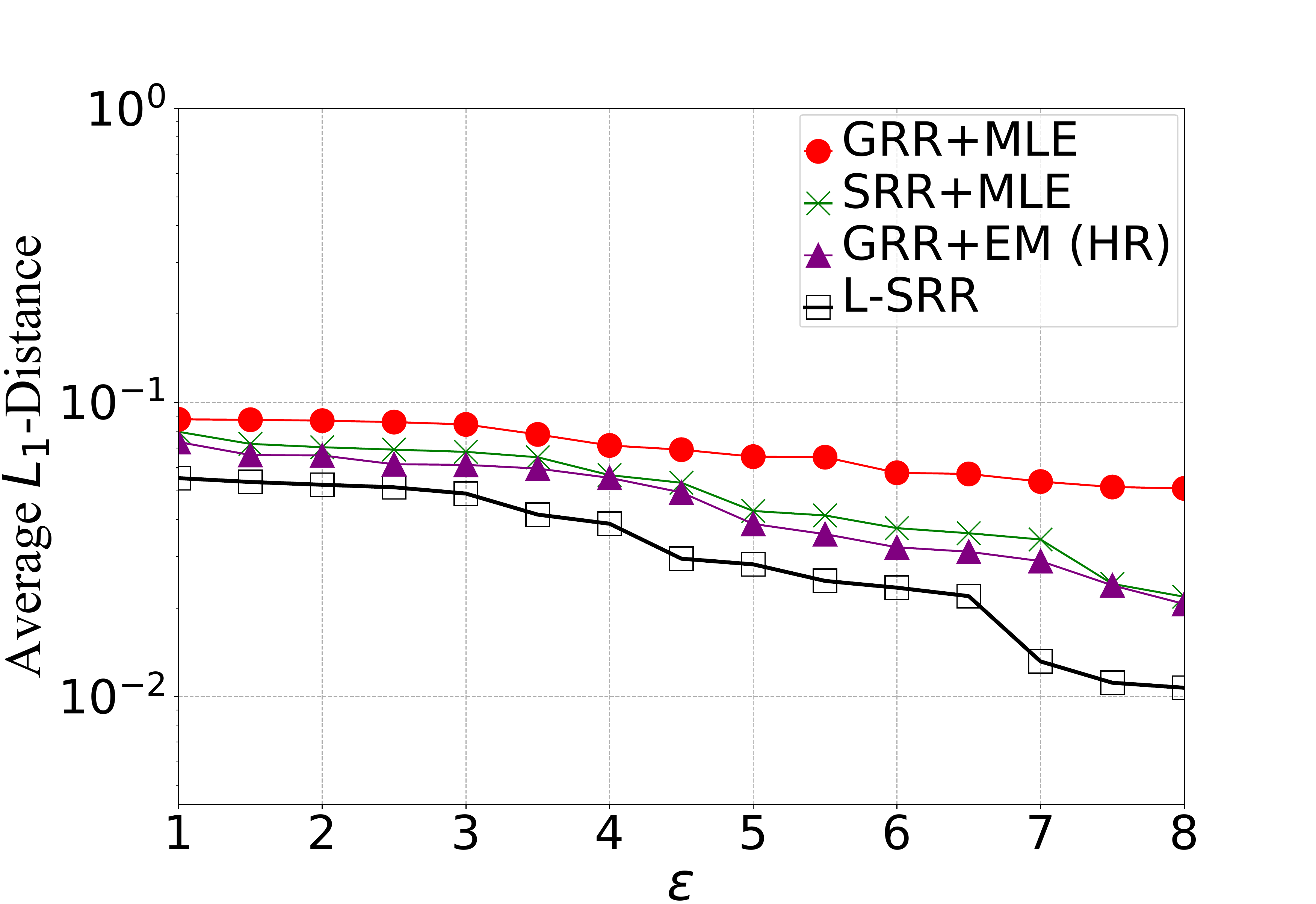}
		\label{fig:ab_S}}\vspace{-0.1in}
	\caption{Average $L_1$-distance for frequency estimation using different combinations of perturbation and estimation methods}\vspace{-0.1in}
	\label{fig:Abstudy}
\end{figure*} 

\vspace{0.05in}

\noindent\textbf{$k$-NN Lists Computed by Server}. Figure \ref{fig:nearest} shows the normalized MSE between the true and estimated coordinates of all the users' $k$-NN lists. 
The normalized MSE also decreases while $\epsilon$ increases. In Figure \ref{fig:N_C}, \ref{fig:N_G}, \ref{fig:N_P}, and \ref{fig:N_S}, \texttt{L-SRR} outperforms \texttt{GRR}, \texttt{OLH-H}, \texttt{PLDP}, and \texttt{HR}, which is consistent with the previous results.

We also present the \emph{precision} and \emph{recall} of all the users' estimated $k$-NN lists in Table \ref{tab:NN}. Again, \texttt{L-SRR} can produce more accurate $k$-NN lists than all the other LDP schemes. Note that $\epsilon$ might be relatively large for very high accuracy (e.g., $\epsilon=5$ similar to the privacy setting by Apple \cite{Apple}). If involving more users in the practical LBS App, $\epsilon$ can be much smaller for such very high accuracy. 

\vspace{-0.1in}

\subsection{Case Study II: Trajectory-Input LBS}
\label{multi-exp}
We next evaluate the performance of \texttt{L-SRR} on collecting trajectories for two example LBS applications: (1) origin and destination (OD) analysis which estimates the OD pairs frequencies with the Lasso regression, and (2) traffic-aware GPS navigation. 

\vspace{0.05in}

\noindent\textbf{OD Analysis}. The true number of distinct OD pairs in four datasets are $2,315$, $876$, $1,034$, and $5,634$, respectively. We apply the same Lasso regression algorithm to all the LDP schemes. Figure \ref{fig:OD} presents the average $L_1$-distance between the true and estimated OD pair distribution.  
As $\epsilon$ increases, $L_1$-distance decreases. \texttt{L-SRR} again shows the smallest $L_1$-distance of \texttt{L-SRR} in all the experiments. Moreover, we also observe that the $L_1$-distance is smaller than LBS with single-location input (see Figure \ref{fig:distance}).

\vspace{0.05in}

\noindent\textbf{Traffic-Aware GPS Navigation}. To test the performance of the traffic-aware GPS navigation, we make the simulation the experiment of recommendation for the fastest route. We can also draw the conclusion that \texttt{L-SRR} outperforms other LDP schemes (see the detailed results and discussions in Appendix \ref{sec:trajectory}).

\subsection{Ablation Study and Runtime}
\label{sec:ablation}

\begin{figure*}[!tbh]
	\centering
	\subfigure[Gowalla]{
		\includegraphics[angle=0, width=0.26\linewidth]{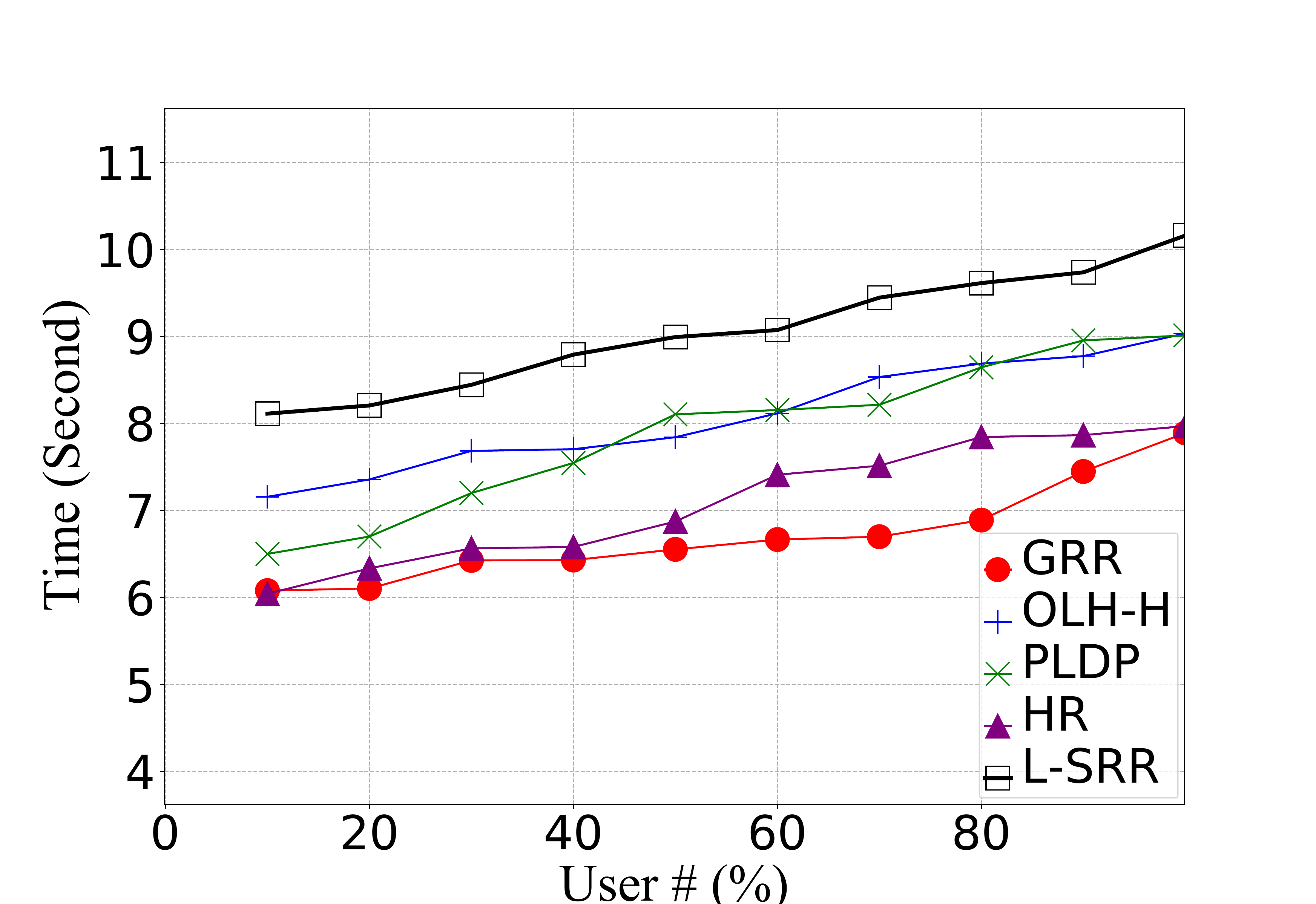}
		\label{fig:Time_C} }
		\hspace{-0.27in}
	\subfigure[Geolife]{
		\includegraphics[angle=0, width=0.26\linewidth]{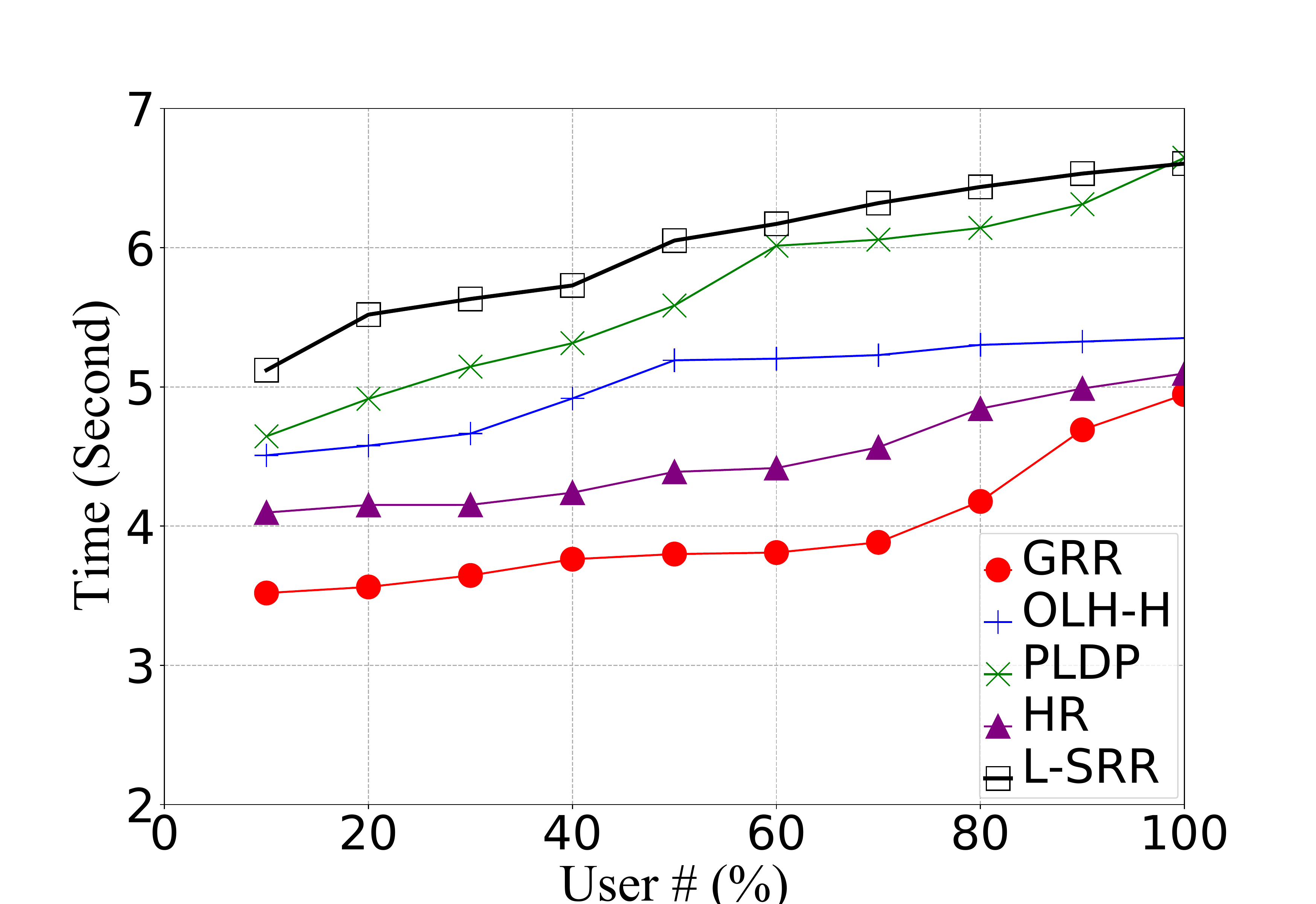}
		\label{fig:Time_G} }
		\hspace{-0.2in}
	\subfigure[Portocabs]{
		\includegraphics[angle=0, width=0.26\linewidth]{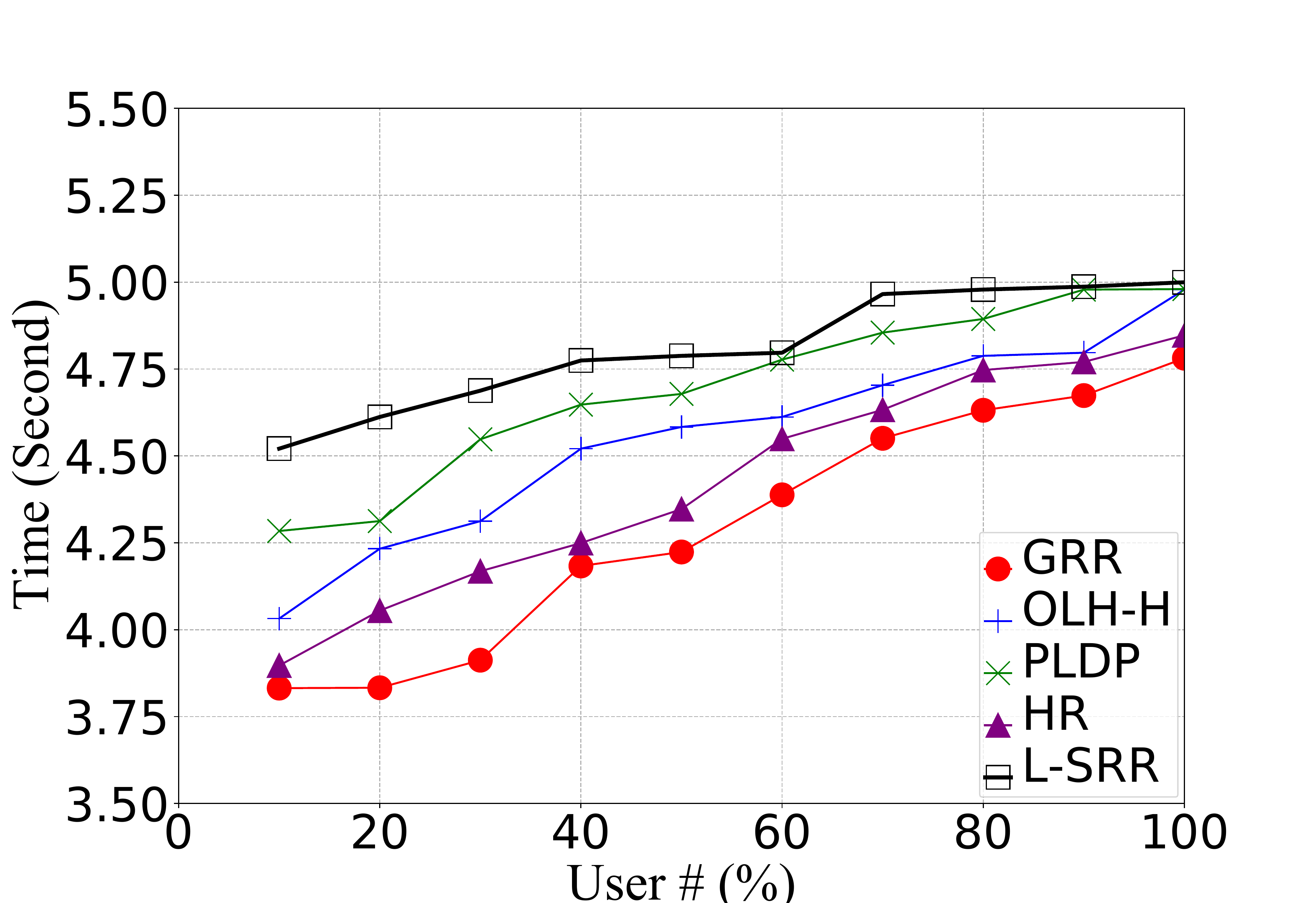}
		\label{fig:Time_P}}
		\hspace{-0.25in}
	\subfigure[Foursquare]{
		\includegraphics[angle=0, width=0.26\linewidth]{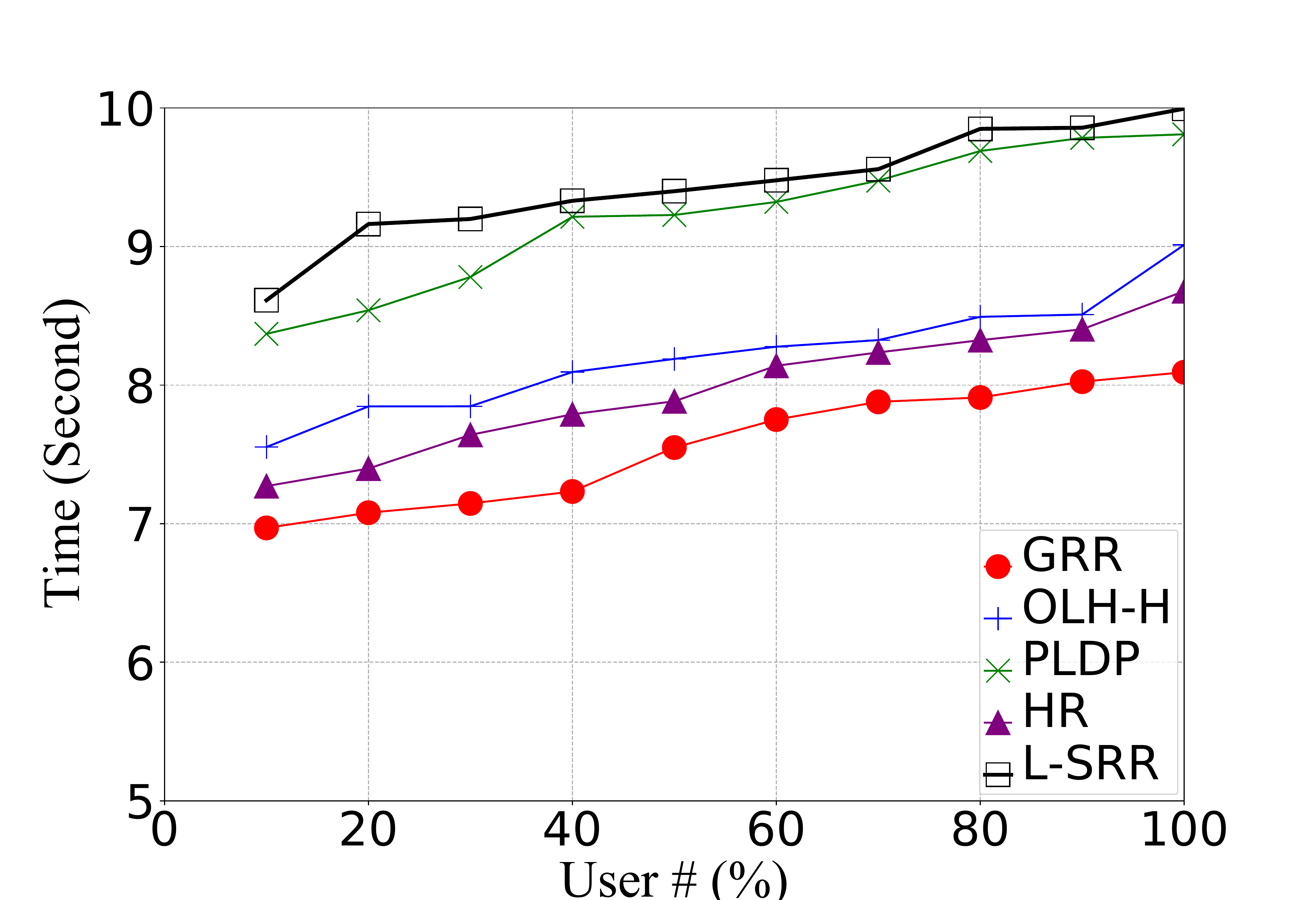}
		\label{fig:Time_S}}\vspace{-0.1in}
	\caption{Runtime for the server (vs. the number of users)}\vspace{-0.15in}
	\label{fig:time}
\end{figure*} 

\noindent \textbf{Ablation Study}. We compare the results with different combinations of perturbation mechanisms (\texttt{GRR}, \texttt{HR} and \texttt{SRR}) and estimation methods.  
Since the standard estimation method cannot be applied to \texttt{SRR} (more than two perturbation probabilities), we apply the maximum likelihood estimation (\texttt{MLE}) instead. 
Moreover, the \texttt{GRR} with empirical estimation (\texttt{EM}) is a special case of Hadamard response (\texttt{HR}): $|C_x|=1$.
Figure \ref{fig:Abstudy} shows that \texttt{SRR} and the revised \texttt{EM} (\texttt{L-SRR}) perform the best. Even with the \texttt{MLE}, \texttt{SRR} is better than \texttt{GRR} in most cases. Also, the revised \texttt{EM} can further boost the utility of \texttt{SRR} (compared to \texttt{SRR} and \texttt{MLE}). 

\vspace{0.05in}
\noindent\textbf{Runtime}. 
Since users only need to perturb their locations, the user-side runtime is negligible. It takes only $0.014$ second for each user on average in the experiments, and thus we only report the server-side runtime in Figure \ref{fig:time}. We test $10\%$ to $100\%$ of each dataset with a step of $10\%$. Similar to \texttt{GRR}, \texttt{OLH-H}, \texttt{PLDP} and \texttt{HR}, the runtime of \texttt{L-SRR} only slightly increases as the number of users reaches $\sim$1M (e.g., $9$ seconds for Gowalla dataset), which is acceptable. 

\vspace{-0.15in}

\begin{figure}[!h]
	\centering
	\subfigure[Runtime vs.$\#$ of users]{
		\includegraphics[angle=0, width=0.5\linewidth]{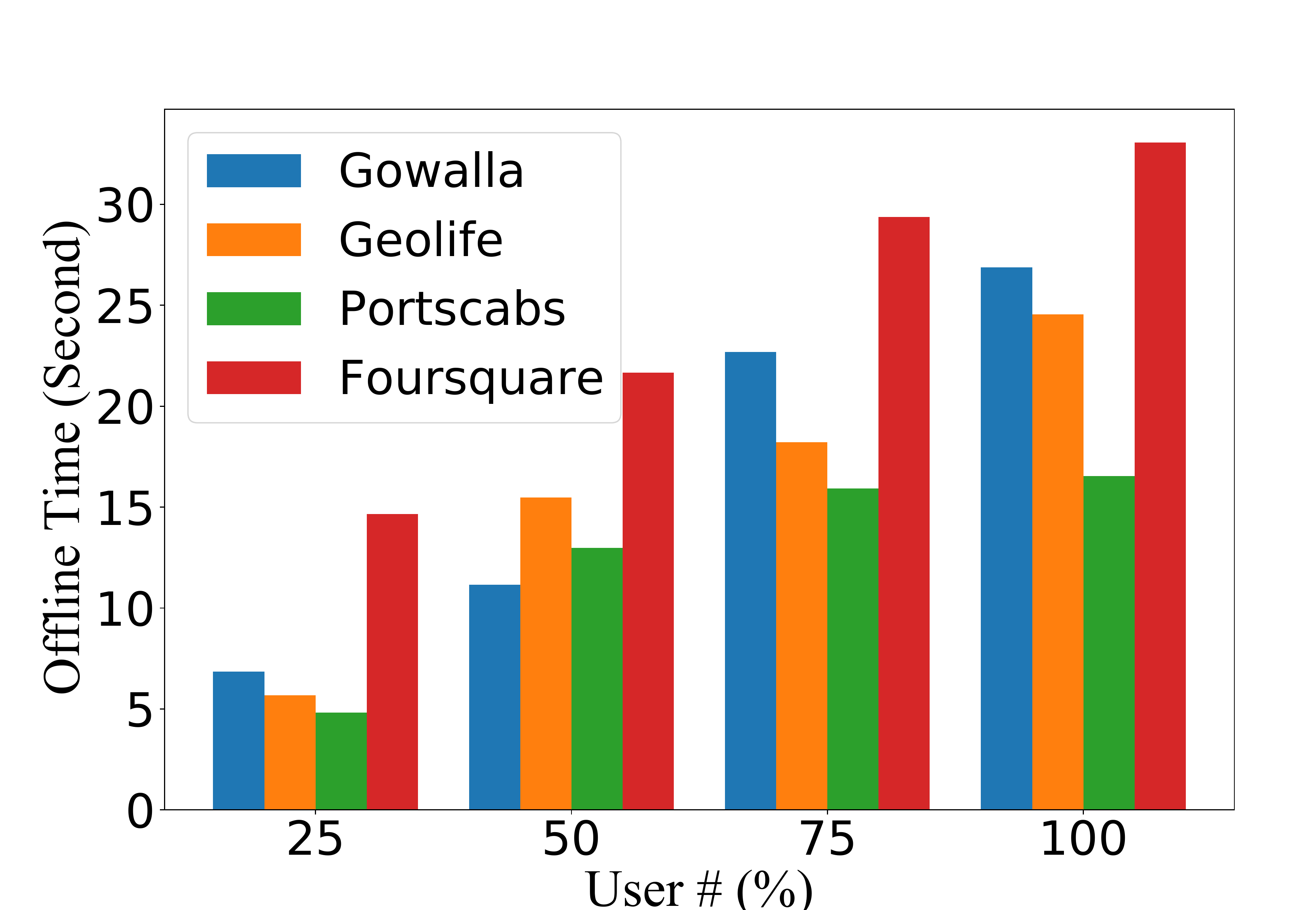}
		\label{fig:off_user} }
		\hspace{-0.18in}
	\subfigure[Runtime vs. $\#$ of location]{		
	\includegraphics[angle=0, width=0.5\linewidth]{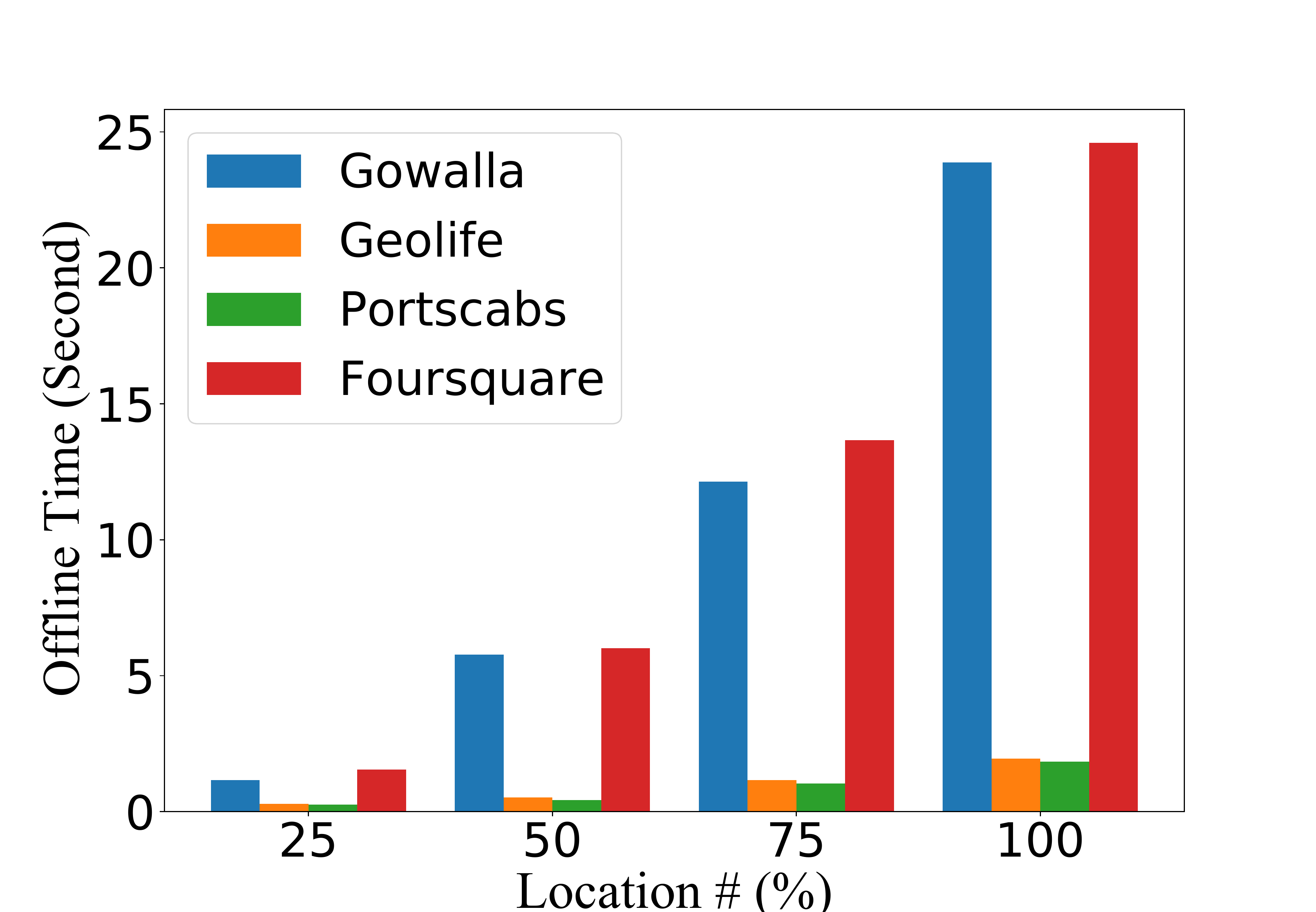}
	\label{fig:off_location} }\vspace{-0.15in}
	\caption{Offline runtime}\vspace{-0.1in}
	\label{fig:off}
\end{figure}

Notice that, group partitioning dominates the offline costs $O(d^2)$ for \texttt{L-SRR}. Thus, we also present such offline partitioning time w.r.t. the number of users and the number of locations, as shown in Figure \ref{fig:off}. We uniformly extract $25\%, 50\%, 75\%$, and $100\%$ of users and locations from each dataset as the test datasets. As shown in Figure \ref{fig:off_user}, the offline time (including the the preprocessing time to get the sub-dataset) increases as the number of users increases due to the growth of distinct locations. Since the group partitioning that is related to the domain size dominates the offline costs. In Figure \ref{fig:off_location}, we also see that the offline runtime (excluding the preprocessing time) grows on the number of locations, and the offline time is around $30$ seconds at most. However, the offline execution is needed when the location domain is updated. Recall that the domain is only updated periodically (e.g., every day). Thus, such offline costs are efficient for real-world deployments.

%% file: related.tex
Many privacy preserving location-based services techniques have been proposed (e.g., \cite{LBS, cloaking1}). $K$-anonymity was first defined to protect privacy for LBS. Dummy locations \cite{LBS} and cloaking region \cite{cloaking1} have been utilized for anonymity. 
However, these methods are highly vulnerable to background knowledge attacks. Another type of techniques design cryptographic protocols \cite{cyber} to securely perform LBS computations. However, both computational costs and communication overheads might be very high. Differenital privacy is a privacy notion that protects the privacy against arbitrary prior knowledge known to the
adversaries \cite{SearchLog}. There are many differential privacy works to solve optimization problems and find optimal mechanisms under different scenarios \cite{R2DP, Ghosh09, Gupte10}.

More recently, rigorous privacy notion differential privacy (DP) has also been applied to LBS \cite{DP1, DP2, liuling17CCS, DPT, OptTrajectory}. For instance, a synthetic data generation method \cite{DP1} was proposed to publish statistics about commuting patterns (including locations) with DP guarantee. Moreover, a quadtree spatial decomposition technique \cite{DP2} has been used to ensure DP in a database with location pattern mining capabilities. However, the DP model may not be suitable to real LBS applications in case that the users do not trust the server.

The emerging LDP models enable private data collection by untrusted server, which provides stronger protection than the centralized DP models. It has been extended to privately collect different types of data (e.g., histogram \cite{rappor14,histogram1}, social graphs \cite{zhan17}, itemsets \cite{ldpusenix17}). Meanwhile, LDP has been successfully deployed in industry (e.g., Google \cite{rappor14}, Apple \cite{Apple}, and Microsoft \cite{boling17}). Recall that two works directly apply randomized response and unary encoding to collect workload-aware indoor positioning data \cite{loc_LDP3} and generate synthetic locations \cite{loc_LDP2} but result in poor utility. Moreover, several relaxed LDP notions have been proposed to protect location privacy \cite{Geo-indistinguishability, location_LDP}. Andrés et al. \cite{Geo-indistinguishability} relaxes the protection for locations within a radius via geo-indistinguishability. Chen et al. \cite{location_LDP} relaxes LDP to \texttt{PLDP} which allows users to specify personalized privacy budgets for private location collection. However, they cannot ensure rigorous LDP and are also less accurate than our \texttt{SRR}.

%% file: concl.tex
 Severe privacy risks arise in LBS applications due to sensitive location collection. To address the deficiency on privately collecting locations with LDP guarantees and high utility, we propose a novel LDP mechanism ``Staircase Randomized Response'' (\texttt{SRR}) and extend the empirical estimation for \texttt{SRR} to significantly improve the accuracy of the LDP model for LBS applications. In addition, we have also extended \texttt{SRR} to privately collect trajectories with $\epsilon$-LDP. We have conducted extensive experiments on real datasets to show that \texttt{L-SRR} drastically outperforms other LDP schemes.

%% file: appendix.tex
\section*{Appendix}

\section{Proof of Convex Property w.r.t. $m$}
\label{sec:convex}

\begin{proof}
With the mutual information bound function $H$, we can take its second order derivative in $m$ as follows:

\vspace{-0.15in}

\small

\begin{align*}
\frac{\partial^2 H }{\partial m^2}
=&[\frac{1}{c \cdot d-\frac{m \cdot (c-1)\cdot d}{2}} \cdot \log \frac{c}{c-\frac{m \cdot (c-1)\cdot d}{2}}]^{''}\\
=&\frac{(c-1)^2d^2}{4(c\cdot d-\frac{(c-1)d}{2}\cdot m)^3}\cdot (2\log\frac{c}{c\cdot d-\frac{(c-1)d}{2}\cdot m}+3)
\end{align*}
\normalsize

When the first order derivative is equal to zero, we have $m=\frac{2 \cdot (c \cdot d - e^{1+\log c})}{(c-1) \cdot d}$. It is very straightforward to prove that the second order derivative is greater than zero since $(cd-\frac{(c-1)d}{2}m) >0$ and $2\log\frac{c}{cd-\frac{(c-1)d}{2}m}+3>0$. Therefore, it is a convex function, and we can derive its minimum value by the derivative.
\end{proof}

\section{Privacy and Utility Analysis}
\label{sec:pri_ut}

\subsection{Proof of Theorem \ref{theorem:epsilondp} (Privacy Analysis)}
\label{sec:ldpproof}
\begin{proof}
     For any pair of input locations $x$, $x' \in \mathcal{D}$ and output $y$, the maximum perturbation probability $q(y|x)$ for sampling location $y$ based on input $x$ is $\alpha_{max}(x)$ when $y$ is in the same group with $x$ (the first group $G_1(x)$); the minimum perturbation probability $q(y|x')$ for sampling location $y$ based on input $x'$ is $\alpha_{min}(x')$ when $y$ in the furthest group for $x'$ (the last group $G_m(x')$). Thus, the \texttt{SRR} mechanism in \texttt{L-SRR} satisfies $\epsilon=\max_{x,x'\in \mathcal{D}} \log(c \cdot \frac{(m-1)d\cdot c-(c-1)\sum_{j=2}^{m-1}[(j-1)\cdot|G_j(x)|]}{(m-1)d\cdot c-(c-1)\sum_{j=2}^{m-1}[(j-1)\cdot|G_j(x')|]})$-LDP in all the cases, where $\epsilon$ is a strict constant privacy bound derived by $c$ and domain $\mathcal{D}$.
\end{proof}

\subsection{Proof of Theorem \ref{theorem:l2} ($L_2$ Error Bound)}
\label{sec:l2}

\begin{proof}
With the estimation formula, we have $p(C_{x})=p(x) \cdot \sum_{y\in C_{x}}q(y|x)+\sum_{x' \neq x} p(x')
\cdot [\sum_{y \in C_{x} \setminus C_{x'}} q(y|x') + \sum_{y \in C_{x} \cap C_{x'}}q(y|x')]$. With the property of Hadamard matrix \cite{emp}, the size of the set difference between any two location candidate sets is $\frac{d}{4}$, and the size of intersection between any two candidate sets of locations is also $\frac{d}{4}$. We can integrate these into the equation. Then, we have: 

\vspace{-0.1in}

\small

\begin{align*}
p(C_{x})&\geq p(x) \cdot [\sum_{y\in C_{x}}q(y|x)]+\sum_{x' \neq x} p(x')
\cdot \frac{d\cdot\min\{q(y|x')\}}{2}\\
&=p(x) \cdot [\sum_{y\in C_{x}}q(y|x)] + [1-p(x)]\cdot \frac{d\cdot \alpha_{min}(x)}{2}\\
&\implies p(x)\leq \frac{q(C_{x})-\frac{d \cdot \alpha_{min}(x)}{2}}{\sum_{y \in C_{x}}[q(y|x)-\frac{d \cdot \alpha_{min}(x)}{2}]}
\end{align*}

\normalsize

Then, we can have the $L^2_2$-distance as below:

\vspace{-0.15in}

\small

\begin{align*}
\mathbb{E}[L^2_2(\tilde{p},p)] \leq &~~(\frac{1}{\sum_{y \in C_{x}} q(y|x)-\frac{d \cdot \mu}{2}})^2
\cdot \mathbb{E}[L^2_2(\tilde{p}(C),p(C)]
\end{align*}

\normalsize

where $\mu=\min\{\alpha_{min}(x)\}$.
Since $\mathbb{E}[\tilde{p}(C_{x})]=\mathbb{E}[\frac{\mathbb I \{ y_j \in C_{x}}{n}\}]\\=p(C_{x})$, we have:

\vspace{-0.15in}

\small

\begin{align*}
\mathbb{E}[L^2_2(\tilde{p}(C),p(C))]&=\mathbb{E}[\sum_{x \in \mathcal{D}}(\tilde{p}(C_{x})-p(C_{x}))^2]
=\sum_{x \in \mathcal{D}} Var(\tilde{p}(C_{x})) 
\end{align*}

\normalsize

Moreover, each $y$ is independently sampled and $\tilde{p}(C_{x})=p(C_{x})$ is the mean of $n$ independent multinomial distributions.

\vspace{-0.15in}

\small

\begin{align*}
\sum_{x \in \mathcal{D}} Var(\tilde{p}(C_{x})) \leq \sum_{x \in \mathcal{D}} \frac{1}{n} \cdot \max \{p(C_{x})\} \leq \frac{d}{n}
\end{align*}

\normalsize

Thus, we have $\mathbb{E}[L_{2}(\tilde{p},p)]\leq (\frac{1}{\sum_{y \in C_{x}} q(y|x)-\frac{d \cdot \mu}{2}}) \cdot \sqrt{\frac{d}{n}}$.
\end{proof}

\vspace{0.05in}
\subsection{Proof of Theorem \ref{theorem:l1} ($L_1$ Error Bound)}
\label{sec:l1}

\normalsize

\begin{proof}
Since $\forall i, a_i >0$, $ n\cdot \sum^n_{i=1}(a_n)^2 \geq [\sum^n_{i=1}(a_n)]^2$ holds, we have $d\cdot L^2_2(\tilde{p},p) \geq [L_1(\tilde{p},p)]^2$. Then. we can derive:

\vspace{-0.1in}

\begin{equation*}
(\mathbb{E}[L_1(\tilde{p},p)])^2 \leq \frac{d^2}{n \cdot(\gamma-\frac{d \cdot \mu}{2})^2}
\end{equation*}

Thus, $\mathbb{E}[L_1(\tilde{p},p)] \leq \frac{2d}{\sqrt{n} \cdot(2\gamma-d \cdot \mu)}$ completes the proof.
\end{proof}

\section{Additional Experiments}
\label{sec:trajectory}

\subsection{Traffic-Aware GPS Navigation}
\begin{figure}[!tbh]
	\centering
	\subfigure[Gowalla]{
		\includegraphics[angle=0, width=0.496\linewidth]{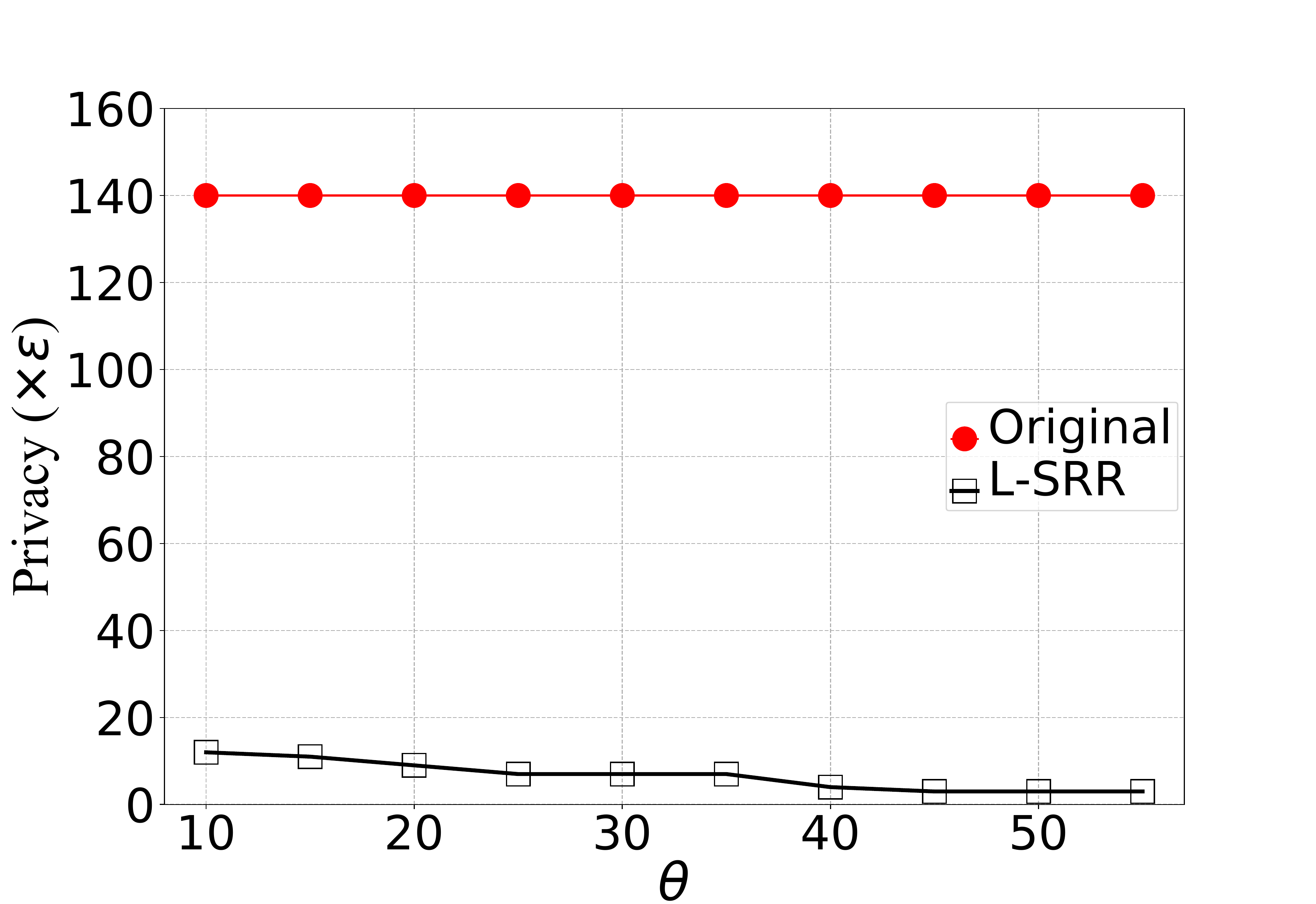}
		\label{fig:OD2_C} }
		\hspace{-0.25in}
	\subfigure[Geolife]{
		\includegraphics[angle=0, width=0.496\linewidth]{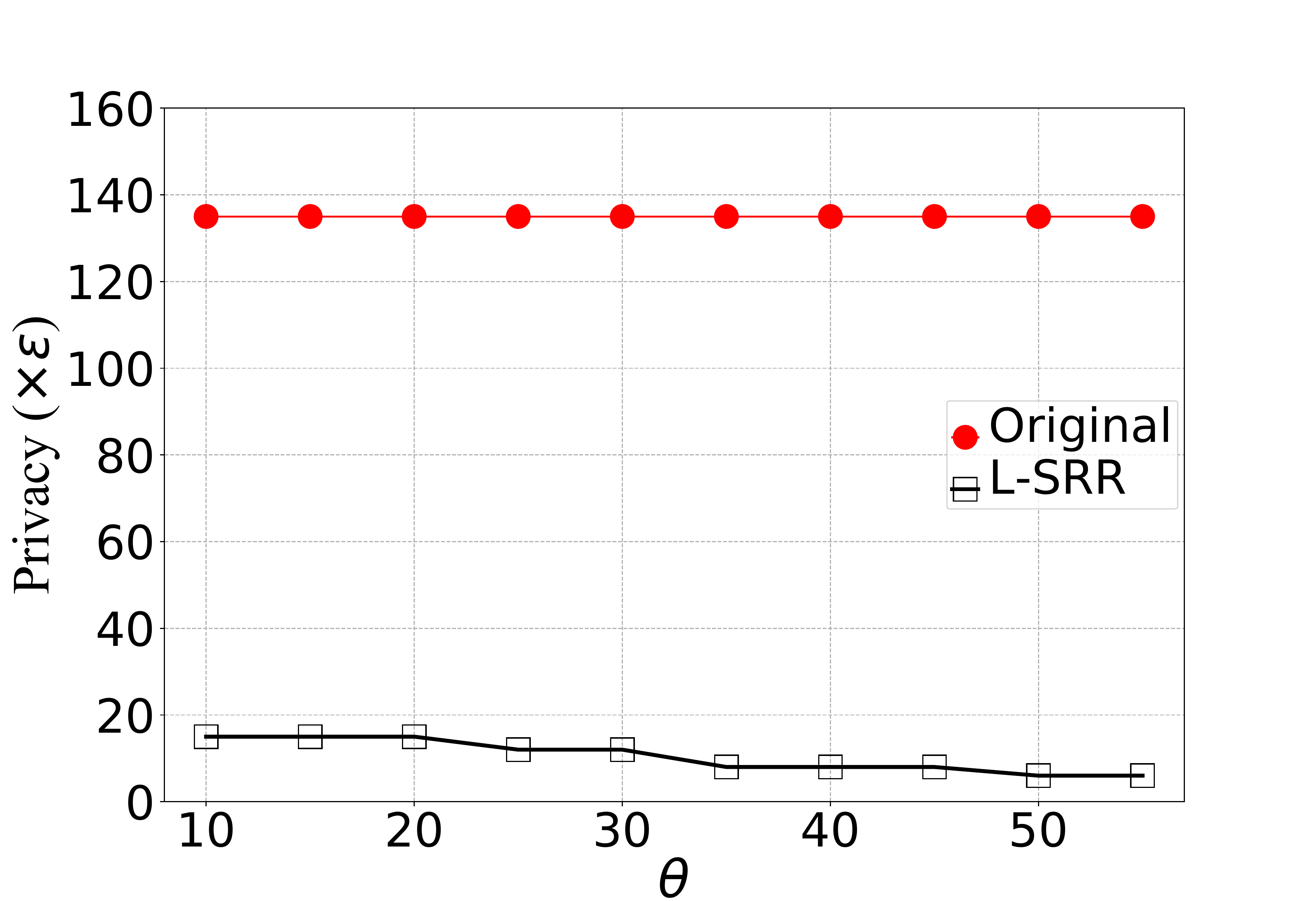}
		\label{fig:OD2_G} }
		\hspace{-0.25in}
	\subfigure[Portocabs]{
		\includegraphics[angle=0, width=0.496\linewidth]{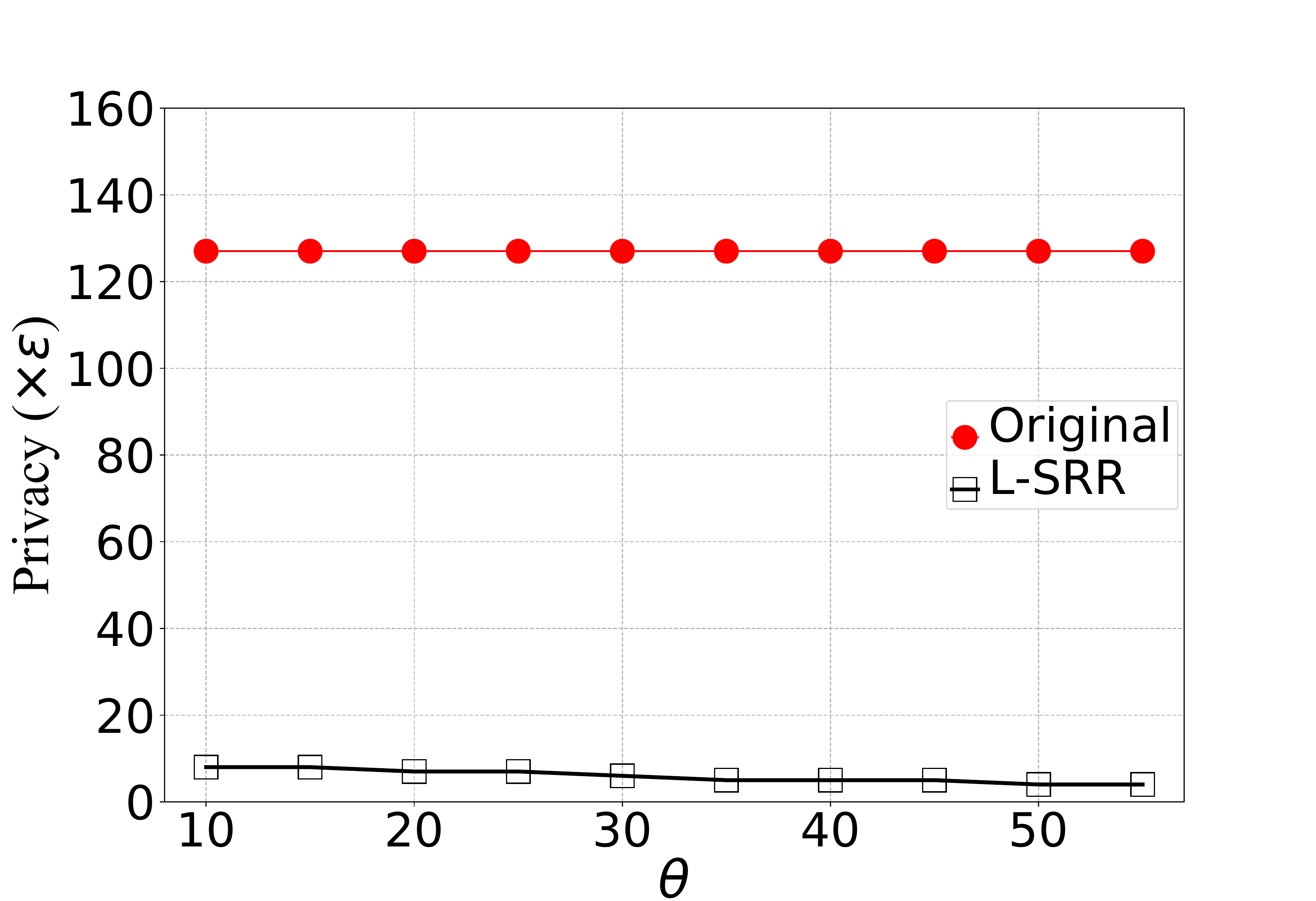}
		\label{fig:OD2_P}}
		\hspace{-0.25in}
	\subfigure[Foursquare]{
		\includegraphics[angle=0, width=0.496\linewidth]{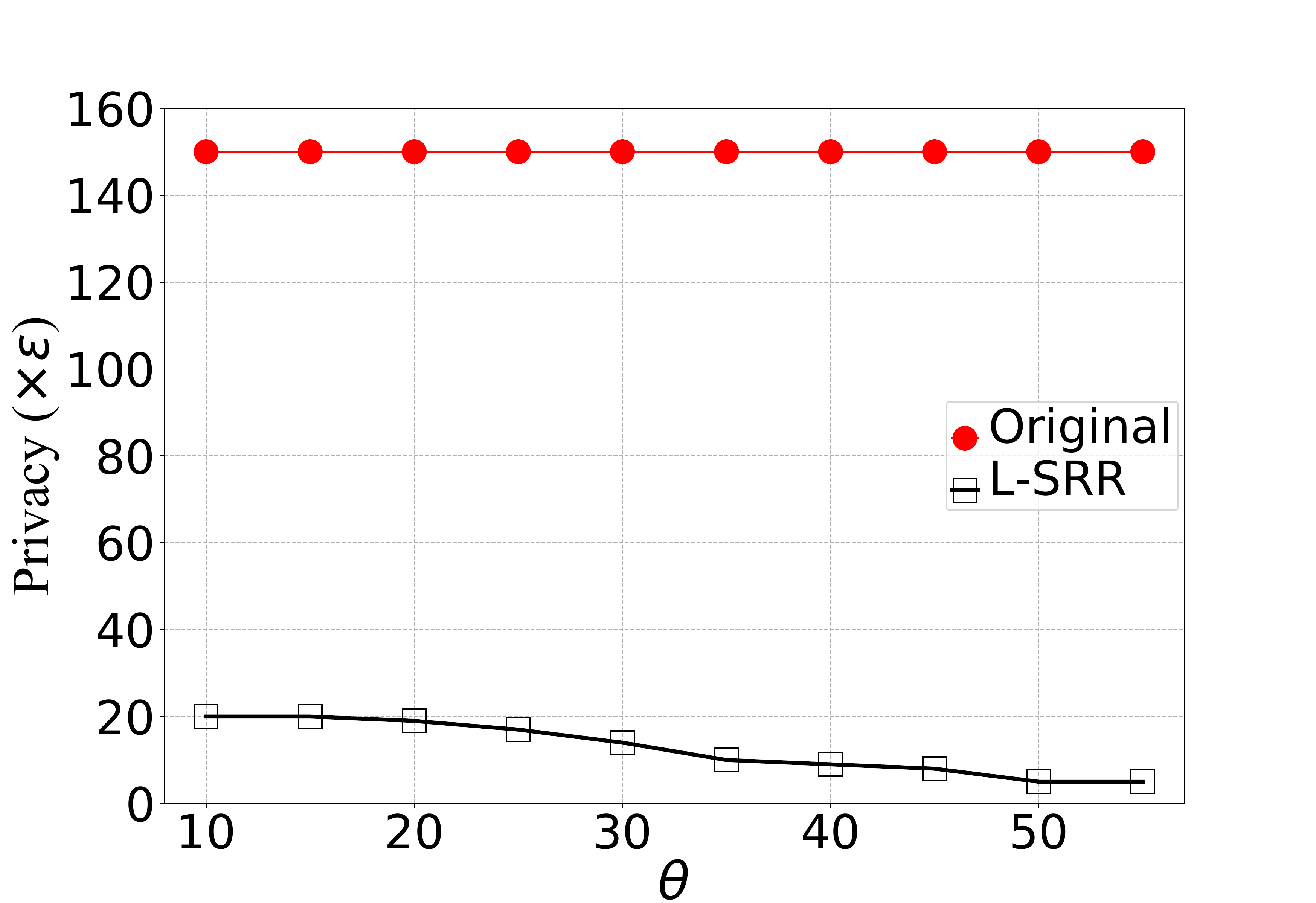}
		\label{fig:OD2_S}}\vspace{-0.1in}
	\caption{The total privacy bound of \texttt{L-SRR} for traffic-aware GPS navigation by collecting trajectories}\vspace{-0.1in}
	\label{fig:theta}
\end{figure}
We simulate many trajectories and predict the time $Agg^p(x_1,$\\$ x_2)$ between any two locations on the trajectory using the Markov Chain \cite{HMM}. Specifically, we generate multiple routes for each OD pair (at client). For each route, we compute the predicted time $t$ based on historical datasets for any two locations. In our experiment, we use the data collected earlier as the historical data (e.g., Geolife dataset collected in 2009 as the historical data, and collected in 2010 as the test data). Furthermore, for locations on each route, such LBS calculates the frequencies of users near the location within a range (e.g., 4.7m). If the frequency exceeds a threshold (e.g., 50 users), a 3-second delay time will be added \cite{trafficpath}. Finally, given the traffic density, it may update the route to avoid heavy traffic. With \texttt{L-SRR}, the route recalculation occurs if $Agg^t(x_o, x_i)-Agg^p(x_o, x_i)>\theta$ holds, where $i\in\mathbb{T}$ and $\theta$ is the delay threshold (e.g., 30 seconds). If yes, the client will submit its perturbed location, and privately retrieve the traffic density at the current position to recalculate the fastest route \cite{trafficpath}. 

In the experiments, we first evaluate how the delay time threshold $\theta$ affects the total privacy guarantee. The maximum numbers of locations on the trajectories for four datasets are $140$, $135$, $127$, and $150$, respectively. 
In Figure \ref{fig:theta}, we set $\theta$ between 10 seconds and 55 seconds with a step of 5 seconds. As $\theta$ increases, the total privacy bound $\epsilon$ decreases with a decreasing number of location updates. As $\theta=60$ (updating the location once delay exceeds 1 minute), the privacy bound is around $3\epsilon$, which is very small for trajectories.
%\vspace{-0.2in}
\begin{figure}[!tbh]
	\centering
	\subfigure[Gowalla]{
		\includegraphics[angle=0, width=0.496\linewidth]{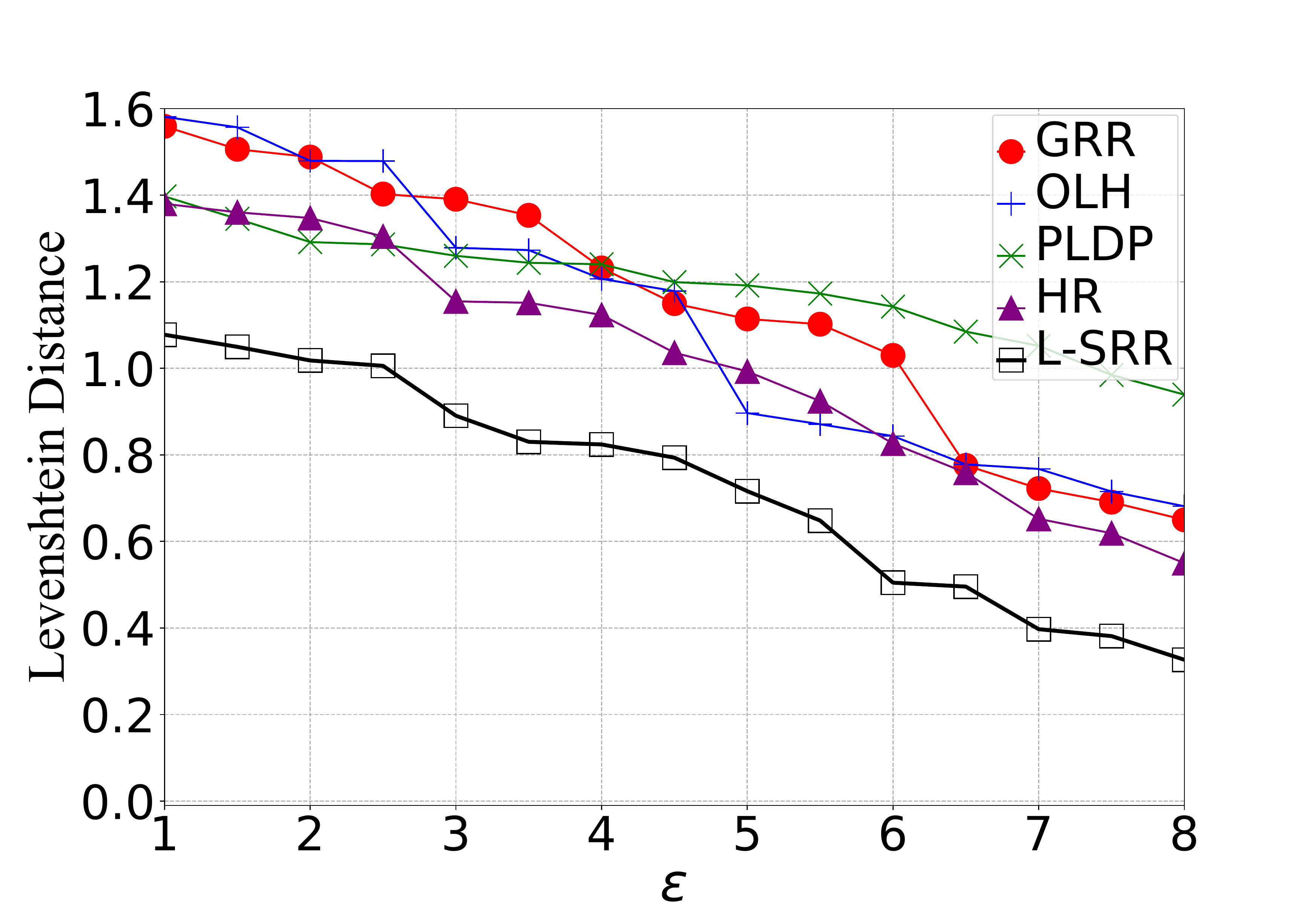}
		\label{fig:OD2_C} }
		\hspace{-0.25in}
	\subfigure[Geolife]{
		\includegraphics[angle=0, width=0.496\linewidth]{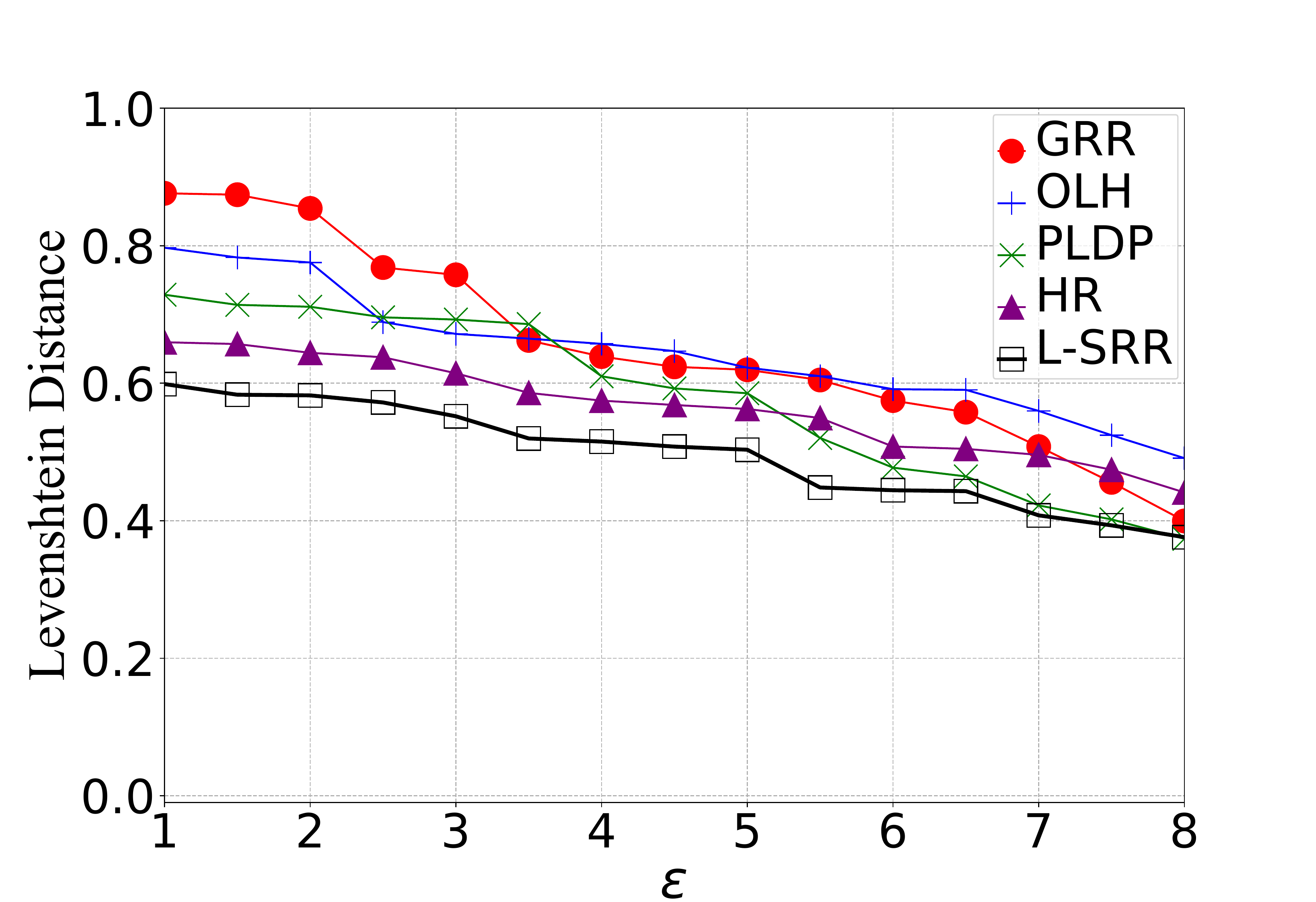}
		\label{fig:OD2_G} }
		\hspace{-0.25in}
	\subfigure[Portocabs]{
		\includegraphics[angle=0, width=0.496\linewidth]{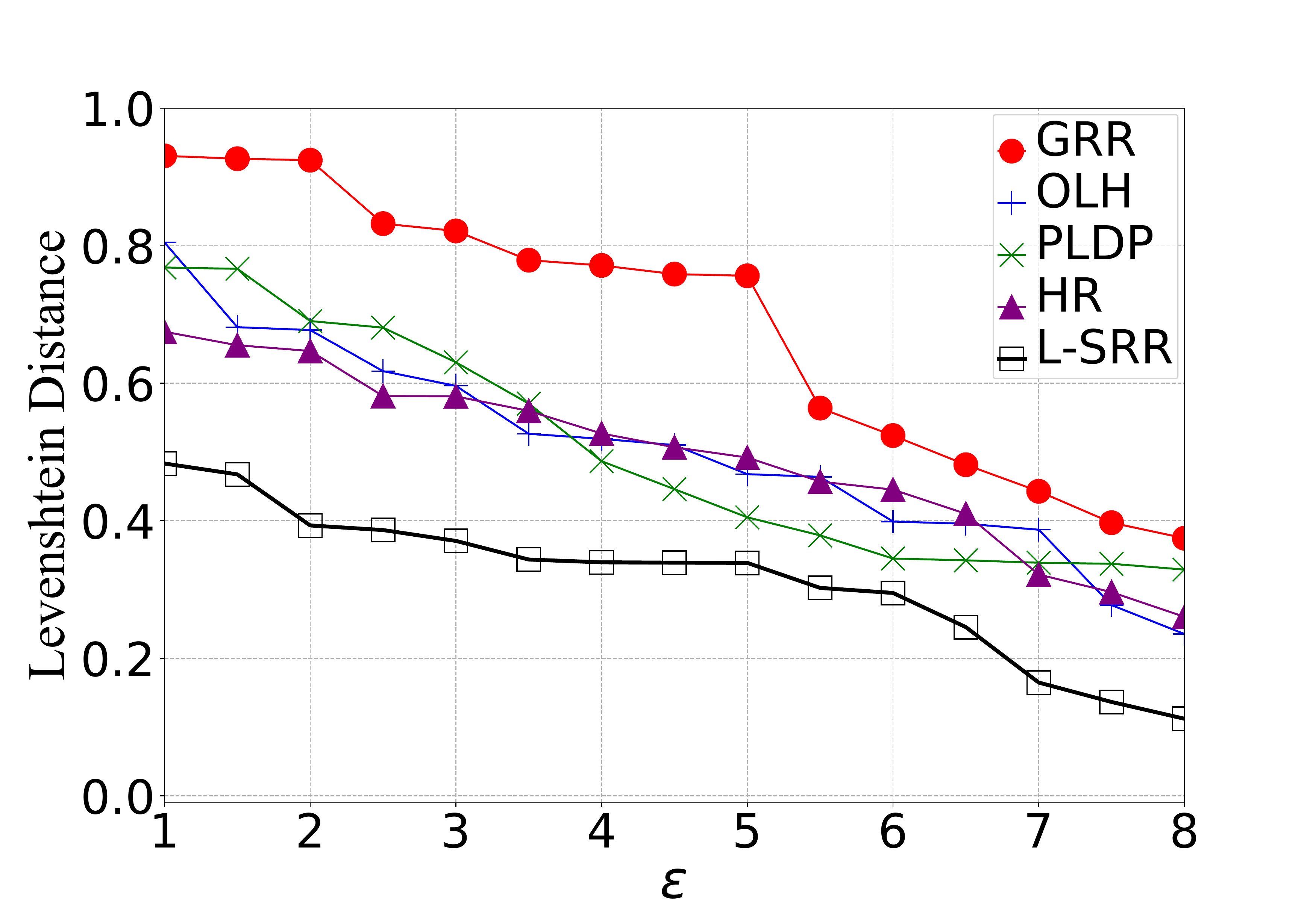}
		\label{fig:OD2_P}}
		\hspace{-0.25in}
	\subfigure[Foursquare]{
		\includegraphics[angle=0, width=0.496\linewidth]{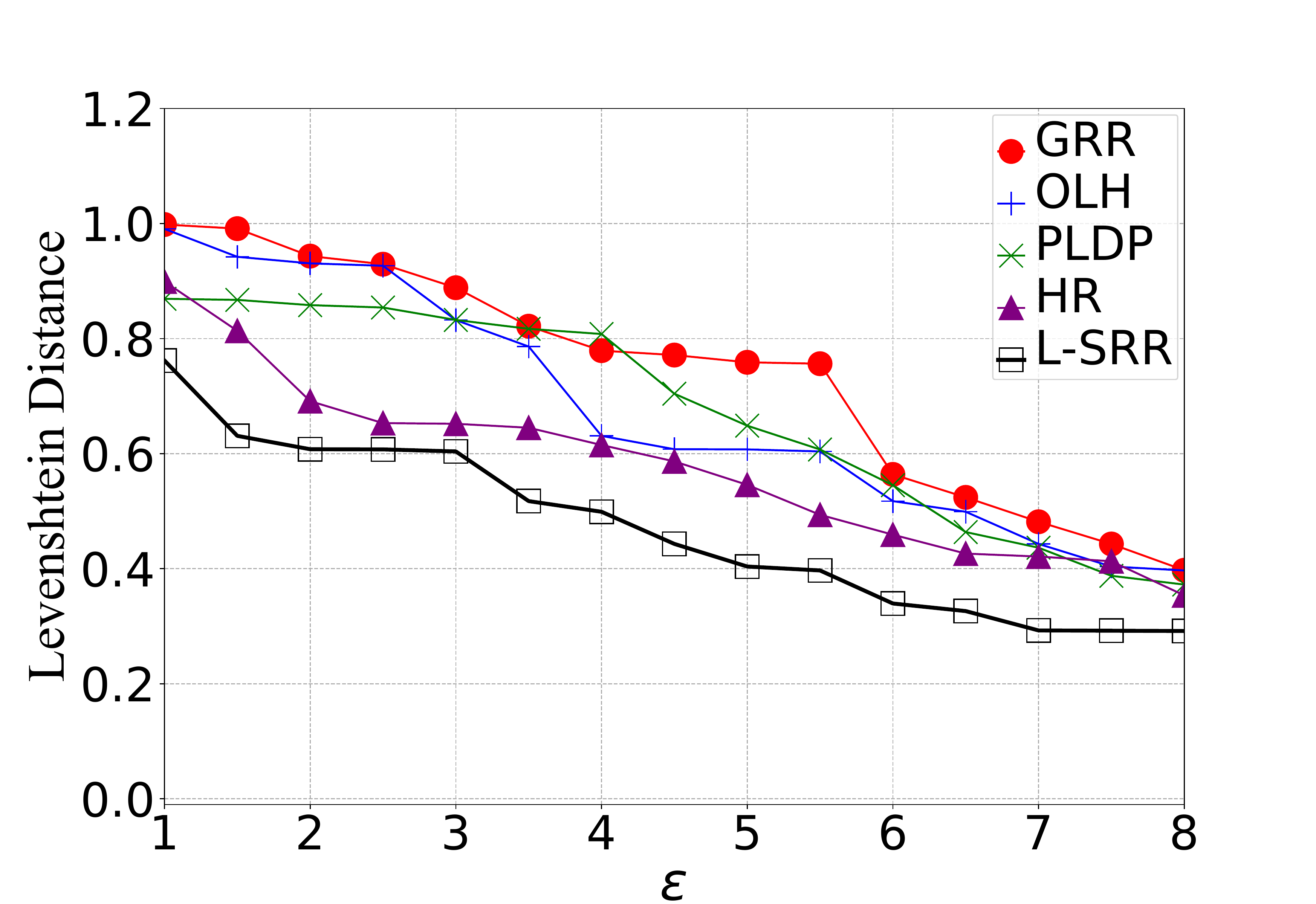}
		\label{fig:OD2_S}}\vspace{-0.15in}
	\caption{Relative levenshtein distance of trajectories in the traffic-aware GPS navigation ($\theta=40$ seconds)}\vspace{-0.2in}
	\label{fig:Rou}
\end{figure}
 
Second, to measure the route deviation, we apply Levenshtein distance to measure the accuracy between true route and the route recommended by \texttt{L-SRR}. It measures the difference by calculating the minimum number of location edits (insertions, deletions, or substitutions) required to change one route to the other. Figure \ref{fig:Rou} shows the relative Levenshtein distance over the total size of the true routes (vs the total privacy bound $\epsilon$). \texttt{L-SRR} again outperforms other LDP schemes. In addition, we also measure the deviation of the total trip time. Figure \ref{fig:walk} shows that the trip time deviation decreases as the privacy bound $\epsilon$ increases for all the LDP schemes, and \texttt{L-SRR} results in the least trip time deviation.

\begin{figure}[!tbh]
	\centering
	\subfigure[Gowalla]{
		\includegraphics[angle=0, width=0.496\linewidth]{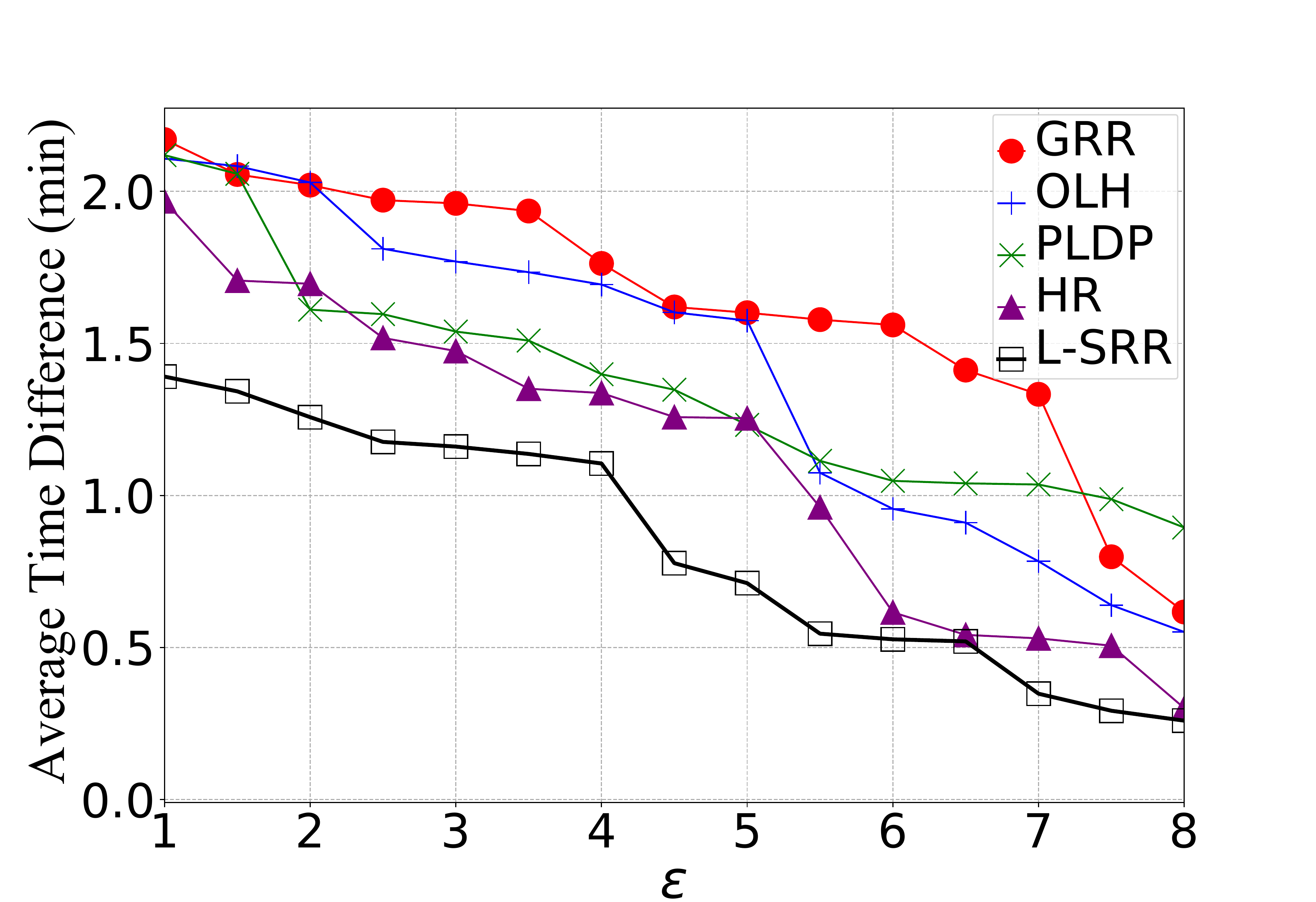}
		\label{fig:walk_C} }
		\hspace{-0.25in}
	\subfigure[Geolife]{
		\includegraphics[angle=0, width=0.496\linewidth]{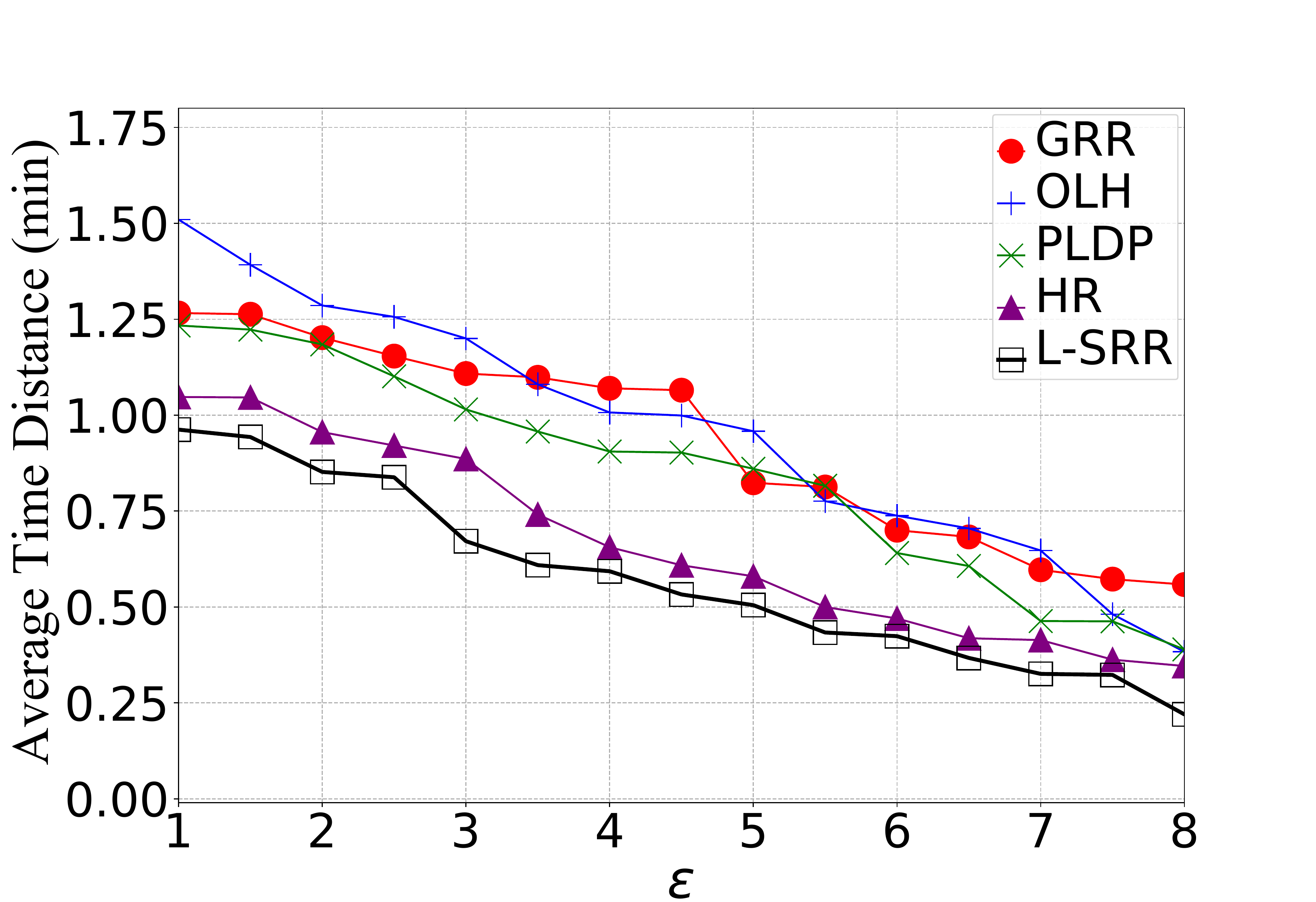}
		\label{fig:walk_G} }
		\hspace{-0.25in}
	\subfigure[Portocabs]{
		\includegraphics[angle=0, width=0.496\linewidth]{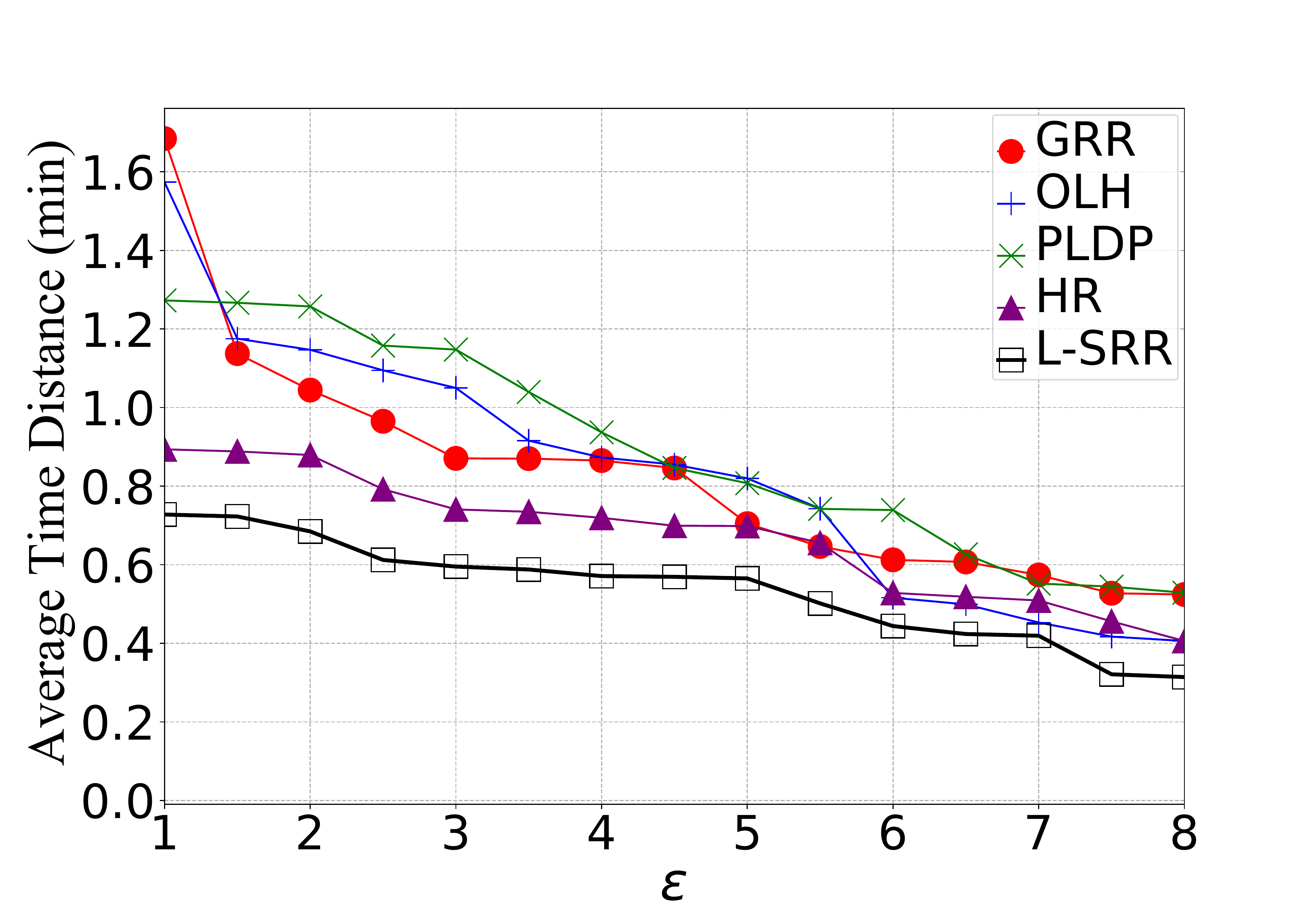}
		\label{fig:walk_P}}
		\hspace{-0.25in}
	\subfigure[Foursquare]{
		\includegraphics[angle=0, width=0.496\linewidth]{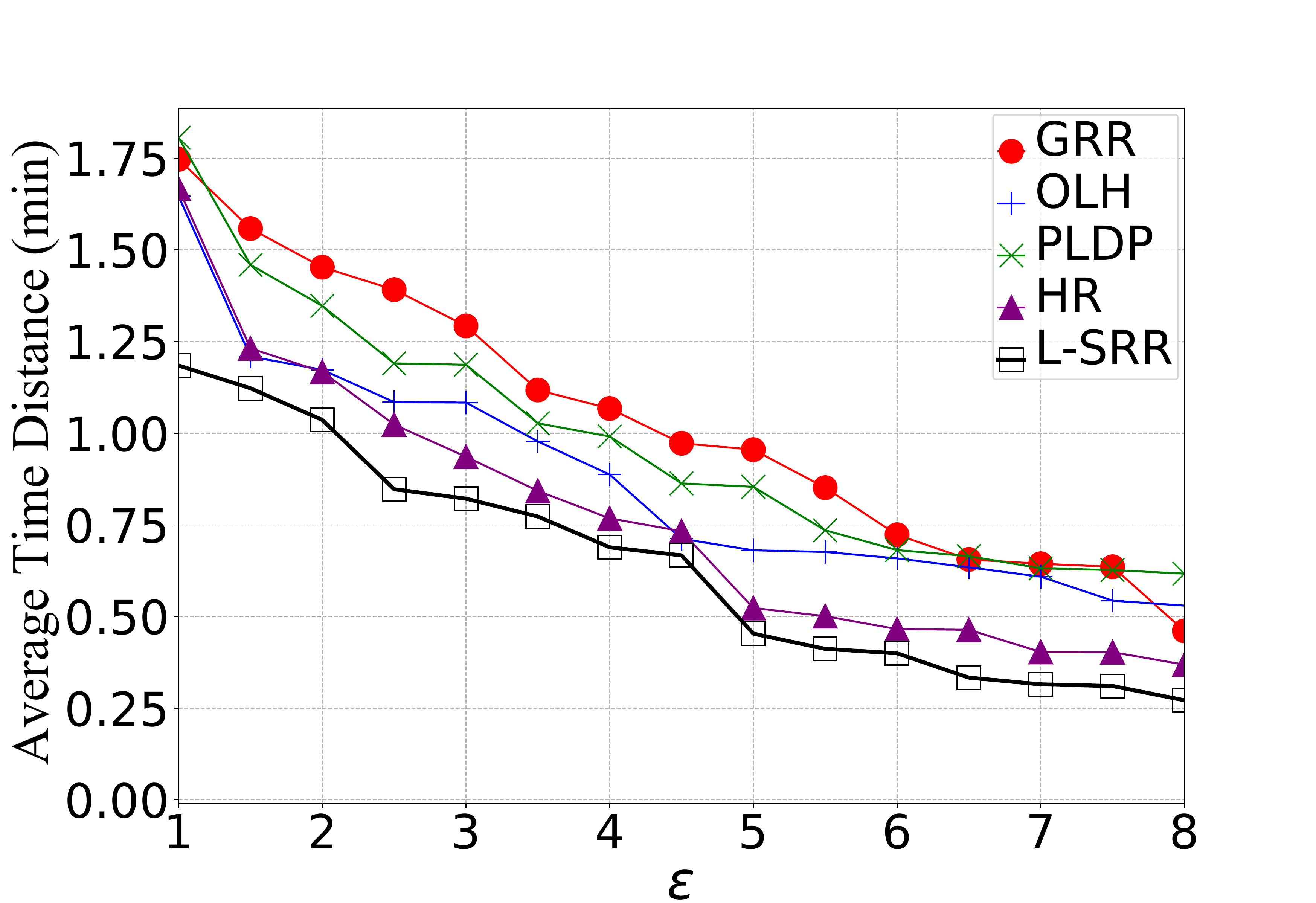}
		\label{fig:walk_S}}\vspace{-0.15in}
	\caption{Average trip time deviation in the traffic-aware GPS navigation ($\theta=40$ seconds)} \vspace{-0.1in}
	\label{fig:walk}
\end{figure}

\subsection{\texttt{SRR} and Differential Privacy}
We discuss the utility of centralized differential privacy. A generic solution is to add the Laplace noise to the frequency of each location (after aggregation). Thus, $\epsilon$ should be equally allocated for $d$ locations. Table \ref{table:lap} presents the $L_1$-distance for the distribution on four datasets using Laplace mechanism. The results show that the $L_1$-distance gets smaller as $\epsilon$ becomes larger. Compared to the LDP mechanism, for Gowalla and Foursquare, the distance with \texttt{SRR} has smaller distance. For Geolife and Portocabs, the distance with \texttt{SRR} has similar distance in case of a smaller domain. Note that the privacy guarantees of LDP and DP are indeed incomparable even for the same $\epsilon$ (since the trust model and indistinguishability are defined in different ways).

\begin{table}[!h]\small
\caption{Average $L_1$-distance for the location distribution on four datasets using Laplace mechanism for centralized DP}
\vspace{-0.1in}
\begin{center}
\begin{tabular}{|c|c|c|c|c|}
\hline
		\textbf{Dataset} &  \textbf{$\epsilon=1$} &
		\textbf{$\epsilon=3$} &  \textbf{$\epsilon=5$} &
		\textbf{$\epsilon=7$}\\
		\hline
		 	Gowalla  &0.18 & 0.16 & 0.15 & 0.13 \\\hline
			Geolife  &0.29 & 0.11 & 0.08 & 0.04 \\\hline
			Portocabs  &0.43 & 0.26 & 0.15 & 0.07  \\\hline
	 	Foursquare  &13.87 & 4.69 & 2.78 & 1.91 \\\hline
\end{tabular}
\label{table:lap}
\end{center}\vspace{-0.05in}
\end{table}

\subsection{Utility vs. Smaller $\epsilon$}
To evaluate the performance with smaller $\epsilon$, we vary the $\epsilon$ from 0.5 to 1 with a step of 0.1. Table \ref{table:1}-\ref{table:4} show the average $L_1$-distance for the location distribution on the Gowalla, Geolife, Portocabs, and Foursquare datasets, respectively. Table \ref{table:5}-\ref{table:8} show the average KL-divergence. It follows the same trend as results when $\epsilon\geq 1$. 

\begin{table}[!h]\small
\caption{Average $L_1$-distance for the location distribution on the Gowalla dataset ($\epsilon\leq 1$)}
\vspace{-0.1in}
\begin{center}
\begin{tabular}{|c|c|c|c|c|c|c|}
\hline
		\textbf{} &  \textbf{$\epsilon=0.5$} &
		\textbf{$\epsilon=0.6$} &  \textbf{$\epsilon=0.7$} &
		\textbf{$\epsilon=0.8$} &
		\textbf{$\epsilon=0.9$} &
		\textbf{$\epsilon=1$}\\
		\hline
		 	\texttt{GRR}  &0.154 & 0.138 & 0.125 & 0.113 & 0.108 &0.091 \\\hline
			\texttt{OLH-H}  &0.129 & 0.121 & 0.118 & 0.109 &0.097  &0.084\\\hline
			\texttt{PLDP}  &0.125 & 0.119 &  0.115 & 0.107  &0.094 & 0.086 \\\hline
	 	\texttt{HR}  &0.104 & 0.096 & 0.088 & 0.085 & 0.083 & 0.082 \\\hline
	 	\hline
	 	\texttt{L-SRR}	&0.097 & 0.086 & 0.079 & 0.078 & 0.075  &0.072\\\hline
\end{tabular}
\label{table:1}
\end{center}\vspace{-0.05in}
\end{table}

\begin{table}[!h]\small
\caption{Average $L_1$-distance for the location distribution on the Geolife dataset ($\epsilon\leq 1$)}
\vspace{-0.1in}
\begin{center}
\begin{tabular}{|c|c|c|c|c|c|c|}
\hline
		\textbf{} &  \textbf{$\epsilon=0.5$} &
		\textbf{$\epsilon=0.6$} &  \textbf{$\epsilon=0.7$} &
		\textbf{$\epsilon=0.8$} &
		\textbf{$\epsilon=0.9$} &
		\textbf{$\epsilon=1$}\\
		\hline
		 	\texttt{GRR}  &1.128 & 1.005 & 0.985 & 0.980 &0.945 &0.939 \\\hline
			\texttt{OLH-H}  &0.957 & 0.938 & 0.933 & 0.927 & 0.925 &0.922\\\hline
			\texttt{PLDP}  &0.941 & 0.930 &  0.916 & 0.897  &0.874 &0.852 \\\hline
	 	\texttt{HR}  &0.873 & 0.869 & 0.857 & 0.791 & 0.764 &0.755\\\hline
	 	\hline
	 	\texttt{L-SRR}	&0.832 & 0.793 & 0.782 & 0.687 & 0.641  &0.612\\\hline
\end{tabular}
\label{table:2}
\end{center}\vspace{-0.05in}
\end{table}

\begin{table}[!h]\small
\caption{Average $L_1$-distance for the location distribution on the Portocabs dataset ($\epsilon\leq 1$)}
\vspace{-0.1in}
\begin{center}
\begin{tabular}{|c|c|c|c|c|c|c|}
\hline
		\textbf{} &  \textbf{$\epsilon=0.5$} &
		\textbf{$\epsilon=0.6$} &  \textbf{$\epsilon=0.7$} &
		\textbf{$\epsilon=0.8$} &
		\textbf{$\epsilon=0.9$} &
		\textbf{$\epsilon=1$}\\
		\hline
		 	\texttt{GRR}  &2.358 & 2.216 & 2.195 & 2.193 &2.185 & 2.164 \\\hline
			\texttt{OLH-H} &1.972 & 1.961 & 1.958 & 1.905 & 1.878 &1.845\\\hline
			\texttt{PLDP}  &2.043 & 2.026 &  2.015 & 2.003  &1.974 &1.941 \\\hline
	 	\texttt{HR}  &1.978 & 1.969 & 1.885 & 1.871 &1.869 & 1.861 \\\hline
	 	\hline
	 	\texttt{L-SRR}	&1.887 & 1.869 & 1.797 & 1.751 & 1.705  &1.635\\\hline
\end{tabular}
\label{table:3}
\end{center}\vspace{-0.05in}
\end{table}

\begin{table}[!h]\small
\caption{Average $L_1$-distance for the location distribution on the Foursquare dataset ($\epsilon\leq 1$)}
\vspace{-0.1in}
\begin{center}
\begin{tabular}{|c|c|c|c|c|c|c|}
\hline
		\textbf{} &  \textbf{$\epsilon=0.5$} &
		\textbf{$\epsilon=0.6$} &  \textbf{$\epsilon=0.7$} &
		\textbf{$\epsilon=0.8$} &
		\textbf{$\epsilon=0.9$} &
		\textbf{$\epsilon=1$}\\
		\hline
		 	\texttt{GRR}  &0.138 & 0.126 & 0.115 & 0.113 & 0.094 &0.088 \\\hline
			\texttt{OLH-H}  &0.134 & 0.127 & 0.118 & 0.094 & 0.087 &0.082\\\hline
			\texttt{PLDP}  &0.123 & 0.119 &  0.105 & 0.097  & 0.085 & 0.079\\\hline
	 	\texttt{HR}  &0.118 & 0.109 & 0.097 & 0.091 &0.083 &0.073\\\hline
	 	\hline
	 	\texttt{L-SRR}	&0.087 & 0.079 & 0.072 & 0.069 & 0.061  &0.055\\\hline
\end{tabular}
\label{table:4}
\end{center}\vspace{-0.05in}
\end{table}

\begin{table}[!h]\small
\caption{Average KL-divergence for the location distribution on the Gowalla dataset ($\epsilon\leq 1$)}
\vspace{-0.1in}
\begin{center}
\begin{tabular}{|c|c|c|c|c|c|c|}
\hline
		\textbf{} &  \textbf{$\epsilon=0.5$} &
		\textbf{$\epsilon=0.6$} &  \textbf{$\epsilon=0.7$} &
		\textbf{$\epsilon=0.8$} &
		\textbf{$\epsilon=0.9$} &
		\textbf{$\epsilon=1$}\\
		\hline
		 	\texttt{GRR}  &2.990 & 2.924 & 2.911 & 2.894 & 2.881 & 2.874 \\\hline
			\texttt{OLH-H}  &2.533 & 2.491 & 2.437 & 2.432 & 2.401& 2.385 \\\hline
			\texttt{PLDP}  &2.713 & 2.680 &  2.649 & 2.614  &2.582 & 2.413 \\\hline
	 	\texttt{HR}  & 2.472 & 2.421 & 2.394 & 2.372 & 2.220 & 2.171\\\hline
	 	\hline
	 	\texttt{L-SRR}	& 1.953 & 1.951 & 1.871 & 1.867 & 1.853 & 1.753
\\\hline
\end{tabular}
\label{table:5}
\end{center}\vspace{-0.05in}
\end{table}

\begin{table}[!h]\small
\caption{Average KL-divergence for the location distribution on the Geolife dataset ($\epsilon\leq 1$)}
\vspace{-0.1in}
\begin{center}
\begin{tabular}{|c|c|c|c|c|c|c|}
\hline
		\textbf{} &  \textbf{$\epsilon=0.5$} &
		\textbf{$\epsilon=0.6$} &  \textbf{$\epsilon=0.7$} &
		\textbf{$\epsilon=0.8$} &
		\textbf{$\epsilon=0.9$} &
		\textbf{$\epsilon=1$}\\
		\hline
		 	\texttt{GRR}  &1.330 & 1.292 & 1.285 & 1.201 & 1.137 &1.069 \\\hline
			\texttt{OLH-H}  &1.292 & 1.113 & 1.082 & 1.044 & 0.940 & 0.939\\\hline
			\texttt{PLDP}  &1.114 & 1.109 & 1.097 & 1.078  &1.022 & 0.991 \\\hline
	 	\texttt{HR}  &0.973 & 0.897 & 0.892 & 0.815 &0.787 & 0.643\\\hline
	 	\hline
	 	\texttt{L-SRR}	&0.895 & 0.887 & 0.775 & 0.713 & 0.693  & 0.623\\\hline
\end{tabular}
\label{table:6}
\end{center}\vspace{-0.05in}
\end{table}

\begin{table}[!h]\small
\caption{Average KL-divergence for the location distribution on the Portocabs dataset ($\epsilon\leq 1$)}
\vspace{-0.1in}
\begin{center}
\begin{tabular}{|c|c|c|c|c|c|c|}
\hline
		\textbf{} &  \textbf{$\epsilon=0.5$} &
		\textbf{$\epsilon=0.6$} &  \textbf{$\epsilon=0.7$} &
		\textbf{$\epsilon=0.8$} &
		\textbf{$\epsilon=0.9$} &
		\textbf{$\epsilon=1$}\\
		\hline
		 	\texttt{GRR}  &1.384 & 1.257 & 1.154 & 1.033 & 0.957 & 0.905 \\\hline
			\texttt{OLH-H}  &1.292 & 1.121 & 1.087 & 0.848 & 0.833 & 0.749 \\\hline
			\texttt{PLDP}  &0.943 & 0.926 &  0.895 & 0.871  &0.735 & 0.718 \\\hline
	 	\texttt{HR}  &1.070 & 1.039 & 0.978 & 0.841 &0.748 & 0.662\\\hline
	 	\hline
	 	\texttt{L-SRR}	&0.899 & 0.841 & 0.838 & 0.775 &0.619  &0.606 \\\hline
\end{tabular}
\label{table:7}
\end{center}\vspace{-0.05in}
\end{table}

\begin{table}[!h]\small
\caption{Average KL-divergence for the location distribution on the Foursquare dataset ($\epsilon\leq 1$)}
\vspace{-0.1in}
\begin{center}
\begin{tabular}{|c|c|c|c|c|c|c|}
\hline
		\textbf{} &  \textbf{$\epsilon=0.5$} &
		\textbf{$\epsilon=0.6$} &  \textbf{$\epsilon=0.7$} &
		\textbf{$\epsilon=0.8$} &
		\textbf{$\epsilon=0.9$} &
		\textbf{$\epsilon=1$}\\
		\hline
		 	\texttt{GRR}  &2.251 & 2.163 & 2.155 & 2.134 &2.019 &1.952
 \\\hline
			\texttt{OLH-H}  &2.112 & 2.091 & 2.084 & 2.052 &1.974 & 1.946
\\\hline
			\texttt{PLDP}  &2.154 & 2.126 &  2.075 & 2.034  &1.994 & 1.957
 \\\hline
	 	\texttt{HR}  &2.020 & 1.989 & 1.978 & 1.952 & 1.912 &1.832
\\\hline
\hline
	 	\texttt{L-SRR}	&1.917 & 1.898 & 1.883 & 1.871 &1.807  &1.782
\\\hline
\end{tabular}
\label{table:8}
\end{center}\vspace{-0.05in}
\end{table}

\subsection{Utility vs. Multiple Cities in The World}
\label{subsec:largerdomain}
Four real-world LBS datasets in the experiments are collected from four different cities worldwide: Gowalla (Austin, USA), Geolife (Beijing, China), Portocabs (Porto, Portugal), and Foursquare (New York, USA). Then, we present a new set of experiments by comparing the LDP schemes on one dataset (1 city), two merged datasets (2 cities), three merged datasets (3 cities), and four merged datasets (4 cities) -- merging both the location domain and data:

\vspace{0.05in}

\begin{itemize}
\setlength\itemsep{0.5em}

    \item 1 City: Foursquare (New York, USA)
    \item 2 Cities: Foursquare (New York, USA), Gowalla (Austin, USA) 
    \item 3 Cities: Foursquare (New York, USA), Gowalla (Austin, USA), Geolife (Beijing, China)
    \item 4 Cities: Foursquare (New York, USA), Gowalla (Austin, USA), Geolife (Beijing, China), Portocabs (Porto, Portugal)
\end{itemize}

\vspace{0.05in}

Note that merging the datasets enlarges the domain of the locations as well as diversify the distribution of the locations in the domain (e.g., the encoded bit strings of the locations in different cities would share less prefixes since the cities are far away on the map). Table \ref{table:9} and \ref{table:10} show the $L_1$-distance and KL-divergence for this group of experiments. As more cities are merged, the errors ($L_1$-distance and KL-divergence) slightly increase and \texttt{L-SRR} still works the best in all the cases. Thus, we have the following findings (along with the experimental results):

\vspace{0.05in}

\begin{enumerate}
\setlength\itemsep{0.5em}

    \item If locations are distributed on the worldwide map (e.g., locations in the US, China and Portugal) but the domain size is reasonable (e.g., as above), \texttt{L-SRR} still works. 
    \item If merging more worldwide cities with a large number of locations inside each city (e.g., forming the domain with most of the locations in thousands of cities all over the world), we anticipate that \texttt{L-SRR} still outperforms other LDP schemes but all the LDP schemes may have a limitation on retaining good utility while ensuring strong privacy. 
    
\end{enumerate}

\begin{table}[!h]\small
\caption{Average $L_1$-distance for the location distribution vs. domain/datasets of multiple cities ($\epsilon=1$)}
\vspace{-0.1in}
\begin{center}
\begin{tabular}{|c|c|c|c|c|}
\hline
		\textbf{City \#} &  \textbf{$1$} &
		\textbf{$2$} &  \textbf{$3$} &
		\textbf{$4$} \\
		\hline
		 	\texttt{GRR}  &0.088 & 0.175  &0.204 & 0.224 \\\hline
			\texttt{OLH-H}  &0.082 &0.107 &0.138  &0.158\\\hline
			\texttt{PLDP}  & 0.079 & 0.095  &0.092  &0.114  \\\hline
	 	\texttt{HR}  &0.073 & 0.094 &0.099 & 0.108 \\\hline
	 	\hline
	 	\texttt{L-SRR}	&0.055 & 0.074 & 0.083  & 0.090 \\\hline
\end{tabular}
\label{table:9}
\end{center}\vspace{-0.15in}
\end{table}

\begin{table}[!h]\small
\caption{Average KL-divergence for the location distribution vs. domain/datasets of multiple cities ($\epsilon=1$)}
\vspace{-0.1in}
\begin{center}
\begin{tabular}{|c|c|c|c|c|}
\hline
		\textbf{City \#} &  \textbf{$1$} &
		\textbf{$2$} &  \textbf{$3$} &
		\textbf{$4$} \\
		\hline
		 	\texttt{GRR}  &1.952 & 2.087  &2.115 & 2.158 \\\hline
			\texttt{OLH-H}  &1.946 &2.065 &2.109 & 2.128 \\\hline
			\texttt{PLDP}  &1.957 & 2.043  & 2.098  &2.123  \\\hline
	 	\texttt{HR}  &1.832 &1.957  & 2.034  & 2.115 \\\hline
	 	\hline
	 	\texttt{L-SRR}	&1.782 & 1.893  & 1.983  & 2.085  \\\hline
\end{tabular}
\label{table:10}
\end{center}
\end{table}

%% file: main.bbl
\begin{thebibliography}{75}
\bibitem{Apple}
https://www.apple.com/privacy/docs/Differential\_Privacy\_Overview.pdf

\bibitem{Map}
https://docs.microsoft.com/en-us/bingmaps/articles/bing-maps-tile-system

\bibitem{Gowalla}
http://snap.stanford.edu/data/loc-gowalla.html
 
 \bibitem{Geo-indistinguishability}
M.~Andrés, N.~Bordenabe, K.~Chatzikokolakis, and C.~Palamidessi.
\newblock Geo-indistinguishability: Differential privacy for location-based systems.
\newblock In {\em Computer and Communications Security}, pages 901--914, 2013.

 \bibitem{PIR2}
 W.~Al Amiri, M.~Baza, K.~Banawan, M.~Mahmoud, W.~Alasmary, and K.~Akkaya.
\newblock Privacy-preserving smart parking system using blockchain and private information retrieval.
\newblock In {\em SmartNets}, pages 1--6, 2019.

\bibitem{KUI2021}
H.~Arcolezi, J.~Couchot, B.~AI Bouna, and X.~Xiao. \newblock 
Improving the Utility of Locally Differentially Private Protocols for Longitudinal and Multidimensional Frequency Estimates. arXiv preprint arXiv:2111.04636, 2021.


\bibitem{OD}
M.G.~Bell.
\newblock The estimation of an origin-destination matrix from traffic counts.
\newblock In {\em Transportation Science}, pages 198--217, 1983.

\bibitem{cloaking1}
B.~Bamba, L.~Liu, P.~Pesti, and T.~Wang.
\newblock Supporting
anonymous location queries in mobile environments with
privacygrid.
\newblock In {\em World Wide Web}, pages 237--246, 2008.

\bibitem{histogram1}
R. Bassily and A. Smith. Local, private, efficient protocols for succinct histograms. In {\em Symposium on Theory of Computing}, 2015.


\bibitem{LU1}
\newblock C.~Camarero.
\newblock Simple, Fast and Practicable Algorithms for Cholesky, LU and QR Decomposition Using Fast Rectangular Matrix Multiplicatio.
\newblock arXiv preprint arXiv:1812.02056, 2018.


\bibitem{LocationSurvey}
 K.~Chatzikokolakis, E.~ElSalamouny, C.~Palamidessi, and A.~Pazii.
\newblock Methods for location privacy: A comparative overview.
\newblock In {\em Foundations and Trends® in Privacy and Security}, 1(4), pages 199--257, 2017.

\bibitem{location_LDP}
R.~Chen, H. Li, A. K. Qin, S. P. Kasiviswanathan, and H. Jin. Private spatial data aggregation in the local setting. In {\em International Conference on Data Engineering}, pages 289--300, 2016.


\bibitem{ruichenkdd12}
R.~Chen, B. C.~Fung, B.C.~Desai, and N. M.~Sossou. Differentially private transit data publication: a case study on the montreal transportation system. In {\em Special Interest Group on Knowledge Discovery in Data}, pages 213--221, 2012.

\bibitem{side_channel}
S.~Chen, X.~Zhang, M. K.~Reiter, and Y.~Zhang. Detecting privileged side-channel attacks in shielded execution with Déjá Vu. In {\em ASIA Conference on Computer and Communications Security}, pages 7--18, 2017.

\bibitem{CormodeLDP18}
G.~Cormode, S.~Jha, T.~Kulkarni, N.~Li, D.~Srivastava, and T.~Wang. Privacy at scale: Local differential privacy in practice. In {\em Special Interest Group on Management of Data}, pages 1655--1658, 2018.

\bibitem{LBSAdv}
 T.~Dao, S.~Jeong, and H.~Ahn.
\newblock A novel recommendation model of location-based advertising: Context-Aware Collaborative Filtering using GA approach.
\newblock In {\em Expert Systems with Applications}, 39(3), pages 3731--3739, 2012.


\bibitem{boling17}
B.~Ding, J.~Kulkarni, and S.~Yekhanin. \newblock Collecting telemetry data privately. arXiv preprint arXiv:1712.01524.


\bibitem{Locprivlimitation}
 K.~Dong, T.~Guo, H.~Ye, X.~Li, and Z.~Ling.
\newblock On the limitations of existing notions of location privacy.
\newblock In {\em Future Generation Computer Systems}, pages 1513--1522, 2018.

\bibitem{LBSSearch}
R.~Dewri and T.~Ramakrisha.
\newblock Exploiting service similarity for privacy in location-based search queries.
\newblock In {\em Transactions on Parallel and Distributed Systems}, pages 374--383, 2013.
 
 \bibitem{MI3}
J.~Duchi, M.~Jordan, and M.~Wainwright.
\newblock Local privacy and statistical minimax
rates.
\newblock In {\em Foundations of Computer Science}, pages 429--438, 2013.

\bibitem{trafficpath}
U.~Demiryurek,  F.~Banaei-Kashani, C.~Shahabi, and A.~Ranganathan.
\newblock Online computation of fastest path in time-dependent spatial networks.
\newblock In {\em Symposium on Spatial and Temporal Databases}, pages 91--111, 2011.

\bibitem{Dwork14}
C.~Dwork and A.~Roth.
\newblock The algorithmic foundations
of differential privacy.
\newblock In {\em Foundations and Trends® in Theoretical Computer Science}, 9(3-4), pages 211--407, 2014.


\bibitem{shuffler}
{\'U}. Erlingsson, V. Feldman, I. Mironov, A. Raghunathan, K. Talwar, and A. Thakurta. Amplification by Shuffling: From Local to Central Differential Privacy via Anonymity. In {\em ACM-SIAM Symposium on Discrete Algorithms}, pages 2468--2479, 2019.

\bibitem{rappor14}
{\'U}. Erlingsson, V. Pihur, and A. Korolova.  Rappor: Randomized aggregatable privacy-preserving ordinal response. In {\em Computer and Communications Security}, pages 1054--1067, 2014.


\bibitem{PIR}
G.~Ghinita, P.~Kalnis, A.~Khoshgozaran, C.~Shahabi, and K.~Tan. 
Private queries in location based services: Anonymizers are not necessary. In {\em Special Interest Group on Management of Data}, pages 121--132, 2008.

\bibitem{loccor}
R.~Gideon and A. M.~Rothan.  
Location and scale estimation with correlation coefficients.In {\em Commun. Stat.}, pages 1561--1572, 2011.

\bibitem{liulingtra18}
 M.E.~Gursoy, L.~ Liu, S.~Truex, L.~Yu, and W.~Wei.
\newblock Utility-aware synthesis of differentially private and attack-resilient location traces.
\newblock In {\em Computer and Communications Security}, pages 196--211, 2018.

\bibitem{Ghosh09}
A.~Ghosh, T.~Roughgarden, and M.~Sundararajan. 
\newblock Universally utility-maximizing privacy mechanisms.
\newblock In {\em Symposium on Theory of Computing}, pages 351--360, 2009.

\bibitem{staircase}
 Q.~Geng, P.~Kairouz, S.~Oh, and P.~Viswanath. 
\newblock The staircase mechanism in differential privacy.
\newblock In {\em Journal of Selected Topics in Signal Processing}, pages 1176--1184, 2015.

\bibitem{InputDis2020}
 X.~Gu, M.~Li,  L.~Xiong, and Y. Cao.
\newblock Providing input-discriminative protection for local differential privacy.
\newblock In {\em International Conference on Data Engineering}, pages 505--516, 2020.

\bibitem{CLDP}
 M.E.~Gursoy,  A.~Tamersoy, S.~Truex, W.~Wei, and L.~ Liu.
\newblock Secure and Utility-Aware Data Collection with Condensed Local Differential Privacy.
\newblock In {\em IEEE Transactions on Dependable and Secure Computing}, pages 2365--2378, 2019.

\bibitem{Gupte10}
M.~Gupte and M.~Sundararajan.
\newblock Universally optimal privacy mechanisms for minimax agents.
\newblock In {\em Portable On Demand Storage}, pages 135--146, 2010.

\bibitem{consistency}
M.~Hay, V.~Rastogi, G.~Miklau, and D.~Suciu.
\newblock Boosting the accuracy
of differentially private histograms through consistency.
\newblock In {\em Very Large Data Base Endowment}, 2010.


\bibitem{DPT}
X.~He, G.~Cormode,  A.~Machanavajjhala, C. M.~Procopiuc, and D.~Srivastava.
\newblock DPT: differentially private trajectory synthesis using hierarchical reference systems. 
\newblock In{\em Very Large Data Base Endowment}, pages 1154--1165, 2015.

\bibitem{DP2}
S.S.~Ho and S. Ruan. 
\newblock Differential privacy for location pattern mining. 
\newblock In {\em SPRINGL}, pages 17--24,
2011.

\bibitem{SearchLog}
Y.~Hong, J.~Vaidya, H.~Lu, and M.~Wu. 
\newblock Differentially private search log sanitization with optimal output utility. 
\newblock In {\em Proceedings of the 15th International Conference on Extending Database Technology}, pages 50--61,
2012.

\bibitem{PIR3}
P.~Indyk and D. P. ~Woodruff. 
\newblock Polylogarithmic Private Approximations and Efficient Matching. 
\newblock In {\em Theory of Cryptography Conference}, pages 245--264, 2006.

\bibitem{emp}
A.~Jayadev, Z.~Sun, and H.~Zhang. \newblock Hadamard response: Estimating distributions privately, efficiently, and with little communication. 
\newblock In {\em Proceedings of Machine Learning Research}, pages 1120--1129, 2019.

\bibitem{keahey2020lessons}
K.~Kate and A.~Jason and Z.~Zhuo and R.~Pierre and R.~Paul and S.~Dan and C.~Mert and C.~Jacob and G.~Haryadi S and H.~Cody.
\newblock Lessons Learned from the Chameleon Testbed.
\newblock In Annual Technical Conference, pages 219--233, 2020.

\bibitem{loc_LDP3}
 J. Kim and B. Jang. \newblock Workload-aware indoor positioning data collection via local differential privacy.
\newblock In {\em Communications Letters}, pages 1352--1356, 2019.

\bibitem{EncLBS}
L.~Li, R.~Lu, and C.~Huang. EPLQ: Efficient privacy-preserving location-based query over outsourced encrypted data. {\em IEEE Internet of Things Journal}, pages 206--218, 2016.

\bibitem{POIDP}
C.~Li, B.~Palanisamy, and J.~Joshi.
\newblock Differentially private trajectory analysis for points-of-interest recommendation.
\newblock In {\em BigData Congress}, pages 49--56, 2017.

\bibitem{NumLDP}
Z.~Li, T.~Wang, M.~Lopuhaä-Zwakenberg, N.~Li, and B.~Škoric.
\newblock Estimating numerical distributions under local differential privacy.
\newblock In {\em Special Interest Group on Management of Data}, pages 621--635, 2020.

\bibitem{trajectoryTDSC}
B.~Liu, S.~Xie, H.~Wang, Y.~Hong, X.~Ban and M.~Mohammady.
\newblock VTDP: Privately Sanitizing Fine-grained Vehicle Trajectory Data with Boosted Utility.
\newblock In {\em IEEE Transactions on Dependable and Secure Computing}, pages 2643--2657, 2021.

\bibitem{SA}
Van Laarhoven, P. JM, and E. H. ~Aarts. 
\newblock Simulated annealing. \newblock In {\em Simulated annealing: Theory and applications}, pages 7--15, 1987.

\bibitem{attack}
A. Machanavajjhala, D. Kifer, J. Gehrke, and M. Venkitasubramaniam.
\newblock l-diversity: Privacy beyond k-anonymity. 
\newblock {\em Transactions on Knowledge Discovery from Data}, pages 3--es, 2007.

\bibitem{DP1}
A. Machanavajjhala, D. Kifer, J. M. Abowd, J. Gehrke, and L. Vilhuber.
\newblock Privacy: Theory meets practice on the map. In
\newblock {\em International Conference on Data Engineering}, pages 277--286, 2008.

\bibitem{EvolvingData}
J. Matthew, R. Aaron, U. Jonathan, and W. Bo
\newblock Local differential privacy for evolving data. arXiv preprint arXiv:1802.07128.

 \bibitem{MI2}
 A.~McGregor, I.~Mironov, T.~Pitassi, O.~Reingold, K.~Talwar, and S.~Vadhan.
\newblock The limits of two-party differential privacy.
\newblock In {\em the Foundations of Computer Science}, pages 81--90, 2010.

 \bibitem{R2DP}
 M.~Mohammady, S.~Xie, Y.~Hong, M.~Zhang, L. ~Wang, M.~Pourzandi, and M.~Debbabi.
\newblock R2DP: A universal and automated approach to optimizing the randomization mechanisms of differential privacy for utility metrics with no known optimal distributions.
\newblock In {\em Proceedings of the 2020 ACM SIGSAC Conference on Computer and Communications Security}, pages 677-696, 2020.

\bibitem{ULDP19}
T.~Murakami, and Y.~Kawamoto.
\newblock Utility-optimized local differential privacy mechanisms for distribution estimation.
\newblock In {\em USENIX Security Symposium}, pages 1877--1894, 2019.

\bibitem{cabs}
 L.~Moreira-Matias, J.~Gama, M.~Ferreira, J.~Mendes-Moreira, and L.~Damas.
\newblock Predicting taxi–passenger demand using streaming data.
\newblock In Transactions on Intelligent Transportation Systems, pages 1393--1402, 2013.

\bibitem{OptTrajectory}
L.~Ou, Z.~Qin, S.~Liao, Y.~Hong and X.~Jia.
\newblock Releasing Correlated Trajectories: Towards High Utility and Optimal Differential Privacy.
\newblock In {\em IEEE Transactions on Dependable and Secure Computing}, pages 1109--1123, 2020.

\bibitem{LU2}
C.~Ozcan and B.~Sen.
\newblock Investigation of the performance of LU decomposition method using CUDA.
\newblock In {\em  Procedia Technology}, pages 50--54, 2012.

\bibitem{traden}
C.~Ozkurt and  F.~Camci.
\newblock Automatic traffic density estimation and vehicle classification for traffic surveillance systems using neural networks.
\newblock In {\em  Mathematical and Computational Applications}, pages 187--196, 2009.

\bibitem{cyber}
A.~Philip, C.~Haitham, M.~Jeremy, O.~Chibueze P Anyigor, L.~Ao, H.~Waleed, B.~Shihan, and S.~Zhili. \newblock Security and privacy in location-based services for vehicular and mobile communications: An overview, challenges, and countermeasures.
\newblock In {\em IEEE Internet of Things Journal}, pages: 4778--4802, 2018.

\bibitem{zhan17}
Z. Qin, T. Yu, Y. Yang, I. Khalil, X. Xiao, and K. Ren.  Generating synthetic decentralized social graphs with local differential privacy. In {\em Computer and Communications Security}, pages 425--438, 2017.

\bibitem{LoPub}
X.~Ren, C.~Yu, W.~Yu, S.~Yang, X.~Yang, J.~McCann, and Y.~Philip. \newblock $\textsf {LoPub} $: High-Dimensional Crowdsourced Data Publication with Local Differential Privacy.
\newblock In {\em IEEE Transactions on Information Forensics and Security}, pages 2151--2166, 2018.

\bibitem{croden}
P.~Shankar, V.~Ganapathy, and L.~Iftode. 
Privately querying location-based services with sybilquery. In {\em UbiComp},
pages 31--40, 2009.

\bibitem{route}
S.~Shang, H.~Lu, T. B.~Pedersen, and X.~Xie. Finding traffic-aware fastest paths in spatial networks. In {\em International Symposium on Spatial and Temporal Databases}, pages 128--145, 2013.

\bibitem{LBS}
R.~Shokri, G. ~Theodorakopoulos, J.Y. ~Le Boudec, and J.-
P.~Hubaux. 
Quantifying location privacy. In {\em IEEE Symposium on Security and Privacy}, pages 247--262, 2011.

\bibitem{ldpusenix17}
T. Wang, J. Blocki, N. Li, and S. Jha. Locally differentially private
  protocols for frequency estimation. In {\em USENIX Security Symposium}. pages 729--745, 2017.
  
\bibitem{TianhaoMultiLDP}
T.~Wang, B.~Ding, J.~Zhou, C.~Hong, Z.~Huang, N.~Li, and S.~Jha.
\newblock Answering multi-dimensional analytical queries under local differential privacy.
\newblock In {\em Special Interest Group on Management of Data}, pages 159--176, 2019.

\bibitem{GRR}
S.~Wang, L.~Huang, P.~Wang, H.~Deng, H.~Xu, and W.~Yang. \newblock Private weighted histogram aggregation in crowdsourcing.
\newblock In {\em Wireless Algorithms, Systems, and Applications}, pages 250--261, 2016.

\bibitem{MI}
S.~Wang, L.~Huang, P.~Wang, Y.~Nie, H.~Xu, W.~Yang, X.~Li, and C.~Qiao. \newblock Mutual information optimally local private discrete distribution estimation.
\newblock arXiv:1607.08025, 2016.

\bibitem{PM}
 N.~Wang, X.~Xiao, Y.~Yang, T.D.~Hoang, H.~Shin, J.~Shin, and G.~Yu.
\newblock PrivTrie: Effective frequent term discovery under local differential privacy. \newblock In {\em International Conference on Data Engineering}, pages 821--832, 2018.

\bibitem{RRLDP}
Y.~Wang, X.~Wu, and D.~Hu.
\newblock PrivTrie: Using Randomized Response for Differential Privacy Preserving Data Collection.
\newblock In {\em EDBT/ICDT Workshops}, pages 0090--6778, 2018.

\bibitem{VideoDP}
H.~Wang, S.~Xie, and Y.~Hong.
\newblock VideoDP: A Flexible Platform for Video Analytics with Differential Privacy.
\newblock In {\em Proc. Priv. Enhancing Technol.}, pages 277--296, 2020.

\bibitem{VideoLDP}
H.~Wang, Y.~Kong, Y.~Hong and J.~Vaidya.
\newblock Publishing video data with indistinguishable objects.
\newblock {\em Advances in database technology: proceedings. International Conference on Extending Database Technology.}, pages 323--334, 2020.

\bibitem{HMM}
D.~Woodard, G.~Nogin, P.~Koch, D.~Racz, M.~Goldszmidt, and E.~Horvitz.
\newblock Predicting travel time reliability using mobile phone GPS data.
\newblock In {\em Transportation Research Part C: Emerging Technologies} pages 30--44, 2017.

\bibitem{Four}
D.~Yang, B.~Qu, J.~Yang, and P. ~Cudre-Mauroux.
\newblock Revisiting user mobility and social relationships in lbsns: A hypergraph embedding approach.
\newblock In World Wide Web, pages 2147--2157, 2019.

\bibitem{KNN}
X.~Yi, R.~Paulet, E.~Bertino, and V.~Varadharajan.
\newblock Practical approximate k nearest neighbor queries with location and query privacy.
\newblock In Transactions on Knowledge and Data Engineering, 28(6), pages 1546--1559, 2016.

\bibitem{liuling17CCS}
L.~Yu, L.~Liu, and C.~Pu. Dynamic Differential Location Privacy with Personalized Error Bounds. In {\em Network and Distributed System Security Symposium}, 2017.

\bibitem{loc_LDP2}
X.~Zhao, Y.~Li, Y.~Yuan, X.~Bi, and G.~Wang.
\newblock Ldpart:
Effective location-record data publication via local
differential privacy.
\newblock In {\em IEEE Access}, pages 31435--31445, 2019.

\bibitem{Geolife}
Y.~Zheng, L.~Zhang, X.~Xie, and W.Y.~Ma.
\newblock Mining interesting locations and travel sequences from GPS trajectories.
\newblock In {\em World Wide Web Conference}, pages 791--800, 2009.
\end{thebibliography}
